%% file: main_jmlr.tex
\documentclass[twoside,11pt]{article}

\usepackage{blindtext}

\usepackage[abbrvbib, preprint]{jmlr2e}

\usepackage{lastpage}
\jmlrheading{27}{2026}{1-\pageref{LastPage}}{02/26; Revised -/-}{-/-}{21-0000}{Manish Prajapat, Johannes K\"ohler, Melanie N. Zeilinger, Andreas Krause}

\usepackage[T1]{fontenc}    
\usepackage{hyperref}       
\usepackage{url}            
\usepackage{booktabs}       
\usepackage{amsfonts}       
\usepackage{nicefrac}       
\usepackage{microtype}      
\usepackage{xspace}
\usepackage{siunitx}        

\usepackage{enumitem}
\setlist[itemize]{labelsep=0.6em, left=12pt}
\usepackage{bm}
\usepackage{comment}
\usepackage{amsmath}
\usepackage{graphicx}
\usepackage{rotating}
\usepackage{float}
\usepackage{mathrsfs}
\usepackage{amssymb}
\usepackage{autobreak}
\usepackage{mathtools}
\usepackage{wrapfig}
\usepackage{xcolor}
\usepackage{bbm}

\usepackage[noend]{algpseudocode}
\usepackage{algorithm}

\usepackage{subcaption}
\usepackage{lipsum}
\usepackage{outlines}
\usepackage{blindtext}
\usepackage{multicol}

\usepackage{tikz}
\usetikzlibrary{shapes,backgrounds,patterns}
\usepackage{array}
\usepackage[export]{adjustbox}
\newsavebox{\figright}
\usepackage[capitalise]{cleveref}

\usepackage[normalem]{ulem}   
\usepackage{tcolorbox}
\usepackage{thmtools} 
\usepackage{thm-restate}
\usepackage[toc,page,header]{appendix}
\usepackage{minitoc}
\usepackage{longtable,xtab,booktabs}

\usepackage{makecell}
\usepackage{pifont}  
\newcommand{\xmark}{\ding{55}}  

\crefname{assumption}{Assumption}{Assumptions}

\newsavebox{\algleft}
\newsavebox{\algright}
\newsavebox{\mdpfigright}

\algdef{SE}[DOWHILE]{Do}{doWhile}{\algorithmicdo}[1]{\algorithmicwhile\ #1}%

\DeclareMathOperator*{\argmax}{arg\,max}

\newcommand\numberthis{\addtocounter{equation}{1}\tag{\theequation}}

\input{macros_jmlr}

\newcounter{relctr} 
\everydisplay\expandafter{\the\everydisplay\setcounter{relctr}{0}} 
\newcommand\labelrel[2]{%
  \begingroup
    \refstepcounter{relctr}%
    \stackrel{\textnormal{(\text{\roman{relctr})}}}{\mathstrut{#1}}%
    \originallabel{#2}%
  \endgroup
}
\AtBeginDocument{\let\originallabel\label} 



\makeatletter
\renewcommand{\theHALG@line}{\thealgorithm.\arabic{ALG@line}}
\makeatother

\allowdisplaybreaks
\makeatletter
\def\thanks#1{\protected@xdef\@thanks{\@thanks
        \protect\footnotetext{#1}}}
\makeatother

\usepackage{comment}

\newcommand{\JK}{\textcolor{black}}
\newcommand{\manish}[1]{\textcolor{red}{#1}}
\newcommand{\change}[1]{{\color{black}{#1}}}

\ShortHeadings{Safe and Near-Optimal Control with Online Dynamics Learning}{M. Prajapat, J. K\"ohler, M. N. Zeilinger, A. Krause}
\firstpageno{1}

\begin{document}
\title{Safe and Near-Optimal Control with Online Dynamics Learning}


\author{%
  \name  Manish Prajapat \thanks{$\dagger$ Joint supervision. Code available at \url{https://github.com/manish-pra/SageDynX}}  
  \email manishp@ai.ethz.ch \\
\addr   ETH Zurich 
   \AND
  \name Johannes K\"ohler \email jkoehle@ethz.ch \\
 \addr   ETH Zurich 
   \AND
 \name   Melanie N. Zeilinger\textsuperscript{$\dagger$} \email mzeilinger@ethz.ch \\
\addr    ETH Zurich 
   \AND
 \name   Andreas Krause\textsuperscript{$\dagger$} \email krausea@ethz.ch \\
\addr    ETH Zurich 
}

\editor{My editor}

\maketitle


\input{sections/0_abstract}

\input{sections/1_introduction}
\input{sections/2_problem_setting}
\input{sections/psuedocode_full_explor}
\input{sections/3_safe_dynamics_exploration}

\input{sections/4_guaranteed_exploration_in_task_oriented}
\input{sections/4.1_sample_complexity_lower_bound}
\input{sections/5_experiments}
\input{sections/6_conclusion}




\newpage

\appendix

\input{appx/extended_related_works}
\input{appx/safe_dyn_exploration}

\input{appx/goal_directed_exploration}
\vspace{-2em}
\input{appx/experiments}

\newpage
\bibliography{ref.bib}

\end{document}

%% file: macros_jmlr.tex
\newcommand{\qedhere}{\tag*{\BlackBox}}

\newcommand{\epssafeset}{\epsilon'}
\newcommand{\dynExplor}{\textsc{\small{SageDynX}}\xspace}

\newcommand{\Intp}{\mathbb{N}}
\newcommand{\Intrange}[2]{\mathbb{N}_{[#1,#2]}}
\newcommand{\Liprew}{L'_r}
\newcommand{\LiprewPi}{L_r}
\newcommand{\stateOpti}{\state^{\mathrm{o}}}
\newcommand{\statePessi}[1][]{\state^{\mathrm{p} #1}}

\newcommand{\upPi}{\Pi}
\newcommand{\entropy}{\mathcal H}
\newcommand{\Jobj}[3][]{J^{#1}({{#2}};{#3})}
\newcommand{\h}{h}
\newcommand{\truePolicySet}[2][]{\Pi_{#2}^{ \star #1}}
\newcommand{\statetrue}[1][]{\state^{\star #1}}
\newcommand{\cwidth}[1][]{w_{#1}}

\newcommand{\numdataPts}{d}
\newcommand{\epscollect}{\epsilon_c}
\newcommand{\epsdiff}{\epsilon_c}
\newcommand{\epspessi}{\epsilon'}
\newcommand{\epsanywhere}{\epsilon'}
\newcommand{\epsterminal}{\epsilon_d}
\newcommand{\epsguarantee}{\epsilon}

\newcommand{\epsclose}{\epsilon_u}
\newcommand{\totHorizon}{\Horizon_c}
\newcommand{\seqinput}{\bm{{\coninput}}}
\newcommand{\StatModel}[1][]{\mathcal{M}_{#1}}

\newcommand{\dynVec}{\bm{f}}

\newcommand{\W}{\mathcal W}

\newcommand{\gpinp}{z}
\newcommand{\dyntrue}{f^{\star}}
\newcommand{\dyntrueVec}{\bm{f}^{\star}}
\newcommand{\statedim}{{n_x}}
\newcommand{\inputdim}{{n_u}}

\newcommand{\LipDyn}{L}
\newcommand{\LipPi}{L_{\pi}}
\newcommand{\noise}[1][]{{\eta_{#1}}}
\newcommand{\noisebound}{\Tilde{\eta}}

\newcommand{\coninput}[1][]{{u_{{#1}}}}
\newcommand{\state}[1][]{{x_{{#1}}}}

\newcommand{\datasetdim}{D}

\newcommand{\Horizon}{\mathrm{H}}

\newcommand{\probab}{\delta}
\newcommand{\dynSet}[1][]{\mathcal{F}_{#1}}
\newcommand{\n}{n}

\newcommand{\identity}{I}

\newcommand{\GP}{\mathcal{GP}}

\newcommand{\kernelfunc}{k}

\newcommand{\Bg}{B_i}

\newcommand{\RKHS}[1][]{{\mathcal{H}_{k^{#1}}}}

\newcommand{\probability}[1]{\mathrm{Pr}\left( {#1} \right)}

\newcommand{\gppostmean}[1][]{{\mu_{#1}}}
\newcommand{\gppostvar}[1][]{{\sigma_{#1}}}

\newcommand{\betadata}[1][]{\beta_{#1}}
\newcommand{\ki}{k}

\newcommand{\inputSpace}{\mathcal{U}}

\newcommand{\Horizonmid}{H^\prime}
\newcommand{\timeX}{\Delta \Horizon}
\newcommand{\stateAct}{z}

\newcommand{\xpos}{x_p}
\newcommand{\ypos}{y_p}

\newcommand{\dyn}{f}

\newcommand{\ball}[1]{\mathcal{B}_{#1}}

\newcommand{\mypar}[1]{\textbf{#1}.}

\newcommand{\U}{\mathcal{U}}

\renewcommand{\Re}{\mathbb{R}}

\newcommand{\X}{\mathcal X}

\newcommand{\D}{\mathcal D}

\newcommand{\E}{\mathbb E}

\newcommand{\N}{\mathcal N}

\newcommand{\R}{\mathbb{R}}

\newtheorem{assumption}{Assumption}
\newtheorem{objective}{Objective}

\newcommand{\LocAgents}{X}

\newcommand{\constrain}{q}

\newcommand{\LipWidth}{L'_{\cwidth}}
\newcommand{\LipWidthPi}{L_{\cwidth}}

\newcommand{\Rcontoper}[2][]{\mathcal{R}_{\Horizon}({{#1}},{#2})}
\newcommand{\RTcoper}[3][]{\mathcal{R}_{#3}({{#1}},{#2})}

\newcommand{\sigconst}[1][]{\sigma_{#1}}

\newcommand{\epsconst}{\epsilon}

\newcommand{\noiseconst}{\sigma^{-2}}

\newcommand{\betaconst}[1][]{\beta_{#1}}

\newcommand{\gammaconst}[1]{\gamma_{#1}}

\newcommand{\retPolicy}[1][]{\pi^{\mathrm{r}}_{#1}}
\newcommand{\pessiPolicy}[1][]{\pi^{\mathrm{p}}_{#1}}
\newcommand{\optiPolicy}[1][]{\pi^{\mathrm{o}}_{#1}}

\newcommand{\pessiSet}[2][]{\Pi_{#2}^{\mathrm{p} #1}}

\newcommand{\optiSet}[2][]{\Pi_{#2}^{ \mathrm{o}, \epsconst #1}}
\newcommand{\optiSetwe}[2][]{\Pi_{#2}^{ \mathrm{o} #1}}

\newcommand{\sumMaxwidth}[2][]{w^{#1}_{#2}}

\newcommand{\safeInit}[1]{\mathbb{\LocAgents}_{#1}}

\newcommand{\nfin}{{\bar{\n}}}


%% file: sections/0_abstract.tex
\begin{abstract}
\looseness -1 
Achieving both optimality and safety under unknown system dynamics is a central challenge in 
real-world deployment of agents.
To address this, we introduce a notion of maximum safe dynamics learning, where sufficient exploration is performed within the space of safe policies.
Our method executes \emph{pessimistically} safe policies while \emph{optimistically} exploring informative states and, despite not reaching them due to model uncertainty, ensures continuous online learning of dynamics. 
The framework achieves first-of-its-kind results: learning the dynamics model sufficiently — up to an arbitrary small tolerance (subject to noise) — in a finite time, while ensuring provably safe operation throughout with high probability and without requiring resets. 
Building on this, we propose an algorithm to maximize rewards while learning the dynamics \emph{only to the extent needed} to achieve close-to-optimal performance. 
Unlike typical reinforcement learning (RL) methods, our approach operates online in a non-episodic setting and ensures safety throughout the learning process. 
We demonstrate the effectiveness of our approach in challenging domains such as autonomous car racing and drone navigation under aerodynamic effects — scenarios where safety is critical and accurate modeling is difficult. 

\end{abstract}

%% file: sections/1_introduction.tex
\section{Introduction}
\label{sec:introduction}

\looseness -1
Deploying agents in the real world is inherently challenging due to their \emph{a priori} unknown dynamics and the need for rigorous safety and optimality guarantees.
Without an accurate model, the agent cannot predict the consequences of actions, and thus risks taking unsafe actions or failing to complete their tasks. 
Moreover, real-world deployments are non-episodic, where resetting back to an initial state is often not possible or very costly.
In such cases, the agent must actively interact with the environment and learn the dynamics online from only a single trajectory while remaining safe throughout. 
Successfully addressing this challenge of sufficient learning for optimal behavior in a safe, online, and non-episodic manner is essential across a wide range of applications, especially in autonomous robotics~\citep{wong2018autonomous}, including car racing \citep{kabzan2020amz}, drones, and underwater vehicles~\citep{9062680}.
\Cref{fig:dyn_expl_problem} illustrates the considered safe online dynamics learning framework in a drone navigation example.

\begin{figure}
   \scalebox{0.62}{\input{figure/sagedynx_intro}}
    \caption{\change{Illustration of the online dynamics learning problem.
    A drone navigates a cluttered environment while satisfying safety constraints, despite its dynamics being a-priori unknown. These constraints require avoiding collisions with the orange obstacles and passing through the green gate. 
    The green curve illustrates the optimal trajectory that the drone would have taken if the dynamics were known exactly.
    The dotted curve shows the actual trajectory executed by the drone, which deviates from the optimal path due to model uncertainty.
    The gray region represents the propagated uncertainty during planning at the current location of the drone. Initially, uncertainty is large, leading the drone to plan conservatively. As the drone learns the dynamics online, the propagated uncertainty shrinks, resulting in less conservative plans, and the executed trajectory gets close to optimal.
    }} 
    \label{fig:dyn_expl_problem}
    \vspace{-0.3em}
\end{figure}


\looseness -1 \mypar{Related work}
Model-based reinforcement learning (RL) is a well-established framework to achieve close-to-optimal performance through optimistic or sampling-based exploration~\citep{chua2018deep,kakade2020information,curi2020efficient}.
To ensure safety, constraints under expectation can be additionally enforced through Constrained Markov Decision Processes (CMDPs) \citep{altman2021constrained}; however, these approaches typically ensure safety only for the final learned policy, not throughout the learning process itself~\citep{achiam2017constrained,ding2020natural,as2022constrained,muller2024truly,gu2024review}. 
Safety during learning is often addressed through pessimistic exploration strategies inspired by the safe Bayesian optimization (BO) literature~\citep{safe-bo-sui15,wachi2020safe,turchetta2016safemdp,prajapat2022near}.
These methods have been extended to optimize safe policies~\citep{berkenkamp2023bayesian,as2024actsafe}. 
However, they require resets during the learning process, which are often not feasible in real-world settings. 
Safe exploration {\em without} resets has been developed more recently, using techniques from model predictive control (MPC)~\citep{prajapat2025safe}, i.e., online finite-horizon re-planning. 
However, this method is only applicable to learning uncertain constraints, and steering the system to informative states requires that the dynamics are known exactly.
The problem of safely controlling an uncertain dynamical system without resets has also been studied in the control community, especially using (learning-based) MPC~\citep{hewing2020learning,brunke2021safe,koller_learning_based_2018}. 
However, only a few works consider actively learning dynamics \citep{heirung2018model,lew2022safe,soloperto2023guaranteed}. In general, these algorithms do not guarantee
that the dynamics model is learned sufficiently accurately, 
let alone providing optimality guarantees. 
In contrast, our work {\em uniquely addresses safe guaranteed exploration without resets, ensuring close-to-optimal performance in finite-time}, see also~\cref{apxsec:extended_related_works} for a more detailed comparison to related work.

\mypar{Contributions}
To address the open problem of safe, non-episodic learning of unknown dynamics with sufficient exploration, we make the following key contributions:
\begin{itemize}
\item To ensure sufficient dynamics learning for any task, we introduce the notion of guaranteed exploration in the policy space, which ensures \emph{maximum safe dynamics learning}. We propose a general framework that pessimistically ensures safety against model uncertainty while optimistically exploring informative states. 
Even though the desired state cannot be reached because of unknown dynamics, the proposed approach ensures that the agent obtains informative measurements.
\item We theoretically guarantee that the dynamics model is learned sufficiently up to a user-chosen tolerance (subject to noise) in finite time, while being provably safe throughout with high probability.
\item Building on the general framework, we propose an efficient algorithm to maximize rewards, where the dynamics model is learned only to the extent needed to achieve optimal behavior.
The approach maintains safety by ensuring returnability to a known safe set, but without having to actually return, thus ensuring efficiency. We prove that the agent achieves close-to-optimal performance in finite time,
while being provably safe throughout the (non-episodic) online learning process. 
\item \change{We derive a lower bound on the sample complexity of the dynamics learning problem to provide insight into the optimality of exploration. 
By reducing it to Gaussian Process (GP) bandit optimization~\citep{scarlett2017lower}, we show that for the Squared Exponential kernel, the proposed algorithm matches this lower bound up to logarithmic factors in the state dimension, while a larger gap remains for the Mat\'ern kernel.}
\item We provide an efficient sampling-based implementation of our method and demonstrate its effectiveness with unknown dynamical models in the challenging safety-critical task of autonomous racing and drone navigation under aerodynamic effects. Our experiments demonstrate rapid convergence to optimal policies without compromising safety.
\end{itemize}

%% file: figure/sagedynx_intro.tex
\tikzset{every picture/.style={line width=0.75pt}} 

\begin{tikzpicture}[x=0.75pt,y=0.75pt,yscale=-1,xscale=1]

\draw [color={rgb, 255:red, 65; green, 117; blue, 5 }  ,draw opacity=1 ][line width=2.25]    (933,199) .. controls (941.8,202.2) and (965.8,207.4) .. (984.57,200.29) ;
\draw [line width=2.25]  [dash pattern={on 6.75pt off 4.5pt}]  (927.8,197.8) .. controls (945.4,202.6) and (969.8,206.6) .. (984.57,200.29) ;
\draw  [color={rgb, 255:red, 0; green, 0; blue, 0 }  ,draw opacity=1 ][fill={rgb, 255:red, 155; green, 155; blue, 155 }  ,fill opacity=0.3 ] (210,219) .. controls (192,230) and (142,274) .. (148,279) .. controls (154,284) and (174,284) .. (211,300) .. controls (248,316) and (305,329) .. (307,315) .. controls (309,301) and (284,145) .. (260,174) .. controls (236,203) and (228,208) .. (210,219) -- cycle ;
\draw  [fill={rgb, 255:red, 108; green, 215; blue, 108 }  ,fill opacity=1 ] (864.22,50.94) -- (958.43,139.31) -- (957.78,345.46) -- (863.57,257.09) -- cycle(933.45,162.56) -- (888.9,120.77) -- (888.55,233.84) -- (933.1,275.63) -- cycle ;
\draw (120.29,296.49) node [rotate=-51.4] {\includegraphics[width=46.25pt,height=52.98pt]{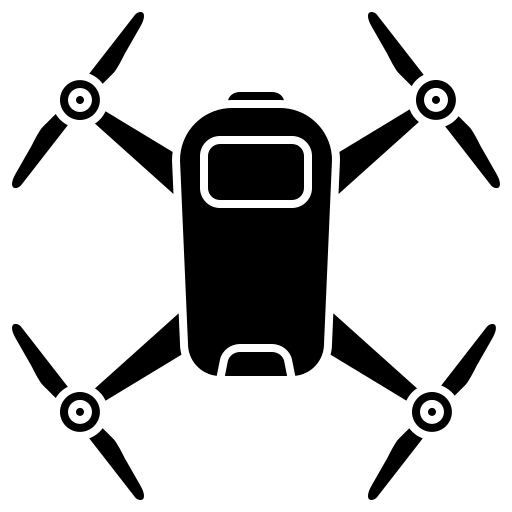}};
\draw [color={rgb, 255:red, 65; green, 117; blue, 5 }  ,draw opacity=1 ][line width=2.25]    (148,279) .. controls (177,248) and (148,274) .. (207,227) .. controls (266,180) and (282,175) .. (330,172) .. controls (378,169) and (440,209) .. (469,218) .. controls (498,227) and (522.67,227.33) .. (542.33,226.33) .. controls (562,225.33) and (587.65,219.51) .. (618.16,214.92) .. controls (648.67,210.33) and (680.83,175.67) .. (710.33,183) .. controls (739.83,190.33) and (772.2,239.8) .. (842.6,214.2) .. controls (913,188.6) and (917,195) .. (933,199) ;
\draw (374.29,196.49) node [rotate=-92.72] {\includegraphics[width=46.25pt,height=52.98pt]{figure/drone.png}};
\draw (596.62,205.49) node [rotate=-113.44] {\includegraphics[width=46.25pt,height=52.98pt]{figure/drone.png}};
\draw (789.69,211.49) node [rotate=-113.44] {\includegraphics[width=46.25pt,height=52.98pt]{figure/drone.png}};
\draw [color={rgb, 255:red, 0; green, 0; blue, 0 }  ,draw opacity=1 ][line width=2.25]  [dash pattern={on 6.75pt off 4.5pt}]  (148,279) .. controls (188,249) and (270,271) .. (292,238) .. controls (314,205) and (309,211) .. (337,201) .. controls (365,191) and (435,202) .. (473,207) .. controls (511,212) and (552,164) .. (583,201) .. controls (614,238) and (665,191) .. (688,191) .. controls (711,191) and (758,200) .. (777,206) .. controls (796,212) and (795,222) .. (817,222) .. controls (839,222) and (895.8,183.4) .. (927.8,197.8) ;
\draw  [fill={rgb, 255:red, 245; green, 166; blue, 35 }  ,fill opacity=0.7 ] (470,253.9) -- (491.9,232) -- (543,232) -- (543,284.1) -- (521.1,306) -- (470,306) -- cycle ; \draw   (543,232) -- (521.1,253.9) -- (470,253.9) ; \draw   (521.1,253.9) -- (521.1,306) ;
\draw  [color={rgb, 255:red, 0; green, 0; blue, 0 }  ,draw opacity=1 ][fill={rgb, 255:red, 245; green, 166; blue, 35 }  ,fill opacity=0.7 ] (215,159) -- (215,211) .. controls (215,215.97) and (201.57,220) .. (185,220) .. controls (168.43,220) and (155,215.97) .. (155,211) -- (155,159) .. controls (155,154.03) and (168.43,150) .. (185,150) .. controls (201.57,150) and (215,154.03) .. (215,159) .. controls (215,163.97) and (201.57,168) .. (185,168) .. controls (168.43,168) and (155,163.97) .. (155,159) ;
\draw  [fill={rgb, 255:red, 245; green, 166; blue, 35 }  ,fill opacity=0.7 ] (177,338.5) -- (193.5,322) -- (272,322) -- (272,360.5) -- (255.5,377) -- (177,377) -- cycle ; \draw   (272,322) -- (255.5,338.5) -- (177,338.5) ; \draw   (255.5,338.5) -- (255.5,377) ;
\draw  [fill={rgb, 255:red, 245; green, 166; blue, 35 }  ,fill opacity=0.7 ] (689,126.9) -- (710.9,105) -- (762,105) -- (762,157.1) -- (740.1,179) -- (689,179) -- cycle ; \draw   (762,105) -- (740.1,126.9) -- (689,126.9) ; \draw   (740.1,126.9) -- (740.1,179) ;
\draw  [color={rgb, 255:red, 0; green, 0; blue, 0 }  ,draw opacity=1 ][fill={rgb, 255:red, 245; green, 166; blue, 35 }  ,fill opacity=0.7 ] (475,129) -- (475,181) .. controls (475,185.97) and (461.57,190) .. (445,190) .. controls (428.43,190) and (415,185.97) .. (415,181) -- (415,129) .. controls (415,124.03) and (428.43,120) .. (445,120) .. controls (461.57,120) and (475,124.03) .. (475,129) .. controls (475,133.97) and (461.57,138) .. (445,138) .. controls (428.43,138) and (415,133.97) .. (415,129) ;
\draw  [color={rgb, 255:red, 0; green, 0; blue, 0 }  ,draw opacity=1 ][fill={rgb, 255:red, 245; green, 166; blue, 35 }  ,fill opacity=0.7 ] (722.82,209.58) -- (722.82,261.58) .. controls (722.82,266.56) and (709.39,270.58) .. (692.82,270.58) .. controls (676.25,270.58) and (662.82,266.56) .. (662.82,261.58) -- (662.82,209.58) .. controls (662.82,204.61) and (676.25,200.58) .. (692.82,200.58) .. controls (709.39,200.58) and (722.82,204.61) .. (722.82,209.58) .. controls (722.82,214.56) and (709.39,218.58) .. (692.82,218.58) .. controls (676.25,218.58) and (662.82,214.56) .. (662.82,209.58) ;
\draw  [color={rgb, 255:red, 0; green, 0; blue, 0 }  ,draw opacity=1 ][fill={rgb, 255:red, 155; green, 155; blue, 155 }  ,fill opacity=0.3 ] (461.1,190.07) .. controls (430.2,196.15) and (390.01,194.93) .. (396.7,198.96) .. controls (403.39,202.99) and (435.38,215.19) .. (459,222) .. controls (482.62,228.81) and (527.15,230.14) .. (527,216) .. controls (526.85,201.86) and (531,177) .. (515,176) .. controls (499,175) and (492,184) .. (461.1,190.07) -- cycle ;
\draw  [color={rgb, 255:red, 0; green, 0; blue, 0 }  ,draw opacity=1 ][fill={rgb, 255:red, 155; green, 155; blue, 155 }  ,fill opacity=0.3 ] (656.78,195.15) .. controls (628.56,218.61) and (609.58,212.73) .. (618.16,214.92) .. controls (626.74,217.11) and (660.19,206.21) .. (675.97,199.67) .. controls (691.75,193.14) and (716.28,197.49) .. (719.5,203.21) .. controls (722.73,208.92) and (724.06,191.38) .. (721.96,186.9) .. controls (719.86,182.41) and (685.01,171.68) .. (656.78,195.15) -- cycle ;
\draw  [color={rgb, 255:red, 0; green, 0; blue, 0 }  ,draw opacity=1 ][fill={rgb, 255:red, 155; green, 155; blue, 155 }  ,fill opacity=0.3 ] (835.4,213.4) .. controls (811.4,222.2) and (798.6,216.2) .. (807.8,220.6) .. controls (817,225) and (839,220.6) .. (854.2,215) .. controls (869.4,209.4) and (904.6,199.4) .. (910.6,200.6) .. controls (916.6,201.8) and (918.2,192.2) .. (912.2,189.4) .. controls (906.2,186.6) and (859.4,204.6) .. (835.4,213.4) -- cycle ;

\end{tikzpicture}

%% file: sections/2_problem_setting.tex
\looseness -1 
\textit{Notation:} The set of non-negative integers is denoted by $\Intp$. We use the short-hand notation $\Intrange{i}{j}$ for a set of integers $\{ i,i+1, \hdots, j\}$. 
We denote $f(x) = \mathcal{O}(g(x))$ as $x \to 0$ if $\exists x_0 > 0, c>0 : \forall~0 < x < x_0, f(x) \leq c g(x)$. Similarly, we write $f(x) = \Omega(g(x))$ as $x \to 0$ if $\exists x_0 > 0, c>0 : \forall~0 < x < x_0, f(x) \geq c g(x)$.
The Minkowski sum and Pontryagin set difference for two sets $X, Y \subseteq \R^n$ are defined by $X \oplus Y \coloneqq \{ x+y \in \R^\n | x \in X, y \in Y \}$ and $X \ominus Y \coloneqq \{ z \in \R^\n | z+y \in X~ \forall y \in Y \}$, respectively. 
All norms, balls of radius $r$, and Lipschitz constants of functions are based on the infinity norm $\|x\|_\infty = \max_{i} |x_i|$ with $x\in\mathbb{R}^n$. A sequence $\{\noise(\ki)\}_{\ki \geq 0}$ is conditionally $\sigma$-sub-Gaussian for a fixed $\sigma \geq0$ if $\forall \ki \in \Intp, \forall s \in \R, \E[e^{s \noise(\ki)}|\{\noise(i)\}_{i=1}^{\ki-1},\{\state(i)\}_{i=1}^{\ki}] \leq \mathrm{exp}(\frac{s^2 \sigma^2}{2})$, see~\citet{beta_chowdhury17a,ao2025stochastic}. 
We assume that for all optimization problems a maximizer can be attained and hence write $\max$ instead of $\sup$.  
\section{Problem setup}
\looseness -1
We consider the task of learning and controlling a nonlinear dynamical system 
\begin{align}
\state(\ki+1) = \dyntrueVec(\state(\ki), \coninput(\ki)) 
+  \noise(\ki)  \label{eq:system_dyn}
\end{align}
with state $\state(\ki) \in \R^\statedim$ and input $\coninput(\ki) \in \R^\inputdim$. The noise $\noise(\ki) \in \W \subseteq \R^\statedim$ is assumed to be bounded and conditionally $\sigma$-sub-Gaussian~\citep{beta_chowdhury17a}. 
Here, \mbox{$\dyntrueVec: \Re^{\statedim} \times \Re^{\inputdim} \rightarrow \Re^{\statedim}$} denotes the dynamics model, which is a-priori unknown and needs to be learned using online state measurements $\state(\ki)$. The disturbance set is given by $\W \coloneqq \ball{\noisebound}$, where $\ball{r}$ is an infinity-norm ball of radius $r$ with appropriate dimension.
The system should satisfy (known) state and input constraints during closed-loop operation:
\begin{align}\label{eq:constraint_set}
    \state(\ki) \in \X, \coninput(\ki) \in \inputSpace,~\forall \ki \in \Intp. \vspace{-0.1em}
\end{align}

\looseness -1
\mypar{Objective}  
First, we consider the problem of
learning the unknown system~\eqref{eq:system_dyn} up to user-chosen tolerance $\epsilon>0$, while satisfying the constraints~\eqref{eq:constraint_set} with arbitrarily high-probability \mbox{$1-\delta\in(0,1)$} during runtime, i.e., in a non-episodic fashion without resets (\cref{sec:full_expl_theory}). In our setting, the agent plans online at time step $\ki \in \Intp$ for a finite horizon $\Horizon$ with $\state_{\h+1} = \dynVec(\state_\h, \coninput_\h) +\noise_h$ denoting $(\h + 1)^{th}$ step prediction in the future starting from $\state_0 = \state(\ki)$. Using the guaranteed exploration framework developed in \cref{sec:full_expl_theory}, 
our goal is 
to safely maximize rewards, i.e., to find a policy $\pi$ that maximizes,
\begin{align}
    \Jobj[]{\state(\ki), \dynVec}{\pi} \coloneqq \E \left[ \sum\nolimits_{\h = 0}^{\Horizon-1} r(\state_\h, \coninput_\h) | \state_0 = \state(\ki), \pi \right]\!\!, \label{eq:obj_def}
\end{align}
over a finite horizon $\Horizon$ for the (unknown) true dynamics $\dynVec=\dyntrueVec$~\eqref{eq:system_dyn}, with $\coninput_\h = \pi_\h(\state_\h)$ while ensuring safety~\eqref{eq:constraint_set} throughout the exploration process 
(\cref{sec:reward_maxim_exploration}). We operate in MPC-style, that is, at real time $k$, we optimize (plan) for a prediction of finite horizon $\Horizon$, then execute actions and repeat planning (see \cref{fig:dyn_expl_problem}).
Without loss of generality, we assume that the reward \mbox{$r\!:\!\R^{\statedim}\!\!\times\!\R^\inputdim \!\!\to\! \R$} is non-negative, i.e., $r: \R \times \R \to \R_{\geq0}$. 
We optimize over the class of non-stationary deterministic policies $\pi \coloneqq [\pi_0, \hdots \pi_{\Horizon-1}]^\top \!\in \Pi_{\Horizon}$ with $\pi_\h : \X \to \inputSpace, \h \in \Intp$.
\subsection{Probabilistic dynamic model}
We make a standard regularity assumption on the unknown dynamics $\dyntrueVec$~\citep{abbasi2013online,beta_srinivas,as2024actsafe,prajapat2022near}:
\begin{assumption}[Regularity]
    \label{assump:q_RKHS}
    \looseness -1  
    Each component~$\dyntrue_i, i \in \Intrange{1}{\statedim}$ of the unknown dynamics model is an element of the Reproducing Kernel Hilbert Space~(RKHS)~$\RKHS[i]$ associated with a 
    continuous, positive definite kernel \mbox{$\kernelfunc^i : \Re^{\statedim + \inputdim} \times \Re^{\statedim + \inputdim} \rightarrow \Re_{\geq 0}$}, 
    with a bounded RKHS norm, i.e., \mbox{$\dyntrue_i \in \RKHS[i]$} with \mbox{$\|\dyntrue_i\|_{\RKHS[i]} \leq \Bg < \infty$} where $\Bg$ is known. 
    Furthermore, $k^i(z,z)\leq 1$, $\forall z\in\mathbb{R}^{n_x+n_u}$.
\end{assumption}

\looseness -1
We collect data online as the policy is executed, and construct a growing dataset $\D_\n = \{Z, Y\}$ with $D_\n$ data points, where 
\mbox{$Z\! =\! [ z_1, \ldots,\!z_{\datasetdim_\n} ]$} with \mbox{$z_\ki \doteq (\state_\ki, \coninput_\ki) \in \Re^{\statedim + \inputdim}$}. 
The corresponding noisy measurements are \mbox{$Y = [ y_1, \ldots, y_{\datasetdim_\n} ]$}, where each \mbox{$y_\ki = \dyntrueVec(z_\ki) + \noise_\ki$}. 
We use Gaussian process (GP) regression to compute the posterior mean $\gppostmean[\n,i](\gpinp)$ and the covariance $\gppostvar[\n,i](\gpinp)$ for each $i\in\Intrange{1}{\statedim}$ using noise standard deviation $\sigma$~\citep{beta_chowdhury17a,gp-Rasmussen}.
Here, $\n$ denotes the number of model updates, and each update conditions on all the data $\D_\n$ collected so far. 
Based on this, we build a dynamics set \vspace{-0.2em}
\begin{align}
\dynSet[\n] \coloneqq \left\{ \dynVec 
~|~|\dyn_i(\gpinp) - \gppostmean[\n',i](\gpinp)| \leq \sqrt{\betadata[\n',i]} \gppostvar[\n',i](\gpinp)~\forall z \in \X \times \inputSpace, \n' \in \Intrange{0}{\n}, i \in \Intrange{1}{\statedim}\right\},\label{eq:dyn_set}
\vspace{-0.3em}
\end{align}
with a scaling factor $\betadata[\n,i]>0$ and define the confidence width \mbox{$\cwidth[\n,i](\stateAct) \!\coloneqq\! 2\sqrt{\betadata[\n,i]} \gppostvar[\n,i](\gpinp)$} quantifying the uncertainty of the dynamics set $\dynSet[\n]$.
We use \mbox{$w_n(x,u)\coloneqq\max_{i} w_{n,i}(x,u)$} and $\betadata[\n] \coloneqq \max_{i} \beta_{n,i}$ to denote the maximum across all state components.
The following lemma provides sufficient conditions to obtain well-calibrated bounds $\dynSet[\n]$ when modeling the unknown dynamics with
GPs\footnote{Notably, our approach equally applies to other probabilistic models, as long as they provide an error bound akin to \cref{lem:beta} for safety and a condition comparable to \cref{assump:sublinear} to ensure finite-time termination.}. 
\begin{lemma}[Well-calibrated model {\cite[\!Theorem 3.11]{abbasi2013online}}]
    \label{lem:beta}  \looseness -1 
    Let \cref{assump:q_RKHS} hold and $\sqrt{\betadata[\n, i]} \!\coloneqq\! \Bg \!+\! \sqrt{ \ln(\det(\identity_{D_\n} \!+\! \sigma^{-2} K^{i}_{\smash{D}_\n})) + 2\ln(\statedim/\probab)}$.
    Then, 
    it holds that $\probability{\dyntrueVec\in\dynSet[\n], \forall \n\in \Intp} \geq 1-\delta$.
\end{lemma}
\cref{lem:beta} follows from applying the scalar bound from  \citep[Theorem 3.11]{abbasi2013online} individually for each of the components and then using Boole's inequality, resulting in factor $\delta/n_x$ in $\sqrt{\betadata[\n]}$. 
%
\begin{assumption}[Lipschitz continuity]\label{assump:lipschitz} 
The dynamics model $\dyntrueVec$ is $L_{\mathrm{f}}$-Lipschitz, 
$\cwidth[\n]$ is $\LipWidth$-Lipschitz $\forall \n \in \Intp$, 
the reward $r$ is $\Liprew$-Lipschitz, 
and any policy $\pi\in\Pi_\Horizon$ is $L_\pi$ Lipschitz.
\end{assumption}
Under \cref{assump:q_RKHS}, the dynamics $\dyntrueVec$ is Lipschitz continuous if the kernels $k^i$ are Lipschitz continuous, such as the squared exponential or the Mat\'ern kernels ~\citep{fiedler2023lipschitz}. Similarly, 
Lipschitz continuity of the confidence width $\cwidth[\n]$ can be derived~\citep{curi2020efficient}.  
Moreover, \cref{assump:lipschitz} also implies that closed-loop dynamics $\dyntrueVec(x,\pi_h(x))$ is $\LipDyn$-Lipschitz continuous with some $\LipDyn\leq L_{\mathrm{f}}(1+L_\pi)$ and similarly $\cwidth[\n](x,\pi(x))$ is $\LipWidthPi$-Lipschitz $\forall \n \in \Intp$ and
the reward $r(x,\pi(x))$ is $\LiprewPi$-Lipschitz continuous.  For short notation we define $L_h \coloneqq \sum_{i=0}^{\h-1} L^i$ and $L_{w,h} \coloneqq \sum_{i=0}^{\h-1} (\LipDyn + \LipWidthPi)^i$.

%% file: sections/psuedocode_full_explor.tex
\savebox{\algleft}{
\begin{minipage}[t]{.60\textwidth}
\vspace{-9 em}
\begin{algorithm}[H]
\caption{Safe guaranteed dynamics exploration}
\begin{algorithmic}[1]
\State \textbf{Initialize:} Start at $\state(0) \in \safeInit{0}$, $\dynSet[0]$, $\totHorizon$, Tol. $\epsilon$, Data $\D_0$
\For{$\n = 0,1, \hdots $} \label{alg: termination-condi}
\State $\pi^p \leftarrow$ Solve Problem~\eqref{eq:sampling_rule_timex} with current state $\state(k)$.\label{lin:solve_sampling_problem}
 \If{Problem~\eqref{eq:sampling_rule_timex} is infeasible} terminate \label{lin:full_exploration_terminate} \EndIf 
\State $\state(\ki) \leftarrow$ Apply $\pi^p$ to $\dyntrueVec$ for $\totHorizon$ steps and \\ \qquad \qquad collect $\mathcal{D}_c \coloneqq \{ (\state_{\h+1}, \stateAct_{\h}) | \cwidth[\n-1](\stateAct_\h)\geq \epscollect\}$. \label{lin:execute_pessi_policy}
\State Update $\dynSet[\n]$ model with $\mathcal{D}_{\n+1} \leftarrow \mathcal{D}_\n \cup \mathcal{D}_c$. 
\EndFor
\end{algorithmic}
\label{alg:full_domain_exploration_basic}
\end{algorithm}
\end{minipage}}

\savebox{\mdpfigright}{
\begin{minipage}[t]{.35\textwidth}
\centering
    \scalebox{0.45}{\input{figure/dyn_traj_diff}} 
\end{minipage}     }

%% file: figure/dyn_traj_diff.tex
  
\tikzset {_4g4mb9fal/.code = {\pgfsetadditionalshadetransform{ \pgftransformshift{\pgfpoint{9 bp } { 0 bp }  }  \pgftransformrotate{0 }  \pgftransformscale{6 }  }}}
\pgfdeclarehorizontalshading{_la6g34qx3}{150bp}{rgb(0bp)=(1,1,1);
rgb(37.5bp)=(1,1,1);
rgb(62.5bp)=(0,0,0);
rgb(100bp)=(0,0,0)}
\tikzset{_d2wjdkpvv/.code = {\pgfsetadditionalshadetransform{\pgftransformshift{\pgfpoint{9 bp } { 0 bp }  }  \pgftransformrotate{0 }  \pgftransformscale{6 } }}}
\pgfdeclarehorizontalshading{_9tfu6ysn1} {150bp} {color(0bp)=(transparent!90);
color(37.5bp)=(transparent!90);
color(62.5bp)=(transparent!10);
color(100bp)=(transparent!10) } 
\pgfdeclarefading{_6milg03lk}{\tikz \fill[shading=_9tfu6ysn1,_d2wjdkpvv] (0,0) rectangle (50bp,50bp); } 
\tikzset{every picture/.style={line width=0.75pt}} 

\begin{tikzpicture}[x=0.75pt,y=0.75pt,yscale=-1,xscale=1]

\path  [shading=_la6g34qx3,_4g4mb9fal,path fading= _6milg03lk ,fading transform={xshift=2}] (401.9,377.29) -- (264.72,376.76) -- (245.28,357.18) -- (245.81,220) -- (314.67,151.67) -- (314.67,151.67) -- (470.76,308.96) -- (470.76,308.96) -- cycle ; 
 \draw   (401.9,377.29) -- (264.72,376.76) -- (245.28,357.18) -- (245.81,220) -- (314.67,151.67) -- (314.67,151.67) -- (470.76,308.96) -- (470.76,308.96) -- cycle ; 

\draw [color={rgb, 255:red, 0; green, 0; blue, 0 }  ,draw opacity=0.52 ][line width=2.25]    (346.67,194) .. controls (351.83,266.33) and (403.83,283) .. (436.5,292.83) ;
\draw [line width=2.25]    (347,168.67) .. controls (373.67,124) and (391.75,151.17) .. (409.13,182.83) .. controls (426.5,214.5) and (438.5,261.33) .. (443.67,277.33) ;
\draw  [fill={rgb, 255:red, 0; green, 0; blue, 0 }  ,fill opacity=1 ][line width=3.75]  (419.7,261.85) .. controls (419.7,261.14) and (420.28,260.57) .. (420.98,260.57) .. controls (421.69,260.57) and (422.26,261.14) .. (422.26,261.85) .. controls (422.26,262.55) and (421.69,263.12) .. (420.98,263.12) .. controls (420.28,263.12) and (419.7,262.55) .. (419.7,261.85) -- cycle ;
\draw [color={rgb, 255:red, 0; green, 119; blue, 255 }  ,draw opacity=1 ][line width=1.5]  [dash pattern={on 5.63pt off 4.5pt}]  (253.7,366.85) .. controls (294.39,372.67) and (297.05,359.33) .. (305.05,337.33) .. controls (313.05,315.33) and (353.05,300.67) .. (371.05,311.33) .. controls (388.42,321.63) and (407.04,299.02) .. (419.7,265.67) ;
\draw [shift={(421.05,262)}, rotate = 109.72] [fill={rgb, 255:red, 0; green, 119; blue, 255 }  ,fill opacity=1 ][line width=0.08]  [draw opacity=0] (13.4,-6.43) -- (0,0) -- (13.4,6.44) -- (8.9,0) -- cycle    ;
\draw [shift={(293,361.35)}, rotate = 143.42] [fill={rgb, 255:red, 0; green, 119; blue, 255 }  ,fill opacity=1 ][line width=0.08]  [draw opacity=0] (13.4,-6.43) -- (0,0) -- (13.4,6.44) -- (8.9,0) -- cycle    ;
\draw [shift={(339.92,310.19)}, rotate = 161.94] [fill={rgb, 255:red, 0; green, 119; blue, 255 }  ,fill opacity=1 ][line width=0.08]  [draw opacity=0] (13.4,-6.43) -- (0,0) -- (13.4,6.44) -- (8.9,0) -- cycle    ;
\draw [shift={(406.42,292.93)}, rotate = 124.07] [fill={rgb, 255:red, 0; green, 119; blue, 255 }  ,fill opacity=1 ][line width=0.08]  [draw opacity=0] (13.4,-6.43) -- (0,0) -- (13.4,6.44) -- (8.9,0) -- cycle    ;
\draw  [fill={rgb, 255:red, 0; green, 0; blue, 0 }  ,fill opacity=1 ][line width=3.75]  (253.7,366.85) .. controls (253.7,366.14) and (254.28,365.57) .. (254.98,365.57) .. controls (255.69,365.57) and (256.26,366.14) .. (256.26,366.85) .. controls (256.26,367.55) and (255.69,368.12) .. (254.98,368.12) .. controls (254.28,368.12) and (253.7,367.55) .. (253.7,366.85) -- cycle ;
\draw [color={rgb, 255:red, 0; green, 0; blue, 0 }  ,draw opacity=1 ][line width=1.5]  [dash pattern={on 5.63pt off 4.5pt}]  (254.98,365.57) .. controls (294.11,352.99) and (290.67,339.84) .. (288.23,316.55) .. controls (285.8,293.27) and (315.34,262.58) .. (336.19,264.29) .. controls (356.32,265.94) and (366.91,244.89) .. (371.82,212.26) ;
\draw [shift={(372.33,208.67)}, rotate = 97.67] [fill={rgb, 255:red, 0; green, 0; blue, 0 }  ,fill opacity=1 ][line width=0.08]  [draw opacity=0] (13.4,-6.43) -- (0,0) -- (13.4,6.44) -- (8.9,0) -- cycle    ;
\draw [shift={(287.9,343.42)}, rotate = 117.46] [fill={rgb, 255:red, 0; green, 0; blue, 0 }  ,fill opacity=1 ][line width=0.08]  [draw opacity=0] (13.4,-6.43) -- (0,0) -- (13.4,6.44) -- (8.9,0) -- cycle    ;
\draw [shift={(307.71,276.89)}, rotate = 135.98] [fill={rgb, 255:red, 0; green, 0; blue, 0 }  ,fill opacity=1 ][line width=0.08]  [draw opacity=0] (13.4,-6.43) -- (0,0) -- (13.4,6.44) -- (8.9,0) -- cycle    ;
\draw [shift={(365.37,239.3)}, rotate = 111.87] [fill={rgb, 255:red, 0; green, 0; blue, 0 }  ,fill opacity=1 ][line width=0.08]  [draw opacity=0] (13.4,-6.43) -- (0,0) -- (13.4,6.44) -- (8.9,0) -- cycle    ;
\draw  [fill={rgb, 255:red, 0; green, 0; blue, 0 }  ,fill opacity=1 ][line width=3.75]  (370.7,210.85) .. controls (370.7,210.14) and (371.28,209.57) .. (371.98,209.57) .. controls (372.69,209.57) and (373.26,210.14) .. (373.26,210.85) .. controls (373.26,211.55) and (372.69,212.12) .. (371.98,212.12) .. controls (371.28,212.12) and (370.7,211.55) .. (370.7,210.85) -- cycle ;
\draw [line width=1.5]  [dash pattern={on 5.63pt off 4.5pt}]  (375.66,212.13) -- (420.21,259.66) ;
\draw [shift={(422.26,261.85)}, rotate = 226.85] [color={rgb, 255:red, 0; green, 0; blue, 0 }  ][line width=1.5]    (14.21,-4.28) .. controls (9.04,-1.82) and (4.3,-0.39) .. (0,0) .. controls (4.3,0.39) and (9.04,1.82) .. (14.21,4.28)   ;
\draw [shift={(373.61,209.95)}, rotate = 46.85] [color={rgb, 255:red, 0; green, 0; blue, 0 }  ][line width=1.5]    (14.21,-4.28) .. controls (9.04,-1.82) and (4.3,-0.39) .. (0,0) .. controls (4.3,0.39) and (9.04,1.82) .. (14.21,4.28)   ;
\draw [line width=1.5]  [dash pattern={on 5.63pt off 4.5pt}]  (372.33,254.33) -- (372.17,147.67) ;
\draw [shift={(372.17,144.67)}, rotate = 89.92] [color={rgb, 255:red, 0; green, 0; blue, 0 }  ][line width=1.5]    (14.21,-4.28) .. controls (9.04,-1.82) and (4.3,-0.39) .. (0,0) .. controls (4.3,0.39) and (9.04,1.82) .. (14.21,4.28)   ;
\draw [shift={(372.33,257.33)}, rotate = 269.92] [color={rgb, 255:red, 0; green, 0; blue, 0 }  ][line width=1.5]    (14.21,-4.28) .. controls (9.04,-1.82) and (4.3,-0.39) .. (0,0) .. controls (4.3,0.39) and (9.04,1.82) .. (14.21,4.28)   ;
\draw [line width=1.5]  [dash pattern={on 5.63pt off 4.5pt}]  (422.33,286.67) -- (422.33,219) ;
\draw [shift={(422.33,216)}, rotate = 90] [color={rgb, 255:red, 0; green, 0; blue, 0 }  ][line width=1.5]    (14.21,-4.28) .. controls (9.04,-1.82) and (4.3,-0.39) .. (0,0) .. controls (4.3,0.39) and (9.04,1.82) .. (14.21,4.28)   ;
\draw [shift={(422.33,289.67)}, rotate = 270] [color={rgb, 255:red, 0; green, 0; blue, 0 }  ][line width=1.5]    (14.21,-4.28) .. controls (9.04,-1.82) and (4.3,-0.39) .. (0,0) .. controls (4.3,0.39) and (9.04,1.82) .. (14.21,4.28)   ;

\draw (273.6,246.17) node [anchor=north west][inner sep=0.75pt]  [font=\LARGE]  {$\dynVec^s$};
\draw (328.26,311.5) node [anchor=north west][inner sep=0.75pt]  [font=\LARGE,color={rgb, 255:red, 0; green, 119; blue, 255 }  ,opacity=1 ]  {$\dyntrueVec$};
\draw (377.43,226.67) node [anchor=north west][inner sep=0.75pt]  [font=\LARGE]  {$d_{\h}{}$};
\draw (372.26,170) node [anchor=north west][inner sep=0.75pt]  [font=\LARGE]  {$\epsterminal$};
\draw (401,227) node [anchor=north west][inner sep=0.75pt]  [font=\LARGE]  {$\epscollect$};

\end{tikzpicture}

%% file: sections/3_safe_dynamics_exploration.tex
\section{Safe guaranteed dynamics exploration}

\label{sec:full_expl_theory}
This section provides the general framework to ensure \textit{maximum safe dynamics learning}, while reward maximization is addressed later in Section~\ref{sec:reward_maxim_exploration}.

\looseness-1
\mypar{Main Idea} 
We execute a policy that pessimistically ensures safety for all dynamics $\dynVec \in \dynSet[\n]$, while optimistically planning to visit informative states with some dynamics $\dynVec^s\in\dynSet[\n]$. 
If the resulting trajectory of the true system $\dyntrueVec$ deviates significantly from the trajectory predicted with $\dynVec^s$, then we \emph{gain information} about the unknown dynamics by observing this discrepancy. Otherwise, if they are close enough, then 
we still \emph{gain information} by reaching the intended informative states. 
In a nutshell; pessimism ensures safety, and no optimistic plan is bad, since we gain information either way.


\begin{figure}
\begin{subfigure}[b]{0.45\columnwidth}
    \centering
    \scalebox{0.55}{\input{figure/dyn_opti_pessi_control_set}} \vspace{-0.7em}
    \caption{\looseness -1 Optimistic and pessimistic trajectories}
    \label{fig:dyn_opti_pessi_definition} 
\end{subfigure}
~
\begin{subfigure}[b]{0.45\columnwidth}
    \centering
    \scalebox{0.55}{\input{figure/dyn_exploration_convergence}} 
    \caption{\looseness -1 Convergence of pessimistic policy set} 
    \label{fig:dyn_exploration_convergence} 
\end{subfigure}
\caption{ 
\looseness -1 
Illustration of policy set in (a) state space and (b) policy space. In \cref{fig:dyn_opti_pessi_definition}, the cyan region denotes the (invariant) safe set $\safeInit{\n}$ and the green region represents the state constraint $\X$. The shaded region shows the reachable set under a pessimistic policy, which starts in the safe set and returns to it while satisfying the constraints. The green curve shows an informative trajectory ensuring sampling condition~\eqref{eq:sampling_rule_timex}. 
The orange curve shows a trajectory under another optimistic policy that ensures constraints are satisfied with an $\epsguarantee$-margin and is appended in the beginning by a small horizon $\delta \h$ to move from $\state(k) \to \state'$ via policy $\hat{\pi}$. 
\cref{fig:dyn_exploration_convergence} shows \cref{obj:maximum_exploration}, where due to exploration the pessimistic policy set starting from $\pessiSet[]{0}$ expands to $\pessiSet[]{\nfin}$, and covers the connected true policy set $\truePolicySet[,\epsguarantee]{c}$. Note that $\truePolicySet[,\epsguarantee]{c}$ is a subset of $\truePolicySet[,\epsguarantee]{}$, which is, in general, disconnected and thus cannot be discovered by executing only safe policies.} 
\label{fig:exploration_process}
\end{figure}
\subsection{\!Maximum safe dynamics learning via guaranteed exploration in policy space}

\looseness -1 
The exploration process is illustrated in~\cref{fig:exploration_process}. We start in some (invariant) safe set $\safeInit{n} \subseteq \X$, where the dependence on $\n$ highlights that the set can also expand with online measurements. 
We optimize a plan that satisfies the constraints and ends again in the safe set for all possible dynamics $\dynVec\in\dynSet[n]$. 
In addition, one of the trajectories in this set also (optimistically) reaches an informative state. 
Intuitively, we achieve \textit{maximum safe dynamics learning} when the set of pessimistically safe policies includes all policies that are safe (with some tolerance $\epsilon$) for the true dynamics. 
To formalize this, we first define the (unknown) true $\epsilon$-safe policy set: 
\begin{multline} \label{eq:true_policy_set}
    \truePolicySet[,\epsilon]{\n}(X; \Horizon) \coloneqq \{ \pi \in \Pi_{\Horizon}~|~\exists \state_0 \in X: \state_{\h+1} = \dyntrueVec(\state_\h, \coninput_\h) + \noise_\h, ~\forall \noise_\h \in \W, \\
    \state_\h \in \X\ominus\ball{\epsguarantee},  
    \coninput_{\h} \coloneqq \pi_\h(\state_\h) \in \inputSpace \ominus\ball{\epsguarantee}, \forall \h \in \Intrange{0}{\Horizon-1}, 
\state_\Horizon \in \safeInit{\n}\ominus \ball{\epsguarantee}\}.
\end{multline}

This denotes the set of all policies that, when applied to the true system starting from some state in a set $X\subseteq\X$, return the system to the safe set while satisfying all constraints with an $\epsguarantee$-margin. 
The user-chosen tolerance $\epsilon$ accounts for the fact that dynamics can be learned only up to a nonzero tolerance in finite time. It can be any positive constant, subject to a lower bound proportional to $\noisebound$, see \cref{apxsec:dyn_exploration_proof}. For example, in the case of no noise $\noisebound=0$, the tolerance $\epsguarantee > 0$ can be chosen arbitrarily small.

To approximate this set via exploration, we define an inner and outer approximation of the true policy set~\eqref{eq:true_policy_set}, namely, pessimistic and optimistic policy sets, respectively, as follows:
\begin{multline}
\pessiSet[]{\n}(X;\Horizon) \coloneqq \{ \pi \in \Pi_{\Horizon} ~| \exists \state_0 \in X: \state_{\h+1} = \dynVec(\state_\h, \coninput_\h) + \noise_\h, \forall \noise_\h \in \W,  \forall \dynVec \in \dynSet[\n],\\ 
\state_\h \in \X,  \coninput_{\h}\coloneqq \pi_\h(\state_\h) \in \inputSpace, \forall \h \in \Intrange{0}{\Horizon-1}, 
\state_\Horizon \in \safeInit{\n} \};
\label{eq:def_pessi_set}
\end{multline}
\begin{multline}
    \optiSet[]{\n}(X;\Horizon) \coloneqq \{ \pi \in \Pi_{\Horizon} | \exists \dynVec \in \dynSet[\n], \exists \state_0 \in X: \state_{\h+1} = \dynVec(\state_\h, \coninput_\h) + \noise_\h, ~\forall \noise_\h \in \W, \\
    \state_\h \in \X\ominus\ball{\epsguarantee},  
    \coninput_{\h} \coloneqq \pi_\h(\state_\h) \in \inputSpace\ominus\ball{\epsguarantee}, \forall \h \in \Intrange{0}{\Horizon-1}, 
\state_\Horizon \in \safeInit{\n}\ominus \ball{\epsguarantee}\}. \label{eq:def_opti_set}
\end{multline}

The \emph{pessimistic} policies ensure that \textit{for all dynamics} $\dynVec \in \dynSet[\n]$, the system is safe and returns to the safe set $\state_\Horizon \in \safeInit{\n}$, as shown with the dashed region in \cref{fig:dyn_opti_pessi_definition}.
In contrast, the \emph{optimistic} policies ensure \textit{at least one dynamics} is safe while keeping a tolerance $\epsguarantee>0$ with respect to the constraint. These set definitions naturally satisfy  $\pessiSet[]{\n}(X;\Horizon)\subseteq \truePolicySet[]{\n}(X;\Horizon)  \subseteq \optiSetwe[]{\n}(X;\Horizon)$,$\forall n\in\mathbb{N}, X\subseteq\mathcal{X}$, given $\dyntrueVec\in\dynSet[n]$ (cf. \cref{prop:over_under_approx} in \cref{apxsec:dyn_exploration_proof}).

\looseness -1 Through continuous exploration, we can expand the pessimistic policy set $\pessiSet[]{\n}(X;\Horizon)$, however, we cannot hope to discover disconnected safe policies which require executing possibly unsafe policies during exploration; see ~\cref{fig:dyn_exploration_convergence} for an illustration. 
Hence, we restrict our theoretical analysis to the \emph{connected policy set} \mbox{$\truePolicySet[]{c} (X;\Horizon) \subseteq \truePolicySet[]{}(X;\Horizon)$}, 
which satisfies the following connectedness condition: \mbox{$\forall \pi^\star \in \truePolicySet[]{c}(X;\Horizon)$}, there exists a continuous curve \mbox{$\rho:[0,1] \to \truePolicySet[]{c}(X;\Horizon)$}, such that $\rho(0)=\pi^\star , \rho(1) \in \pessiSet[]{\n}(X;\Horizon)$.
In addition, during recursive planning, the agent may end in some fixed state $\state(\ki) \in \safeInit{\n}$. But we are interested in the optimistic policy starting from any possible state $\state' \in \safeInit{\n}$, see \cref{fig:dyn_opti_pessi_definition}. Thus, we need to extend our horizon slightly up to $\Delta \Horizon \in\mathbb{N}$, which allows the policy to steer the system from $\state(\ki) \to \state'$ via a policy $\hat{\pi} \in \Pi_{\Delta \Horizon}$ (see controllability property later in \cref{assump:safeSet}). Hence, the best any algorithm can guarantee for exploration is the following:
\begin{tcolorbox}[colframe=white!, top=2pt,left=2pt,right=6pt,bottom=2pt]
\begin{objective}[Maximum safe dynamics exploration] There exists a uniform bound $\n^\star \in \Intp$, such that after some $\nfin \leq  \n^\star$ at the current state $\state(\ki) \!\in \safeInit{\nfin}$, it holds that:  \label{obj:maximum_exploration}
\begin{align}
\forall \state' \in  \safeInit{\nfin}, \pi^\star \in \truePolicySet[,\epsilon]{c,\nfin}(\state'; \Horizon),  \exists \delta \h \in \Intrange{0}{\Delta \Horizon}, \hat{\pi} \in\Pi_{\delta \h}: 
[\hat{\pi}, \pi^\star ] \in \pessiSet[]{\nfin}(\state(\ki);\Horizon + \delta \h) .\nonumber 
\end{align}
\end{objective}
\end{tcolorbox}
Achieving \cref{obj:maximum_exploration} ensures that for every policy $\pi^\star$ in the true policy set $\truePolicySet[,\epsilon]{c}(\safeInit{\nfin}; \Horizon)$, there is a corresponding policy in the pessimistic policy set with slightly extended horizon $\Horizon + \delta \h$. This implies that after exploration within $\nfin \leq \n^\star$ iterations, the $\epsguarantee$-true policy set is effectively a subset of the pessimistic policy set.
Since we can optimize over the pessimistic policy set,
achieving this objective ensures that we have explored enough
to recover close-to-optimal policies from the true set; see~\cref{sec:reward_maxim_exploration}. 

\begin{remark}
    \looseness -1 We consider two horizons, $\Horizon$ and $\Horizon\!+\!\delta \h$, where $\delta \h$ considers the time to move in the safe set; see \cref{fig:dyn_opti_pessi_definition}.
    In case we have a perfect model inside the safe set (known $\dyntrueVec$ and $\noise = 0$), 
    we can exactly steer the system to the desired location in the set $\safeInit{n}$, and \cref{obj:maximum_exploration} can be replaced by the more intuitive result:  
   $\truePolicySet[,\epsilon]{\nfin}(\safeInit{\nfin};\Horizon) \subseteq \optiSet[]{\nfin}(\safeInit{\nfin};\Horizon) \subseteq \pessiSet[]{\nfin}(\safeInit{\nfin};\Horizon)$.
\end{remark}

\subsection{Exploration process for maximum safe dynamics learning}
\looseness -1
Next, we present our novel exploration process, ensuring that, at each planning step, the agent gathers information about the dynamics until a desired tolerance is achieved. 
Specifically, we define three tolerances: \textit{(i)} $\epsguarantee$, a user-defined tolerance that guarantees exploration of the policy set (cf. \cref{obj:maximum_exploration}),
\textit{(ii)} $\epsterminal$, the tolerance on the learned dynamics, 
and \textit{(iii)} $\epscollect$, the tolerance with which measurements are collected.
Given a user-specified $\epsguarantee$ (arbitrary subject to noise bound), the tolerances $\epsterminal$ and $\epscollect$ can be uniquely determined as specified in \cref{apxsec:dyn_exploration_proof}.  

\begin{figure}[t]
    \centering
    \scalebox{0.7}{\input{figure/dyn_traj_diff}} 
\caption{ 
\change{Visual illustration of the tolerances $\epsterminal$ and $\epscollect$ used to define the dynamics exploration scheme. 
Trajectories generated under the sampled dynamics $\dynVec^s$ and the true dynamics $\dyntrueVec$ are shown, with the deviation $d_h$ measuring their discrepancy after horizon $h$. The arrows perpendicular to the trajectory manifold indicate the uncertainty, which exceeds $\epsterminal$ and $\epscollect$ at the location of interest.}}
\label{fig:selecting_sampling_rule} 
\end{figure}

\looseness -1 \cref{fig:selecting_sampling_rule} illustrates the relation between $\epscollect$ and $\epsterminal$. 
Suppose that we learn the dynamics up to $\epscollect$ tolerance. 
Then, in the worst case, after $h$-steps, the planned and true path will deviate by at most $(\epscollect +\noisebound) \LipDyn_{h}$, where the noise bound $\noisebound\geq\|\eta(k)\|$ $\forall \eta\in\mathcal{W}$ and $\LipDyn_{h}$ is related to the Lipschitz constant $\LipDyn$. 
Hence, if the distance $d_h$ (\cref{fig:selecting_sampling_rule}) between the planned and true trajectory after $h$-steps exceeds 
$(\epscollect + \noisebound) \LipDyn_{h}$, then there was at least one point in the trajectory where the dynamics was not known up to $\epscollect$ tolerance. 
Hence, by (optimistically) searching for uncertainty $\cwidth[\n](\state_\h, \coninput[\h])$ greater than $\epsterminal \coloneqq \epscollect + \LipWidthPi \LipDyn_{\totHorizon} (\epscollect+\noisebound)$, 
we guarantee to learn about the dynamics. Here, $\totHorizon\coloneqq\Horizon + \Delta \Horizon$ is a constant horizon with $\Delta \Horizon\geq \delta h$ accounting for the appended trajectory, see \cref{assump:safeSet} later.

\looseness -1
\mypar{Exploration scheme} Given this intuition, we achieve~\cref{obj:maximum_exploration} using the following sampling strategy at the $(\n+1)^{th}$ planning iteration:
\begin{align}\label{eq:sampling_rule_timex}
    \mathrm{Find}~\pi^p \in \pessiSet[]{\n}(\state(k);\totHorizon)~  
    :~\cwidth[\n](\state_\h, \coninput[\h]) \geq \epsterminal~\text{for some } \dynVec\in\dynSet[\n], h\in\Intrange{0}{\totHorizon-1},
\end{align}
where the state $\state_{\h+1} = \dynVec(\state_{\h}, \coninput_{\h})$ is propagated 
using action sequence $\coninput_{\h} \coloneqq \pi^p_\h(\state_\h)$, $\h \in \Intrange{0}{\totHorizon-1}$,  
with the optimized pessimistically safe policy $\pi^p \in \pessiSet[]{\n}(\state(k);\totHorizon)$ starting from the current state $\state_0=\state(\ki)$. 
The optimized policy $\pi^p$ ensures safe operation for the horizon $\totHorizon$ and ends in the safe set $\safeInit{\n}$ for all dynamics $\dynVec\in\dynSet[\n]$. 
In addition, the policy ensures that for at least one dynamics $\dynVec\in\dynSet[\n]$, execution leads to an informative location. 
Problem~\eqref{eq:sampling_rule_timex} provides a sufficient condition to guarantee exploration; \cref{rem:cost_full_exploration} outlines objectives that accelerate complete exploration, and we further expand it to maximize cumulative rewards in \cref{sec:reward_maxim_exploration}.
The finite-horizon Problem~\eqref{eq:sampling_rule_timex} can be solved efficiently using optimization algorithms tailored for MPC, see \cref{rem:implementation} for details on the implementation.

\looseness -1 
\mypar{Safe dynamics exploration process} 
\cref{alg:full_domain_exploration_basic} summarizes the proposed framework. 
The agent starts at location \mbox{$
\state(0)\in\safeInit{0}$} and computes the pessimistic policy by solving Problem~\eqref{eq:sampling_rule_timex} in~\cref{lin:solve_sampling_problem}. 
The agent applies the pessimistic policy to the true (unknown) dynamics, yielding a state sequence $\state_\h, \h\in \Intrange{0}{\totHorizon}$ and collects data along the way 
in \cref{lin:execute_pessi_policy}. 
We update the dynamical model with the collected data, and the agent ends up at some state \mbox{$\state(\ki) = \state_{\totHorizon}\in\safeInit{n}$}.  
We then resolve Problem \eqref{eq:sampling_rule_timex} to identify the next safe exploration policy and the process continues until Problem~\eqref{eq:sampling_rule_timex} is infeasible, i.e., $\cwidth[\n](\state_\h,\coninput[\h]) < \epsterminal$ for any state and action $\state_\h, \coninput_\h$ that can be reached with any dynamics $\dynVec\in\dynSet[n]$ using any pessimistically safe policy $\pessiSet[]{\n}(\state(\ki);\totHorizon)$.
This implies that there are no more safely reachable informative states and thus the algorithm terminates in \cref{lin:full_exploration_terminate}. 

\begin{remark}\label{rem:cost_full_exploration}
\looseness -1
Any policy that satisfies the constraints of Problem~\eqref{eq:sampling_rule_timex} can be applied and will ensure satisfaction of~\cref{obj:maximum_exploration}. 
A particularly attractive objective is $\max_{\pi\in\pessiSet[]{\n}(\state_s;\totHorizon)} \max_{\dynVec\in\mathcal{F}_n}$ $\sum_{\h=0}^{\totHorizon-1} \cwidth[\n](x_h,u_h)$, which ensures that the constraints of~\eqref{eq:sampling_rule_timex} hold if the reward exceeds $\epsilon_d \totHorizon$. If rewards are less than $\epsterminal$, Problem~\eqref{eq:sampling_rule_timex} is infeasible, which ensures termination of \cref{alg:full_domain_exploration_basic}.
\end{remark}

\setcounter{algorithm}{0}
\begin{algorithm}[t]
\caption{Safe guaranteed dynamics exploration}
\begin{algorithmic}[1]
\State \textbf{Initialize:} Start at $\state(0) \in \safeInit{0}$, $\dynSet[0]$, $\totHorizon$, Tol. $\epsilon$, Data $\D_0$
\For{$\n = 0,1, \hdots $} \label{alg:termination-condi-sagedynx}
\State $\pi^p \leftarrow$ Solve Problem~\eqref{eq:sampling_rule_timex} with current state $\state(k)$.\label{lin:solve_sampling_problem}
 \If{Problem~\eqref{eq:sampling_rule_timex} is infeasible} terminate \label{lin:full_exploration_terminate} \EndIf 
\State $\state(\ki) \leftarrow$ Apply $\pi^p$ to $\dyntrueVec$ for $\totHorizon$ steps and collect $\mathcal{D}_c \coloneqq \{ (\state_{\h+1}, \stateAct_{\h}), \forall \h \in \Intrange{0}{\totHorizon-1}\}$. \label{lin:execute_pessi_policy}
\State Update $\dynSet[\n]$ model with $\mathcal{D}_{\n+1} \leftarrow \mathcal{D}_\n \cup \mathcal{D}_c$. 
\EndFor
\end{algorithmic}
\label{alg:full_domain_exploration_basic}
\end{algorithm}

\subsection{Theoretical guarantees}
In the following, we prove that the sampling rule~\eqref{eq:sampling_rule_timex} guarantees  \emph{maximum safe dynamics exploration} (\cref{obj:maximum_exploration}). 
To do so, we begin by formalizing the safe set assumption as follows:
\begin{assumption}[Safe set]
\label{assump:safeSet}
The agent starts in a known safe set $\state(0) \in \safeInit{0}$ which $\forall \n \geq 0$ satisfies 
\begin{itemize}[labelsep=0.6em, left=12pt,itemsep=2pt, parsep=0pt, topsep=4pt, partopsep=0pt]
\item 
Monotonicity: $\safeInit{\n} \subseteq \safeInit{\n+1}$;
\item Invariance: $\safeInit{\n}$ is a robust positive invariant set with a terminal policy  $\pi_\mathrm{f} \in \Pi$, i.e., $\forall \dynVec \in \dynSet[\n]$, $\noise \in \W$, $\state \in \safeInit{\n}$, $\dynVec(\state, \pi_{\mathrm{f}}(\state)) + \noise \in \safeInit{\n} \ominus \ball{\epsguarantee}, \state \in \X \ominus \ball{\epsguarantee}, \pi_{\mathrm{f}}(\state) \in \inputSpace\ominus \ball{\epsguarantee}$;
\item Controllability: All dynamics in the set $\dynSet[\n]$ can be controlled to any state in the safe set while satisfying constraints within time $\timeX$, i.e., 
$\forall \state_s, \state_e \in \safeInit{\n},\dynVec \in \dynSet[\n]$, $\exists \delta \h \in\Intrange{0}{\timeX},  \hat{\pi} \in \Pi_{\delta \h}:$ 
$\state_0 = \state_s,  \state_{\h+1} = \dynVec(\state_\h, \coninput_\h), \state_\h \! \in\! \safeInit{\n} \ominus \ball{\epsguarantee},  \coninput_{\h}\!\coloneqq\! \hat{\pi}_\h(\state_\h) \! \in\! \inputSpace\ominus \ball{\epsguarantee}, \forall \h \! \in\! \Intrange{0}{\delta \h - 1}, \state_{\delta \h} = \state_e$.
\end{itemize}
\end{assumption}

\looseness -1
The assumption is comparable to the safe initial seed for safe exploration~\citep{safe-bo-sui15,berkenkamp2017safe,as2024actsafe,prajapat2025safe} and safe terminal sets assumed in the MPC literature~\citep{wabersich2023data,saccani2022multitrajectory,soloperto2023safe,kohler2023analysis}. The safe set can be designed by sampling the dynamics $\dynVec \in \dynSet[\n]$ and computing a common Lyapunov function around the equilibrium point, similar to~\citet[Sec. 6.1]{prajapat2025finite}.
The safe set $\safeInit{\n}$ can be expanded with $\n$ as the agent gathers more data about the dynamics and shrinks the dynamics set $\dynSet[\n]$. 
To show finite time termination of this simple sampling strategy, we also consider a standard regularity assumption on the kernel. For this, we define maximum information capacity $\gammaconst{n} \coloneqq \sup_{\D_\n \subseteq \X\times\U: |\D_\n| \leq D_{0}+\n\totHorizon} I(Y_{\D_\n};\dyntrueVec_{\D_\n})$~\citep{beta_srinivas} that upper bounds the information that can be obtained with a finite number of measurements $D_\n$ collected at any set $\D_\n$, where $I$ denotes the mutual information associated to the GP model, see \cref{lem:mutual_information} in \cref{apxsec:dyn_exploration_proof} for details.

\begin{assumption} \label{assump:sublinear}
    $\!\betaconst[\n]\!\gamma_{\n}$ grows sublinearly in $\n$, i.e., $\betaconst[\n]\!\gamma_{\n} \!< \!\mathcal{O}(\n)$. 
\end{assumption}

\looseness -1 Such an assumption is common in most prior works \citep{safe-bo-sui15,sui2018stagewise,turchetta2016safemdp,berkenkamp2023bayesian,prajapat2022near,beta_srinivas,beta_chowdhury17a} to establish sample complexity or sublinear regret results and is not restrictive. Indeed, due to the bounded $\Bg$ (\cref{assump:q_RKHS}), $\betaconst[\n]\gamma_{\n}$ grows sublinear in ${\n}$ for compact domains $\X$ and commonly used kernels, e.g., linear, squared exponential, Mat\'ern, etc., with sufficient eigen decay \citep{gammaT-vakili21a,beta_srinivas}. 
The following theorem guarantees that the closed-loop system resulting from~\cref{alg:full_domain_exploration_basic} ensures maximum dynamics exploration and safety for the true (unknown) system~\eqref{eq:system_dyn}.
\begin{restatable}[Maximum safe dynamics exploration]{theorm}{restatedynexploration} Let \cref{assump:q_RKHS,assump:safeSet,assump:lipschitz,assump:sublinear} hold. 
Let $\n^{\star}$ be the largest integer satisfying    
$\frac{\n^{\star}}{\beta_{\n^{\star}} \gamma_{\n^{\star}}} \leq \frac{C_1}{\epscollect^2}$
with $C_1 =  8\totHorizon/ \log (1 + \totHorizon \noiseconst)$.
Then, with at least $1-\delta$ probability, \cref{alg:full_domain_exploration_basic} guarantees 
\begin{itemize}[labelsep=0.6em, left=12pt,itemsep=2pt, parsep=0pt, topsep=4pt, partopsep=0pt]
\item safety for all times: $\state(\ki) \in \X, \coninput(\ki) \in \inputSpace ~ \forall \ki \in \Intp$;
\item termination after $\nfin\leq n^\star$ iterations;
\item \cref{obj:maximum_exploration}, ensuring maximum safe dynamics exploration. 
\end{itemize}
\label{thm:maximum_dynamics_exploration}
\end{restatable}
Thus, by executing a pessimistically safe policy and collecting data where uncertainty is larger than $\epscollect$ via~\eqref{eq:sampling_rule_timex}, we achieve dynamics exploration up to the $\epsterminal$ tolerance while being safe at all times for uncertain nonlinear systems. 
This corresponds to maximal safe exploration of the policies up to arbitrary tolerance $\epsguarantee$ subject to noise (\cref{obj:maximum_exploration}), with the total exploration time bounded by $\ki \leq \n^\star \totHorizon$. 

\mypar{Proof sketch} 
In \cref{thm:maximum_dynamics_exploration}, 
safety is ensured by always executing policies from the pessimistic safe set. 
We then show that only a finite number of samples, $\n^\star$, can be collected before the model uncertainty falls below the $\epsterminal$ threshold, building on the analysis in \citep{beta_srinivas,prajapat2025safe}.
While each iteration of Problem~\eqref{eq:sampling_rule_timex} identifies only a single informative state, collecting just one measurement per trajectory would be inefficient. Therefore, we extend the analysis to collect multiple measurements along a single trajectory, without requiring a GP (posterior) update at each step.
A key challenge in our setting is that these measurements must be acquired by actively steering the unknown dynamical system, cf.~\cref{fig:selecting_sampling_rule}.
Once the model uncertainty in the pessimistic set is small, the uncertainty close to the optimistic trajectory is also small, which implies full exploration, i.e., \cref{obj:maximum_exploration}.
The complete proof is provided in \cref{apxsec:dyn_exploration_proof}.

%% file: figure/dyn_opti_pessi_control_set.tex
 
\tikzset{
pattern size/.store in=\mcSize, 
pattern size = 5pt,
pattern thickness/.store in=\mcThickness, 
pattern thickness = 0.3pt,
pattern radius/.store in=\mcRadius, 
pattern radius = 1pt}
\makeatletter
\pgfutil@ifundefined{pgf@pattern@name@_ow7fk99jj}{
\pgfdeclarepatternformonly[\mcThickness,\mcSize]{_ow7fk99jj}
{\pgfqpoint{-\mcThickness}{-\mcThickness}}
{\pgfpoint{\mcSize}{\mcSize}}
{\pgfpoint{\mcSize}{\mcSize}}
{
\pgfsetcolor{\tikz@pattern@color}
\pgfsetlinewidth{\mcThickness}
\pgfpathmoveto{\pgfpointorigin}
\pgfpathlineto{\pgfpoint{0}{\mcSize}}
\pgfusepath{stroke}
}}
\makeatother
\tikzset{every picture/.style={line width=0.75pt}} 

\begin{tikzpicture}[x=0.75pt,y=0.75pt,yscale=-1,xscale=1]

\draw  [fill={rgb, 255:red, 128; green, 128; blue, 128 }  ,fill opacity=0.3 ] (658.94,211.05) .. controls (658.94,184.56) and (680.42,163.08) .. (706.91,163.08) -- (1025.64,163.08) .. controls (1052.13,163.08) and (1073.61,184.56) .. (1073.61,211.05) -- (1073.61,354.95) .. controls (1073.61,381.44) and (1052.13,402.92) .. (1025.64,402.92) -- (706.91,402.92) .. controls (680.42,402.92) and (658.94,381.44) .. (658.94,354.95) -- cycle ;
\draw  [fill={rgb, 255:red, 108; green, 215; blue, 108 }  ,fill opacity=0.6 ] (731.5,198.5) .. controls (792,213.5) and (900,180.5) .. (971,189) .. controls (1042,197.5) and (1065,284.17) .. (1051.5,356) .. controls (1038,427.83) and (813.6,376) .. (797.2,380.8) .. controls (780.8,385.6) and (704.1,374.1) .. (697,348.5) .. controls (689.9,322.9) and (671,183.5) .. (731.5,198.5) -- cycle ;
\draw  [fill={rgb, 255:red, 108; green, 215; blue, 108 }  ,fill opacity=0.3 ] (739.54,207.45) .. controls (797.04,221) and (899.68,191.18) .. (967.16,198.86) .. controls (1034.64,206.54) and (1056.5,284.86) .. (1043.67,349.78) .. controls (1030.84,414.69) and (817.57,367.85) .. (801.98,372.19) .. controls (786.39,376.53) and (713.5,366.13) .. (706.75,343) .. controls (700,319.87) and (682.04,193.89) .. (739.54,207.45) -- cycle ;
\draw  [fill={rgb, 255:red, 80; green, 227; blue, 194 }  ,fill opacity=0.67 ] (801.57,339.9) .. controls (789.11,313.35) and (796,280.49) .. (816.96,266.51) .. controls (837.91,252.52) and (865,262.71) .. (877.45,289.26) .. controls (889.9,315.81) and (883.01,348.67) .. (862.06,362.66) .. controls (841.1,376.64) and (814.02,366.46) .. (801.57,339.9) -- cycle ;
\draw  [line width=3.75]  (818.2,348.39) .. controls (818.2,347.68) and (817.66,347.11) .. (816.99,347.11) .. controls (816.32,347.11) and (815.78,347.68) .. (815.78,348.39) .. controls (815.78,349.1) and (816.32,349.67) .. (816.99,349.67) .. controls (817.66,349.67) and (818.2,349.1) .. (818.2,348.39) -- cycle ;
\draw [color={rgb, 255:red, 233; green, 160; blue, 49 }  ,draw opacity=1 ][line width=3]    (816.99,349.67) .. controls (745.91,407.5) and (661.4,261.33) .. (731.7,289) .. controls (800.95,316.25) and (726.43,227.72) .. (827.84,273.8) ;
\draw [shift={(832.59,276)}, rotate = 205.18] [fill={rgb, 255:red, 233; green, 160; blue, 49 }  ,fill opacity=1 ][line width=0.08]  [draw opacity=0] (16.97,-8.15) -- (0,0) -- (16.97,8.15) -- cycle    ;
\draw    (855,291.33) .. controls (856.61,293.08) and (856.61,294.79) .. (854.98,296.46) .. controls (853.45,298.22) and (853.63,299.86) .. (855.5,301.39) .. controls (857.42,302.66) and (857.75,304.29) .. (856.48,306.28) .. controls (855.31,308.4) and (855.74,310.02) .. (857.78,311.13) .. controls (859.84,312.36) and (860.23,313.95) .. (858.95,315.9) .. controls (857.4,317.4) and (857.11,319.01) .. (858.1,320.73) .. controls (857.17,323.15) and (855.66,323.74) .. (853.57,322.5) .. controls (851.76,320.97) and (850.11,321.05) .. (848.64,322.73) .. controls (846.8,324.34) and (845.05,324.27) .. (843.39,322.52) .. controls (841.89,320.76) and (840.25,320.66) .. (838.48,322.21) .. controls (836.69,323.76) and (835.09,323.67) .. (833.66,321.93) .. controls (831.88,320.24) and (830.19,320.29) .. (828.6,322.07) .. controls (827.55,324.03) and (826.01,324.65) .. (823.99,323.92) .. controls (821.5,324.27) and (820.58,325.61) .. (821.22,327.96) .. controls (822.24,330.03) and (821.72,331.68) .. (819.65,332.92) .. controls (817.65,334.23) and (817.31,335.88) .. (818.62,337.86) -- (818.62,337.86) -- (817.38,345.69) ;
\draw [shift={(816.99,348.39)}, rotate = 277.95] [fill={rgb, 255:red, 0; green, 0; blue, 0 }  ][line width=0.08]  [draw opacity=0] (8.93,-4.29) -- (0,0) -- (8.93,4.29) -- cycle    ;
\draw [line width=1.5]    (998.19,199.9) -- (996.08,203.7) ;
\draw [shift={(994.13,207.2)}, rotate = 299.05] [fill={rgb, 255:red, 0; green, 0; blue, 0 }  ][line width=0.08]  [draw opacity=0] (4.64,-2.23) -- (0,0) -- (4.64,2.23) -- cycle    ;
\draw [shift={(1000.13,196.4)}, rotate = 119.05] [fill={rgb, 255:red, 0; green, 0; blue, 0 }  ][line width=0.08]  [draw opacity=0] (4.64,-2.23) -- (0,0) -- (4.64,2.23) -- cycle    ;
\draw  [line width=3.75]  (873.58,336.17) .. controls (873.58,335.46) and (874.15,334.89) .. (874.85,334.89) .. controls (875.56,334.89) and (876.13,335.46) .. (876.13,336.17) .. controls (876.13,336.87) and (875.56,337.45) .. (874.85,337.45) .. controls (874.15,337.45) and (873.58,336.87) .. (873.58,336.17) -- cycle ;
\draw  [line width=3.75]  (857.42,291.33) .. controls (857.42,290.63) and (856.88,290.05) .. (856.21,290.05) .. controls (855.54,290.05) and (855,290.63) .. (855,291.33) .. controls (855,292.04) and (855.54,292.61) .. (856.21,292.61) .. controls (856.88,292.61) and (857.42,292.04) .. (857.42,291.33) -- cycle ;
\draw  [pattern=_ow7fk99jj,pattern size=4pt,pattern thickness=0.75pt,pattern radius=0pt, pattern color={rgb, 255:red, 0; green, 0; blue, 0}] (956.89,393.06) .. controls (927.09,392.06) and (887.56,347.11) .. (877.11,338.44) .. controls (866.67,329.78) and (895.11,349.56) .. (919.11,354.22) .. controls (943.11,358.89) and (970,356) .. (970,343.78) .. controls (970,331.56) and (935.56,319.33) .. (944.56,292.22) .. controls (953.56,265.11) and (914.44,268.22) .. (895.78,279.56) .. controls (877.11,290.89) and (880.27,305.49) .. (876.67,304.89) .. controls (873.07,304.29) and (874.05,292.11) .. (866.28,283) .. controls (858.5,273.89) and (850,272.89) .. (849.78,267.33) .. controls (849.56,261.78) and (874.44,237) .. (904.67,228.39) .. controls (934.89,219.78) and (972.44,229.56) .. (988.22,244.67) .. controls (1004,259.78) and (1016.16,275.61) .. (1016,310.61) .. controls (1015.84,345.61) and (986.69,394.06) .. (956.89,393.06) -- cycle ;
\draw [color={rgb, 255:red, 65; green, 117; blue, 5 }  ,draw opacity=1 ][line width=3]    (889.67,346.33) .. controls (937,394) and (1022.33,352.33) .. (984,310.33) .. controls (946.24,268.96) and (959.91,207.54) .. (871.76,277.38) ;
\draw [shift={(867.67,280.67)}, rotate = 321.01] [fill={rgb, 255:red, 65; green, 117; blue, 5 }  ,fill opacity=1 ][line width=0.08]  [draw opacity=0] (16.97,-8.15) -- (0,0) -- (16.97,8.15) -- cycle    ;
\draw  [line width=3.75]  (964.09,364.67) .. controls (964.09,363.96) and (963.55,363.39) .. (962.88,363.39) .. controls (962.21,363.39) and (961.67,363.96) .. (961.67,364.67) .. controls (961.67,365.37) and (962.21,365.95) .. (962.88,365.95) .. controls (963.55,365.95) and (964.09,365.37) .. (964.09,364.67) -- cycle ;

\draw (804.47,300) node [anchor=north west][inner sep=0.75pt]  [font=\Large]  {$\mathbb{X}_{n}$};
\draw (670.9,177.2) node [anchor=north west][inner sep=0.75pt]  [font=\Huge]  {$\mathcal{X}$};
\draw (720.33,240.97) node [anchor=north west][inner sep=0.75pt]  [font=\LARGE]  {$\pi _{n}^{o,\epsilon }$};
\draw (815.47,274.83) node [anchor=north west][inner sep=0.75pt]  [font=\Large]  {$x(k)$};
\draw (828.13,345.5) node [anchor=north west][inner sep=0.75pt]  [font=\Large]  {$x'$};
\draw (830.5,327.8) node [anchor=north west][inner sep=0.75pt]  [font=\Large]  {$\hat{\pi}$};
\draw (1009.13,180.93) node [anchor=north west][inner sep=0.75pt]  [font=\LARGE]  {$\epsilon $};
\draw (858.33,210.3) node [anchor=north west][inner sep=0.75pt]  [font=\LARGE]  {$\pi _{n}^{p}$};
\draw (980,369.3) node [anchor=north west][inner sep=0.75pt]  [font=\LARGE]  {$w_{n} \geq \epsilon _{d}$};

\end{tikzpicture}

%% file: figure/dyn_exploration_convergence.tex
\tikzset{every picture/.style={line width=0.75pt}} 

\begin{tikzpicture}[x=0.75pt,y=0.75pt,yscale=-1,xscale=1]

\draw  [fill={rgb, 255:red, 128; green, 128; blue, 128 }  ,fill opacity=0.3 ] (181.2,774.12) .. controls (181.2,747.65) and (202.65,726.2) .. (229.12,726.2) -- (548.08,726.2) .. controls (574.55,726.2) and (596,747.65) .. (596,774.12) -- (596,917.88) .. controls (596,944.35) and (574.55,965.8) .. (548.08,965.8) -- (229.12,965.8) .. controls (202.65,965.8) and (181.2,944.35) .. (181.2,917.88) -- cycle ;
\draw  [fill={rgb, 255:red, 245; green, 166; blue, 35 }  ,fill opacity=0.2 ] (218.5,781) .. controls (230.5,770.5) and (373.7,720.4) .. (420.5,729) .. controls (467.3,737.6) and (493.5,895) .. (488.5,915) .. controls (483.5,935) and (338.9,967.27) .. (305.6,960.2) .. controls (272.3,953.13) and (222.6,923.6) .. (215.5,898) .. controls (208.4,872.4) and (206.5,791.5) .. (218.5,781) -- cycle ;
\draw  [fill={rgb, 255:red, 108; green, 215; blue, 108 }  ,fill opacity=0.3 ] (222,785.13) .. controls (234,774.63) and (366.7,727.4) .. (413.5,736) .. controls (460.3,744.6) and (480.5,893) .. (475.5,913) .. controls (470.5,933) and (338.8,960.07) .. (305.5,953) .. controls (272.2,945.93) and (228.1,918.73) .. (221,893.13) .. controls (213.9,867.53) and (210,795.63) .. (222,785.13) -- cycle ;
\draw  [fill={rgb, 255:red, 245; green, 166; blue, 35 }  ,fill opacity=0.72 ] (222,785.13) .. controls (234,774.63) and (295.6,758.2) .. (322.4,759) .. controls (349.2,759.8) and (369.1,739.4) .. (403.5,745) .. controls (437.9,750.6) and (482.4,891.8) .. (460.8,906.6) .. controls (439.2,921.4) and (432.4,940.6) .. (398.8,939) .. controls (365.2,937.4) and (316,940.6) .. (300.8,943.4) .. controls (285.6,946.2) and (228.1,918.73) .. (221,893.13) .. controls (213.9,867.53) and (210,795.63) .. (222,785.13) -- cycle ;
\draw  [fill={rgb, 255:red, 245; green, 166; blue, 35 }  ,fill opacity=0.2 ] (587.5,788) .. controls (599.5,806) and (585.5,848) .. (544.5,847) .. controls (503.5,846) and (529.5,807) .. (509.5,779) .. controls (489.5,751) and (575.5,770) .. (587.5,788) -- cycle ;
\draw  [fill={rgb, 255:red, 245; green, 166; blue, 35 }  ,fill opacity=0.72 ] (580.3,792) .. controls (588.3,808) and (576.3,839) .. (545.3,838) .. controls (514.3,837) and (537.3,812) .. (520.3,787) .. controls (503.3,762) and (572.3,776) .. (580.3,792) -- cycle ;
\draw  [fill={rgb, 255:red, 108; green, 215; blue, 108 }  ,fill opacity=0.6 ] (311.5,868) .. controls (320.5,889) and (338.5,888) .. (311.5,901) .. controls (284.5,914) and (236.5,907) .. (246.5,874) .. controls (256.5,841) and (302.5,847) .. (311.5,868) -- cycle ;
\draw [color={rgb, 255:red, 0; green, 0; blue, 0 }  ,draw opacity=0.5 ][line width=3.75]    (539,792.4) .. controls (545.44,774.14) and (546.2,764.7) .. (533.38,752.16) ;
\draw [shift={(528.3,747.6)}, rotate = 46.01] [fill={rgb, 255:red, 0; green, 0; blue, 0 }  ,fill opacity=0.5 ][line width=0.08]  [draw opacity=0] (20.27,-9.74) -- (0,0) -- (20.27,9.74) -- (13.46,0) -- cycle    ;
\draw [color={rgb, 255:red, 0; green, 0; blue, 0 }  ,draw opacity=0.5 ][line width=3.75]    (406.7,779.6) .. controls (418.73,754.63) and (444.44,741.63) .. (471.97,742.19) ;
\draw [shift={(478.7,742.6)}, rotate = 179.96] [fill={rgb, 255:red, 0; green, 0; blue, 0 }  ,fill opacity=0.5 ][line width=0.08]  [draw opacity=0] (20.27,-9.74) -- (0,0) -- (20.27,9.74) -- (13.46,0) -- cycle    ;
\draw [color={rgb, 255:red, 0; green, 0; blue, 0 }  ,draw opacity=0.5 ][line width=3.75]    (464.5,915) .. controls (469.98,931.7) and (480.99,941.75) .. (495.02,947.63) ;
\draw [shift={(501.5,950)}, rotate = 203.19] [fill={rgb, 255:red, 0; green, 0; blue, 0 }  ,fill opacity=0.5 ][line width=0.08]  [draw opacity=0] (20.27,-9.74) -- (0,0) -- (20.27,9.74) -- (13.46,0) -- cycle    ;
\draw [color={rgb, 255:red, 0; green, 0; blue, 0 }  ,draw opacity=0.5 ][line width=3.75]    (482.5,892) .. controls (504,897.67) and (516.88,897.49) .. (538.83,890.67) ;
\draw [shift={(545.2,888.6)}, rotate = 163.49] [fill={rgb, 255:red, 0; green, 0; blue, 0 }  ,fill opacity=0.5 ][line width=0.08]  [draw opacity=0] (20.27,-9.74) -- (0,0) -- (20.27,9.74) -- (13.46,0) -- cycle    ;
\draw [color={rgb, 255:red, 0; green, 0; blue, 0 }  ,draw opacity=0.5 ][line width=3.75]    (581.3,820.6) .. controls (585.48,838.76) and (582.77,848.67) .. (572.47,862.34) ;
\draw [shift={(568.3,867.6)}, rotate = 306.11] [fill={rgb, 255:red, 0; green, 0; blue, 0 }  ,fill opacity=0.5 ][line width=0.08]  [draw opacity=0] (20.27,-9.74) -- (0,0) -- (20.27,9.74) -- (13.46,0) -- cycle    ;
\draw [line width=1.5]    (316.1,937.6) -- (317.1,963.6) ;
\draw [shift={(317.1,963.6)}, rotate = 267.8] [color={rgb, 255:red, 0; green, 0; blue, 0 }  ][line width=1.5]    (0,5.09) -- (0,-5.09)(9.95,-2.99) .. controls (6.32,-1.27) and (3.01,-0.27) .. (0,0) .. controls (3.01,0.27) and (6.32,1.27) .. (9.95,2.99)   ;
\draw [shift={(316.1,937.6)}, rotate = 87.8] [color={rgb, 255:red, 0; green, 0; blue, 0 }  ][line width=1.5]    (0,5.09) -- (0,-5.09)(9.95,-2.99) .. controls (6.32,-1.27) and (3.01,-0.27) .. (0,0) .. controls (3.01,0.27) and (6.32,1.27) .. (9.95,2.99)   ;

\draw (331.28,940.53) node [anchor=north west][inner sep=0.75pt]  [font=\LARGE]  {$\epsilon $};
\draw (552.8,873.67) node [anchor=north west][inner sep=0.75pt]  [font=\LARGE]  {$\upPi ^{\star }$};
\draw (271,867.27) node [anchor=north west][inner sep=0.75pt]  [font=\LARGE]  {$\upPi _{0}^{p}$};
\draw (489,733.27) node [anchor=north west][inner sep=0.75pt]  [font=\LARGE]  {$\upPi ^{\star ,\epsilon }$};
\draw (284,786.27) node [anchor=north west][inner sep=0.75pt]  [font=\LARGE]  {$\upPi _{c}^{\star ,\epsilon }$};
\draw (510,934.27) node [anchor=north west][inner sep=0.75pt]  [font=\LARGE]  {$\upPi _{\overline{n}}^{p}$};
\draw (203.43,741) node [anchor=north west][inner sep=0.75pt]  [font=\Huge]  {$\upPi $};

\end{tikzpicture}

%% file: sections/4_guaranteed_exploration_in_task_oriented.tex
\section{Reward maximization with intrinsic exploration}
\label{sec:reward_maxim_exploration}
\looseness -1 
Building on the general framework for safe dynamics exploration from \cref{sec:full_expl_theory}, we propose \dynExplor, 
an algorithm that focuses on maximizing rewards by reducing uncertainty \emph{only} in the regions essential to achieve optimal operation. 
\dynExplor replans directly whenever new measurements are collected instead of returning back to the safe set, thus enhancing efficiency. Furthermore, the algorithm explores in task-oriented fashion and learns a safe policy $\retPolicy$ that achieves the following objective. 
\begin{tcolorbox}[colframe=white!, top=2pt,left=2pt,right=6pt,bottom=2pt]
\begin{objective}[Safe reward maximization] There exists $\nfin \leq \n^\star$ and a constant $K>0$, such that the known safe policy $\retPolicy \in\truePolicySet{\nfin}(x(\ki), \Horizon + \delta \h)$ computed at the current state $x(\ki) \in \safeInit{\nfin}$ satisfies: \label{obj:reward_maximization}
\vspace{-0.5em}
\begin{align}\vspace{-0.5em}
\Jobj[]{\state(\ki), \dyntrueVec }{\retPolicy} \geq  \max_{\state^\star \in \safeInit{\nfin}, \pi^\star \in {\truePolicySet[, \epsilon]{c,\nfin}}(\safeInit{\nfin};\Horizon)} \Jobj[]{\state^\star, \dyntrueVec}{\pi^\star}    - K \epsguarantee.
\end{align}
\end{objective}
\end{tcolorbox}
\begin{algorithm}[!t]
\caption{Reward maximization with intrinsic exploration (\dynExplor)}
\begin{algorithmic}[1]
\State \textbf{Initialize:} Start at $\state(0) \in \safeInit{0}$, $\dynSet[0]$, $\totHorizon$, Tolerance $\epsilon$, Data $\D_0$.
\For{$\n = 0, 1, \hdots $} 
 \State $\Jobj[]{\state^o, \dynVec^o}{\optiPolicy} \leftarrow$ Solve the optimistic problem \eqref{eq:opti_obj}. \label{lin:solve_opti_prob}
 \State $\Jobj[\mathrm{p}]{\state^{\mathrm{p}}, \bm{\mu}_\n}{\pessiPolicy} \leftarrow$ Solve the pessimistic problem \eqref{eq:pessi_obj}. \label{lin:solve_pessi_prob} 
\If {$\Jobj[\mathrm{p}]{\state^{\mathrm{p}}, \bm{\mu}_\n}{\pessiPolicy} \geq \Jobj[]{\state^o, \dynVec^o}{\optiPolicy} - K \epsguarantee $} 
 \label{lin:MPC:go_termination}
 \State $\!\!\!\state(\ki) \leftarrow$ Return to the safe set using the previously optimized policy. 
    \State 
   \!\!\! Return $\retPolicy \leftarrow$ $[\hat{\pi},\! \pessiPolicy]$,\! where $\hat{\pi}$ steers $\state(\ki)$ to $\state^{\mathrm{p}}$ with $\bm{\mu}_\n$
   and terminate.%
\label{lin:MPC:move_SteadyState}
\EndIf
\State $\pessiPolicy[e], \nu \leftarrow$ Solve Problem \eqref{eq:slack_there_exists_goal}.  \label{lin:solve_relax_prob}
\If{$\nu = 0$} \label{eq:case_distinction} \hfill \texttt{(Task oriented safe expansion)} 
\State $\!\!\!\state(\ki) \!\leftarrow\!$  Execute $\pessiPolicy[e]$ on $\dyntrueVec\!$ until  some $h^\star$ such that $\exists h \leq \h^\star$ with $\cwidth[\n-1](\stateAct_{\h})\!\geq \!\epscollect$. 
\label{lin:moveTOsampling} 
\Else \hfill \texttt{(No informative location found, Return to the safe set)} 
\State $\!\!\!\state(\ki) \leftarrow$ Continue with policy optimized at \mbox{$\n - 1$} until $\safeInit{\n}$ is reached.
\State $\!\!\!\pessiPolicy[e] \leftarrow$ Solve Problem~\eqref{eq:slack_there_exists_goal} with $\nu=0$ (hard constraint). \label{lin:return}
 \State $\!\!\!\state(\ki) \leftarrow$ Execute $\pessiPolicy[e]$  
 until  some $h^\star$ such that $\exists h \leq \h^\star$ with $\cwidth[\n-1](\stateAct_{\h})\!\geq \!\epscollect$.
 \label{lin:MPC:move2}
\EndIf
\State Update $\dynSet[\n]$ with $\mathcal{D}_{\n+1} \leftarrow \mathcal{D}_{\n} \cup \mathcal{D}_c$, $\mathcal{D}_c \coloneqq \{ (\state_{\h+1}, \stateAct_{\h})\}_{\h=0}^{\h^\star-1}$.  \label{lin:MPC:updateGP} \hfill \texttt{(Update model)}
\EndFor
\end{algorithmic}
\label{alg:rew_maxim}
\end{algorithm}
\cref{obj:reward_maximization} ensures that, despite a-priori unknown dynamics, we obtain a policy with a performance that is close to the optimal policy $\pi^\star$ starting from the best state $\state^\star$ within the safe set. 
To achieve~\cref{obj:reward_maximization}, naively, one could explore the complete set of pessimistic policies using~\cref{alg:full_domain_exploration_basic} and then
maximize $\Jobj[]{x, \pessiPolicy}{\dyntrue}$ among the known safe policies, yielding a sequential two-stage algorithm. 
However, exploring the complete pessimistic policy set without considering the objective would be extremely inefficient. 
Therefore, we propose an exploration strategy that leverages the optimistic policy set, 
which excludes only those policies that are provably unsafe for all dynamics.  
Hence, reward maximization in the optimistic set
guarantees maximizing among the plausible policies. 
To formalize this, we define the optimistic and pessimistic problems, respectively, as follows:
\begin{flalign}
&\text{Optimistic problem:} \qquad\qquad  \ \ \Jobj[]{\state^o, \dynVec^o}{\optiPolicy} \!\!\!\!\!\!\!&&\coloneqq  \max_{ \state \in \safeInit{\n}, \dynVec \in \dynSet[\n], \pi \in {\optiSet[]{\n}(\state; \Horizon)}} \Jobj[]{\state, \dynVec}{\pi},  &&&\label{eq:opti_obj}
\\
&\text{Pessimistic problem:} \qquad\qquad\Jobj[\mathrm{p}]{\state^{\mathrm{p}}, \bm{\mu}_{\n}}{\pessiPolicy} \!\!\!\!\!\!\!&&\coloneqq  \max_{{\state \in \safeInit{\n},\pi \in \pessiSet[,\epssafeset]{\n}(\state; \Horizon})}  \Jobj[\mathrm{p}]{\state, \bm{\mu}_{\n}}{\pi},&&&\label{eq:pessi_obj}
  \end{flalign}
$$\text{where,}~\Jobj[\mathrm{p}]{\state, \bm{\mu}_{\n}}{\pi} = \Jobj[]{\state, \bm{\mu}_\n}{\pi} - \E \bigg[ \LiprewPi  \sum_{\h=0}^{\Horizon-1}  \sum_{i=0}^{\h} L^{i} \cwidth[\n](\state_{i}, \pi_{i}(\state_{i}))| \state_0 = \state, \pi \bigg], $$  
\looseness -1
and $\state^o,\dynVec^o,\optiPolicy, \state^{\mathrm{p}},\pessiPolicy$ denote the optimized variables. The pessimistic problem maximizes $J^{\mathrm{p}}$, which provides a pessimistic lower bound on the reward $J$ for all $\dynVec\!\!\in\!\! \dynSet[\n]$.  
The states $\state_i$ are propagated with noise under the mean dynamics $\bm{\mu}_{\n}$, which we use for computational simplicity, although $J^{\mathrm{p}}$ can, in principle, be defined using any dynamics.
Here, $\epssafeset\in(0,\epsguarantee)$ 
is a small tolerance to account for the different initial conditions in the safe set $\safeInit{\n}$ using a controllability argument (see \cref{apxsec:rew_maxim} later). 
Since the pessimistic problem provides a lower bound on the objective under the true dynamics, using it for termination allows early stopping, that is, without accurately knowing the dynamics everywhere, it still lets us recover close to optimal performance.

\looseness -1
\mypar{Reward maximization algorithm--\dynExplor} 
The method is summarized in~\cref{alg:rew_maxim}. 
Overall, we optimistically maximize the rewards while pessimistically ensuring safety. 
We start by solving the optimistic problem in \cref{lin:solve_opti_prob}, optimizing over the initial state $\state^o$, optimistic dynamics, resulting in the policy that could potentially achieve the maximum rewards.
Then we solve the pessimistic problem in \cref{lin:solve_pessi_prob} that ensures the safety of the unknown system while maximizing the rewards.
Suppose that the pessimistic rewards are sufficiently close to the optimistic rewards.
In that case, the algorithm terminates in \cref{lin:MPC:go_termination}, signifying that the agent has explored enough to converge close to the optimal solution. 
Otherwise, to continue exploring, we solve the following problem in~\cref{lin:solve_relax_prob}:
\begin{align}\label{eq:slack_there_exists_goal}
\!\!\!\!\pessiPolicy[e] = \!\!\!\!\!\argmax_{ \nu, \pi \in \pessiSet[]{\n}(\state(\ki);\totHorizon)} \!\!\!\! - \lambda \nu + 
\Jobj[\mathrm{any}]{\state(\ki), {\n}}{\pi},
~\text{s.t.}~\cwidth[\n](\state_\h,\coninput_\h) \geq \epsterminal - \nu, \nu \geq 0, h\!\in\!\Intrange{0}{\totHorizon-1},\!\! 
\end{align}
where $\lambda> 0$ is a user-defined large penalty and $\nu$ is a slack variable to ensure the feasibility of the problem. 
Here, $\state_\h$ is the propagated state using any $\dynVec \in \dynSet[\n]$ with optimized policy $\pi$. 
$J^\mathrm{any}$ can be freely designed and can also depend on $\dynSet[\n], \optiPolicy$ or $\dynVec^o$ indicated by $\n$ in \eqref{eq:slack_there_exists_goal}; see~\cref{rem:Jany} for objectives that may speed up exploration by using the solution of the optimistic problem.

\begin{figure}[t]
\begin{subfigure}[b]{0.45\columnwidth}
    \centering
    \scalebox{0.45}{\input{figure/dyn_receding}} 
    \caption{\looseness -1 Re-planning after collecting  measurements}
    \label{fig:dyn_receding_main} 
\end{subfigure}
~
\begin{subfigure}[b]{0.45\columnwidth}
    \centering
    \scalebox{0.45}{\input{figure/dyn_receding_optimality}} 
    \caption{\looseness -1 Returned policy behaves close to optimal} 
    \label{fig:dyn_receding_optimality_main} 
\end{subfigure}
    \caption{\looseness -1 Illustration of \dynExplor algorithm during a) exploration and b) convergence after satisfying the termination criteria. 
    The cyan region denotes the safe set $\safeInit{\n}$, and the green region represents the state constraints $\X$. 
    The blue dashed line shows the optimal trajectory achieved by the clairvoyant agent starting anywhere in the safe set $\safeInit{\n}$ and satisfying the constraints with $\epsilon$ margin. 
    In~\cref{fig:dyn_receding_main}, the shaded region shows the reachable set under the optimized pessimistic policy by~\eqref{eq:slack_there_exists_goal}, which starts at the current state $\state(\ki)$, ensures all the dynamics satisfy the constraints, and returns to the safe set.  
    The black line shows the agent’s executed trajectory, with small dots marking collected data and large dots indicating the time of model updates.
    With every model update (increasing $n$), the reachable set (represented by increasing darker shades) predicted with a given policy $\pi$ shrinks since the model uncertainty reduces with data.
    The agent keeps on replanning while ensuring returnability to a known safe set, but without having to actually return. 
    Once the termination criteria is satisfied, the agent returns to the safe set. As shown in ~\cref{fig:dyn_receding_optimality_main}, it then executes the returned policy which first navigates in the safe set for small $\delta \h$ horizon (wiggly line), and then executes the optimized policy $\pessiPolicy$ shown by black line which closely matches the optimal trajectory (blue dashed line).}
    \label{fig:sagedynX_receding_main} \vspace{-0.5em}
\end{figure}

\looseness -1
In Problem~\eqref{eq:slack_there_exists_goal}, with $\nu=0$ we are guaranteed to have sufficient information of \mbox{$\cwidth[\n](\state_\h,\coninput_\h) \geq \epsterminal$} while ensuring a \emph{task-oriented} approach by 
maximizing the $J^\mathrm{any}$ objective. 
We execute the resulting policy $\pessiPolicy[e]$ until some $\h^\star$ such that $\exists \h \leq \h^\star$ with $\cwidth[\n-1](\stateAct_{\h^\star})\!\geq \!\epscollect$
and update the dynamics model using all the gathered data. 
Then we replan directly from the state at $h^\star$, without first going back to the safe set, as shown in \cref{fig:dyn_receding_main}. 
If \mbox{$\nu>0$}, it implies that the agent has not found an informative location from the current state. In this case, it
returns to the safe set, and then solves Problem~\eqref{eq:slack_there_exists_goal} (with \mbox{$\nu=0$}, hard constraint) to explore another region. This hard constraint problem is guaranteed to be feasible (c.f. \cref{thm:rew_maxim}) and yields another informative state \mbox{$\cwidth[n](x,u)>\epsterminal$}, as otherwise the algorithm would already have terminated in~\cref{lin:MPC:go_termination}.
Notably, the case $\nu>0$ (\cref{lin:return}) is required only for technical reasons, but typically does not occur in practice. The algorithm mainly revolves around~\cref{lin:solve_relax_prob}, optimizes the objective, collects data, replans, and keeps on going until \cref{obj:reward_maximization} can be achieved, see \cref{fig:sagedynX_receding_main}. 
All of this is achieved while ensuring safety and maintaining a safe returnable path (pessimistic policy to the safe set $\mathbb{X}_n$) without actually returning to the safe set.


\begin{remark}[Implementation]\label{rem:implementation}
Problems~\eqref{eq:sampling_rule_timex} and \eqref{eq:slack_there_exists_goal} require computing policies $\pi\in\pessiSet{n}$ that ensure constraint satisfaction for all dynamics $\dyn\in\dynSet[\n]$ over a finite horizon $\totHorizon$. 
This can be implemented efficiently by over-approximating this reachable set using a finite number of samples from $f\sim\GP$, as shown in~\citet{prajapat2025finite}. 
By restricting the policy parameterization to simple affine policies, the resulting optimization problem can be solved using standard optimization algorithms for finite-horizon optimal control, while maintaining the safety guarantees. 
\end{remark}

\begin{remark}[Task-oriented exploration with {$\Jobj[\mathrm{any}]{\state_s, {\n}}{\pi}$}] \label{rem:Jany}
Based on the theoretical analysis, it suffices if the uncertainty in the dynamics is small around the optimistically planned trajectory. 
Hence, a natural objective $J^{\mathrm{any}}$ to consider is aiming for informative measurements $\cwidth[\n](\state_\h,\coninput_\h)  \gg \epsilon$ that are also close to the optimistically optimal trajectory.
$J^{\mathrm{any}}$ can be the distance to the optimistic path, 
the first state on the optimistic path, or cumulative rewards and uncertainty along the trajectory. 
\end{remark}

\subsection{Theoretical guarantees}

This section presents the theoretical results for the proposed algorithm \dynExplor. 
To allow for earlier termination and be able to recover the close-to-optimal policy from any state in the safe set $\safeInit{\n}$, we make the following assumption on the dynamics in the safe set.

\begin{assumption}[Dynamics in safe set] \label{assump:dynamics_tol_terminal} $\dyntrueVec$ is known up to $\epscollect \geq 0$ tolerance in the safe set, i.e.,  $\cwidth[\n-1](\state,\coninput) \leq \epscollect,~\forall \state\in\safeInit{\n},~\coninput \in \inputSpace,~n\in\mathbb{N}$.
\end{assumption}

The following theorem guarantees that the closed-loop system resulting from~\cref{alg:rew_maxim} provides a close-to-optimal policy while being safe for the unknown system~\eqref{eq:system_dyn}.
\begin{restatable}[Safe reward maximization with unknown dynamics]{theorm}{restaterewmaximization}
\label{thm:rew_maxim}
Let~\cref{assump:q_RKHS,assump:safeSet,assump:lipschitz,assump:sublinear,assump:dynamics_tol_terminal} hold.
Consider $\n^\star$ as in \cref{thm:maximum_dynamics_exploration}.
With probability at least $1\!-\!\delta$,
\begin{itemize}[labelsep=0.6em, left=12pt,itemsep=2pt, parsep=0pt, topsep=4pt, partopsep=0pt]
\item \!All the optimization problems in~\cref{alg:rew_maxim} are feasible for $\forall \n \geq 0$;
    \item \!\cref{alg:rew_maxim} guarantees safety for the unknown system at all times: $\!\state(\ki) \!\in\! \X, \coninput(\ki) \!\in\! \inputSpace, \forall \ki \!\in\! \Intp$;
    \item \!\cref{alg:rew_maxim} is guaranteed to terminate in at most $\n^\star$ iterations;
    \item \!\cref{alg:rew_maxim} returns policy $\retPolicy\!$ that satisfies \cref{obj:reward_maximization} i.e., achieves close-to-optimal performance.\!\!
\end{itemize}
\end{restatable}
Thus,~\cref{alg:rew_maxim} is guaranteed to terminate in finite time and yields a policy that achieves a performance that is arbitrarily close to optimal performance, i.e., achieves ~\cref{obj:reward_maximization}. Note that once the termination criterion is satisfied, the agent can execute the policy $\pessiPolicy$ and does not need to explore the dynamics everywhere. 
Since exploration is essential to guarantee optimality, in the worst case, the algorithm may need to learn dynamics everywhere, and the exploration time for this scenario is bounded by $\ki \leq \n^\star \totHorizon$.

\looseness -1
\mypar{Proof sketch} Analogous to ~\cref{thm:maximum_dynamics_exploration}, safety is guaranteed by only executing pessimistically safe policies.
However, it is challenging to show that once the termination criteria is satisfied, we can find a policy that satisfies~\cref{obj:reward_maximization}. This is due to the fact that the state reached in the safe set will be different from the state optimized in Problem~\eqref{eq:pessi_obj} used to check the termination criteria. To overcome this, we use the fact that the dynamics is known up to $\epscollect$ tolerance in the safe set, and problem~\eqref{eq:pessi_obj} in the termination criterion uses a threshold $\epsilon'\geq0$ for the distance to the constraints.
Finally, the algorithm terminates in finite time, at the latest once the uncertainty is reduced below $\epsterminal$ everywhere, with the worst-case sample complexity shown in \cref{thm:maximum_dynamics_exploration}. The detailed proof can be found in~\cref{apxsec:rew_maxim}.


\change{The following corollary shows near-optimality of the pessimistic policy derived from \eqref{eq:pessi_obj}:
\begin{corollary}[Same horizon] Let~\cref{assump:q_RKHS,assump:safeSet,assump:lipschitz,assump:sublinear,assump:dynamics_tol_terminal} hold.
Consider $\n^\star$ as in \cref{thm:maximum_dynamics_exploration}. With probability at least $1\!-\!\delta$, there exists $\nfin\leq\n^\star$ and a constant $K'>0$ such that policy $\pessiPolicy$, state $\state^{\mathrm{p}}$ obtained from \eqref{eq:pessi_obj} satisfies,
\begin{align}
    \Jobj[]{\state^{\mathrm{p}}, \dyntrueVec }{\pessiPolicy} \geq  \max_{\state^\star \in \safeInit{\nfin}, \pi^\star \in {\truePolicySet[, \epsilon]{c,\nfin}}(\safeInit{\nfin};\Horizon)} \Jobj[]{\state^\star, \dyntrueVec}{\pi^\star}    -  K'\epsguarantee \label{eq:samehorizon_result}
\end{align}
\end{corollary}
In contrast to \cref{obj:reward_maximization}, this result evaluates the policy obtained from the pessimistic problem \eqref{eq:pessi_obj} with respect to optimal policies defined over the same horizon $\Horizon$.

Once the termination criterion comparing the optimistic and pessimistic problems is satisfied, the result follows naturally: 
The pessimistic objective~\eqref{eq:pessi_obj} provides a Lipschitz-based worst-case lower bound on the true return, while the optimistic objective provides an upper bound. If these bounds are sufficiently close, the resulting pessimistic policy is
guaranteed to achieve performance close to the true optimum. In the next section, we leverage this corollary to derive a sample complexity lower bound.}


%% file: figure/dyn_receding.tex
 
\tikzset{
pattern size/.store in=\mcSize, 
pattern size = 5pt,
pattern thickness/.store in=\mcThickness, 
pattern thickness = 0.3pt,
pattern radius/.store in=\mcRadius, 
pattern radius = 1pt}
\makeatletter
\pgfutil@ifundefined{pgf@pattern@name@_2zimakovm}{
\pgfdeclarepatternformonly[\mcThickness,\mcSize]{_2zimakovm}
{\pgfqpoint{-\mcThickness}{-\mcThickness}}
{\pgfpoint{\mcSize}{\mcSize}}
{\pgfpoint{\mcSize}{\mcSize}}
{
\pgfsetcolor{\tikz@pattern@color}
\pgfsetlinewidth{\mcThickness}
\pgfpathmoveto{\pgfpointorigin}
\pgfpathlineto{\pgfpoint{0}{\mcSize}}
\pgfusepath{stroke}
}}
\makeatother

 
\tikzset{
pattern size/.store in=\mcSize, 
pattern size = 5pt,
pattern thickness/.store in=\mcThickness, 
pattern thickness = 0.3pt,
pattern radius/.store in=\mcRadius, 
pattern radius = 1pt}
\makeatletter
\pgfutil@ifundefined{pgf@pattern@name@_7wql7p9fc}{
\pgfdeclarepatternformonly[\mcThickness,\mcSize]{_7wql7p9fc}
{\pgfqpoint{-\mcThickness}{-\mcThickness}}
{\pgfpoint{\mcSize}{\mcSize}}
{\pgfpoint{\mcSize}{\mcSize}}
{
\pgfsetcolor{\tikz@pattern@color}
\pgfsetlinewidth{\mcThickness}
\pgfpathmoveto{\pgfpointorigin}
\pgfpathlineto{\pgfpoint{0}{\mcSize}}
\pgfusepath{stroke}
}}
\makeatother

 
\tikzset{
pattern size/.store in=\mcSize, 
pattern size = 5pt,
pattern thickness/.store in=\mcThickness, 
pattern thickness = 0.3pt,
pattern radius/.store in=\mcRadius, 
pattern radius = 1pt}
\makeatletter
\pgfutil@ifundefined{pgf@pattern@name@_meymhxv1w}{
\pgfdeclarepatternformonly[\mcThickness,\mcSize]{_meymhxv1w}
{\pgfqpoint{-\mcThickness}{-\mcThickness}}
{\pgfpoint{\mcSize}{\mcSize}}
{\pgfpoint{\mcSize}{\mcSize}}
{
\pgfsetcolor{\tikz@pattern@color}
\pgfsetlinewidth{\mcThickness}
\pgfpathmoveto{\pgfpointorigin}
\pgfpathlineto{\pgfpoint{0}{\mcSize}}
\pgfusepath{stroke}
}}
\makeatother

 
\tikzset{
pattern size/.store in=\mcSize, 
pattern size = 5pt,
pattern thickness/.store in=\mcThickness, 
pattern thickness = 0.3pt,
pattern radius/.store in=\mcRadius, 
pattern radius = 1pt}
\makeatletter
\pgfutil@ifundefined{pgf@pattern@name@_se6884mu5}{
\pgfdeclarepatternformonly[\mcThickness,\mcSize]{_se6884mu5}
{\pgfqpoint{-\mcThickness}{-\mcThickness}}
{\pgfpoint{\mcSize}{\mcSize}}
{\pgfpoint{\mcSize}{\mcSize}}
{
\pgfsetcolor{\tikz@pattern@color}
\pgfsetlinewidth{\mcThickness}
\pgfpathmoveto{\pgfpointorigin}
\pgfpathlineto{\pgfpoint{0}{\mcSize}}
\pgfusepath{stroke}
}}
\makeatother

 
\tikzset{
pattern size/.store in=\mcSize, 
pattern size = 5pt,
pattern thickness/.store in=\mcThickness, 
pattern thickness = 0.3pt,
pattern radius/.store in=\mcRadius, 
pattern radius = 1pt}
\makeatletter
\pgfutil@ifundefined{pgf@pattern@name@_9hhkptov0}{
\pgfdeclarepatternformonly[\mcThickness,\mcSize]{_9hhkptov0}
{\pgfqpoint{-\mcThickness}{-\mcThickness}}
{\pgfpoint{\mcSize}{\mcSize}}
{\pgfpoint{\mcSize}{\mcSize}}
{
\pgfsetcolor{\tikz@pattern@color}
\pgfsetlinewidth{\mcThickness}
\pgfpathmoveto{\pgfpointorigin}
\pgfpathlineto{\pgfpoint{0}{\mcSize}}
\pgfusepath{stroke}
}}
\makeatother

 
\tikzset{
pattern size/.store in=\mcSize, 
pattern size = 5pt,
pattern thickness/.store in=\mcThickness, 
pattern thickness = 0.3pt,
pattern radius/.store in=\mcRadius, 
pattern radius = 1pt}
\makeatletter
\pgfutil@ifundefined{pgf@pattern@name@_0v9kfk6ts}{
\pgfdeclarepatternformonly[\mcThickness,\mcSize]{_0v9kfk6ts}
{\pgfqpoint{-\mcThickness}{-\mcThickness}}
{\pgfpoint{\mcSize}{\mcSize}}
{\pgfpoint{\mcSize}{\mcSize}}
{
\pgfsetcolor{\tikz@pattern@color}
\pgfsetlinewidth{\mcThickness}
\pgfpathmoveto{\pgfpointorigin}
\pgfpathlineto{\pgfpoint{0}{\mcSize}}
\pgfusepath{stroke}
}}
\makeatother

\tikzset{every picture/.style={line width=0.75pt}} 

\begin{tikzpicture}[x=0.75pt,y=0.75pt,yscale=-1,xscale=1]

\draw  [pattern=_2zimakovm,pattern size=6pt,pattern thickness=0.75pt,pattern radius=0pt, pattern color={rgb, 255:red, 0; green, 0; blue, 0}] (757.74,719) .. controls (727.94,718) and (688.41,673.06) .. (677.97,664.39) .. controls (667.52,655.72) and (695.97,675.5) .. (719.97,680.17) .. controls (743.97,684.83) and (748,671.47) .. (748,659.25) .. controls (748,647.03) and (721,651.36) .. (730,624.25) .. controls (739,597.14) and (715.17,628.42) .. (696.5,639.75) .. controls (677.83,651.08) and (681.12,631.43) .. (677.52,630.83) .. controls (673.92,630.23) and (674.91,618.06) .. (667.13,608.95) .. controls (659.35,599.83) and (650.85,598.83) .. (650.63,593.28) .. controls (650.41,587.72) and (675.3,562.95) .. (705.52,554.34) .. controls (735.74,545.72) and (773.3,555.5) .. (789.08,570.61) .. controls (804.85,585.72) and (817.01,601.56) .. (816.85,636.56) .. controls (816.7,671.56) and (787.54,720) .. (757.74,719) -- cycle ;
\draw  [fill={rgb, 255:red, 128; green, 128; blue, 128 }  ,fill opacity=0.3 ] (515.2,492.94) .. controls (515.2,458.73) and (542.93,431) .. (577.14,431) -- (933.66,431) .. controls (967.87,431) and (995.6,458.73) .. (995.6,492.94) -- (995.6,678.76) .. controls (995.6,712.97) and (967.87,740.7) .. (933.66,740.7) -- (577.14,740.7) .. controls (542.93,740.7) and (515.2,712.97) .. (515.2,678.76) -- cycle ;
\draw  [fill={rgb, 255:red, 108; green, 215; blue, 108 }  ,fill opacity=0.6 ] (618.8,451) .. controls (679.5,457) and (858.3,446.8) .. (930,459.4) .. controls (1001.7,472) and (940,611.75) .. (957,690.25) .. controls (974,768.75) and (770.5,710) .. (687.5,719) .. controls (604.5,728) and (575.3,722) .. (564.5,692) .. controls (553.7,662) and (560.5,611) .. (558.5,558) .. controls (556.5,505) and (558.1,445) .. (618.8,451) -- cycle ;
\draw  [fill={rgb, 255:red, 80; green, 227; blue, 194 }  ,fill opacity=0.67 ] (603.44,660.82) .. controls (590.08,632.34) and (597.63,596.98) .. (620.31,581.85) .. controls (642.99,566.71) and (672.2,577.53) .. (685.56,606.01) .. controls (698.92,634.49) and (691.37,669.85) .. (668.69,684.99) .. controls (646.01,700.12) and (616.8,689.3) .. (603.44,660.82) -- cycle ;
\draw  [line width=4.5]  (675.43,664.39) .. controls (675.43,662.99) and (676.57,661.86) .. (677.97,661.86) .. controls (679.36,661.86) and (680.5,662.99) .. (680.5,664.39) .. controls (680.5,665.79) and (679.36,666.92) .. (677.97,666.92) .. controls (676.57,666.92) and (675.43,665.79) .. (675.43,664.39) -- cycle ;
\draw  [line width=3.75]  (740.22,694.97) .. controls (740.22,694.27) and (740.79,693.69) .. (741.5,693.69) .. controls (742.21,693.69) and (742.78,694.27) .. (742.78,694.97) .. controls (742.78,695.68) and (742.21,696.25) .. (741.5,696.25) .. controls (740.79,696.25) and (740.22,695.68) .. (740.22,694.97) -- cycle ;
\draw  [color={rgb, 255:red, 0; green, 0; blue, 0 }  ,draw opacity=1 ][pattern=_7wql7p9fc,pattern size=4.574999999999999pt,pattern thickness=0.75pt,pattern radius=0pt, pattern color={rgb, 255:red, 0; green, 0; blue, 0}] (842,708.97) .. controls (779,733.97) and (755,698.94) .. (741.5,694.97) .. controls (728,691) and (808,686.25) .. (818.5,658.25) .. controls (829,630.25) and (810.5,616.97) .. (802,588.47) .. controls (793.5,559.97) and (765,577.47) .. (739.5,567.47) .. controls (714,557.47) and (705.17,609.14) .. (686.5,620.47) .. controls (667.83,631.8) and (681.94,618) .. (676,613.97) .. controls (670.06,609.94) and (677.28,612.08) .. (669.5,602.97) .. controls (661.72,593.86) and (644.22,591.53) .. (644,585.97) .. controls (643.78,580.42) and (671.5,536.47) .. (698.5,522.97) .. controls (725.5,509.47) and (757,506.47) .. (785.5,513.47) .. controls (814,520.47) and (839,531.47) .. (860,573.47) .. controls (881,615.47) and (905,683.97) .. (842,708.97) -- cycle ;
\draw  [color={rgb, 255:red, 0; green, 0; blue, 0 }  ,draw opacity=1 ][pattern=_meymhxv1w,pattern size=2.7pt,pattern thickness=0.75pt,pattern radius=0pt, pattern color={rgb, 255:red, 0; green, 0; blue, 0}] (930,627.75) .. controls (903,660.75) and (858.5,677.25) .. (834,685.75) .. controls (809.5,694.25) and (906,643.25) .. (901,587.75) .. controls (896,532.25) and (841,528.25) .. (801.5,524.75) .. controls (762,521.25) and (743,532.25) .. (717,551.25) .. controls (691,570.25) and (672.5,601.75) .. (663,611.25) .. controls (653.5,620.75) and (646.5,604.25) .. (637.5,603.25) .. controls (628.5,602.25) and (624,600.75) .. (615.5,595.75) .. controls (607,590.75) and (638.5,535.75) .. (686,506.25) .. controls (733.5,476.75) and (782.5,473.25) .. (818,474.75) .. controls (853.5,476.25) and (888.5,480.75) .. (916.5,512.25) .. controls (944.5,543.75) and (957,594.75) .. (930,627.75) -- cycle ;
\draw  [line width=3.75]  (832.22,684.78) .. controls (832.22,684.07) and (832.79,683.5) .. (833.5,683.5) .. controls (834.21,683.5) and (834.78,684.07) .. (834.78,684.78) .. controls (834.78,685.48) and (834.21,686.06) .. (833.5,686.06) .. controls (832.79,686.06) and (832.22,685.48) .. (832.22,684.78) -- cycle ;
\draw  [line width=2.25]  (691.44,674.39) .. controls (691.44,673.68) and (692.01,673.11) .. (692.72,673.11) .. controls (693.42,673.11) and (693.99,673.68) .. (693.99,674.39) .. controls (693.99,675.1) and (693.42,675.67) .. (692.72,675.67) .. controls (692.01,675.67) and (691.44,675.1) .. (691.44,674.39) -- cycle ;
\draw  [line width=2.25]  (705.19,682.39) .. controls (705.19,681.68) and (705.76,681.11) .. (706.47,681.11) .. controls (707.17,681.11) and (707.74,681.68) .. (707.74,682.39) .. controls (707.74,683.1) and (707.17,683.67) .. (706.47,683.67) .. controls (705.76,683.67) and (705.19,683.1) .. (705.19,682.39) -- cycle ;
\draw  [line width=2.25]  (720.94,688.89) .. controls (720.94,688.18) and (721.51,687.61) .. (722.22,687.61) .. controls (722.92,687.61) and (723.49,688.18) .. (723.49,688.89) .. controls (723.49,689.6) and (722.92,690.17) .. (722.22,690.17) .. controls (721.51,690.17) and (720.94,689.6) .. (720.94,688.89) -- cycle ;
\draw  [line width=4.5]  (738.97,694.97) .. controls (738.97,693.57) and (740.1,692.44) .. (741.5,692.44) .. controls (742.9,692.44) and (744.03,693.57) .. (744.03,694.97) .. controls (744.03,696.37) and (742.9,697.5) .. (741.5,697.5) .. controls (740.1,697.5) and (738.97,696.37) .. (738.97,694.97) -- cycle ;
\draw  [line width=4.5]  (830.97,684.78) .. controls (830.97,683.38) and (832.1,682.25) .. (833.5,682.25) .. controls (834.9,682.25) and (836.03,683.38) .. (836.03,684.78) .. controls (836.03,686.18) and (834.9,687.31) .. (833.5,687.31) .. controls (832.1,687.31) and (830.97,686.18) .. (830.97,684.78) -- cycle ;
\draw  [line width=4.5]  (912.22,557.47) .. controls (912.22,556.07) and (913.35,554.94) .. (914.75,554.94) .. controls (916.15,554.94) and (917.28,556.07) .. (917.28,557.47) .. controls (917.28,558.87) and (916.15,560) .. (914.75,560) .. controls (913.35,560) and (912.22,558.87) .. (912.22,557.47) -- cycle ;
\draw  [line width=2.25]  (764.69,697.64) .. controls (764.69,696.93) and (765.26,696.36) .. (765.97,696.36) .. controls (766.67,696.36) and (767.24,696.93) .. (767.24,697.64) .. controls (767.24,698.35) and (766.67,698.92) .. (765.97,698.92) .. controls (765.26,698.92) and (764.69,698.35) .. (764.69,697.64) -- cycle ;
\draw  [line width=2.25]  (785.44,696.64) .. controls (785.44,695.93) and (786.01,695.36) .. (786.72,695.36) .. controls (787.42,695.36) and (787.99,695.93) .. (787.99,696.64) .. controls (787.99,697.35) and (787.42,697.92) .. (786.72,697.92) .. controls (786.01,697.92) and (785.44,697.35) .. (785.44,696.64) -- cycle ;
\draw  [line width=2.25]  (809.44,692.14) .. controls (809.44,691.43) and (810.01,690.86) .. (810.72,690.86) .. controls (811.42,690.86) and (811.99,691.43) .. (811.99,692.14) .. controls (811.99,692.85) and (811.42,693.42) .. (810.72,693.42) .. controls (810.01,693.42) and (809.44,692.85) .. (809.44,692.14) -- cycle ;
\draw  [line width=2.25]  (858.94,668.89) .. controls (858.94,668.18) and (859.51,667.61) .. (860.22,667.61) .. controls (860.92,667.61) and (861.49,668.18) .. (861.49,668.89) .. controls (861.49,669.6) and (860.92,670.17) .. (860.22,670.17) .. controls (859.51,670.17) and (858.94,669.6) .. (858.94,668.89) -- cycle ;
\draw  [line width=2.25]  (881.44,649.64) .. controls (881.44,648.93) and (882.01,648.36) .. (882.72,648.36) .. controls (883.42,648.36) and (883.99,648.93) .. (883.99,649.64) .. controls (883.99,650.35) and (883.42,650.92) .. (882.72,650.92) .. controls (882.01,650.92) and (881.44,650.35) .. (881.44,649.64) -- cycle ;
\draw  [line width=2.25]  (897.69,629.39) .. controls (897.69,628.68) and (898.26,628.11) .. (898.97,628.11) .. controls (899.67,628.11) and (900.24,628.68) .. (900.24,629.39) .. controls (900.24,630.1) and (899.67,630.67) .. (898.97,630.67) .. controls (898.26,630.67) and (897.69,630.1) .. (897.69,629.39) -- cycle ;
\draw  [line width=2.25]  (909.69,607.39) .. controls (909.69,606.68) and (910.26,606.11) .. (910.97,606.11) .. controls (911.67,606.11) and (912.24,606.68) .. (912.24,607.39) .. controls (912.24,608.1) and (911.67,608.67) .. (910.97,608.67) .. controls (910.26,608.67) and (909.69,608.1) .. (909.69,607.39) -- cycle ;
\draw  [line width=2.25]  (915.44,583.39) .. controls (915.44,582.68) and (916.02,582.11) .. (916.72,582.11) .. controls (917.43,582.11) and (918,582.68) .. (918,583.39) .. controls (918,584.1) and (917.43,584.67) .. (916.72,584.67) .. controls (916.02,584.67) and (915.44,584.1) .. (915.44,583.39) -- cycle ;
\draw  [color={rgb, 255:red, 0; green, 0; blue, 0 }  ,draw opacity=1 ][pattern=_0v9kfk6ts,pattern size=0.1750000000000001pt,pattern thickness=0.75pt,pattern radius=0pt, pattern color={rgb, 255:red, 0; green, 0; blue, 0},draw opacity=0.15] (856.8,480.2) .. controls (897.6,503) and (912,548.2) .. (914.75,557.47) .. controls (917.5,566.74) and (880.4,493.4) .. (825.6,481.4) .. controls (770.8,469.4) and (733.2,475.8) .. (709.2,492.6) .. controls (685.2,509.4) and (670,534.2) .. (665.6,545.4) .. controls (661.2,556.6) and (642.04,605.62) .. (643.2,619) .. controls (644.36,632.38) and (616.9,610.91) .. (607.6,614.2) .. controls (598.3,617.49) and (632.8,549.8) .. (640,534.2) .. controls (647.2,518.6) and (678.8,455) .. (740,453) .. controls (801.2,451) and (816,457.4) .. (856.8,480.2) -- cycle ;
\draw [color={rgb, 255:red, 0; green, 103; blue, 255 }  ,draw opacity=1 ][line width=2.25]  [dash pattern={on 6.75pt off 4.5pt}]  (680.8,649) .. controls (726.4,670.6) and (825.6,698.2) .. (856.8,667.8) .. controls (888,637.4) and (916.4,618.2) .. (913.2,566.2) .. controls (910,514.2) and (835.6,467.8) .. (787.2,466.6) .. controls (738.8,465.4) and (716.8,470.6) .. (691.2,488.2) .. controls (666.37,505.27) and (627.61,565.63) .. (636.99,613.04) ;
\draw [shift={(638,617.4)}, rotate = 255.07] [fill={rgb, 255:red, 74; green, 144; blue, 226 }  ,fill opacity=1 ][line width=0.08]  [draw opacity=0] (14.29,-6.86) -- (0,0) -- (14.29,6.86) -- cycle    ;
\draw [shift={(777.29,678.4)}, rotate = 189.57] [fill={rgb, 255:red, 74; green, 144; blue, 226 }  ,fill opacity=1 ][line width=0.08]  [draw opacity=0] (14.29,-6.86) -- (0,0) -- (14.29,6.86) -- cycle    ;
\draw [shift={(903.45,616.16)}, rotate = 120.46] [fill={rgb, 255:red, 74; green, 144; blue, 226 }  ,fill opacity=1 ][line width=0.08]  [draw opacity=0] (14.29,-6.86) -- (0,0) -- (14.29,6.86) -- cycle    ;
\draw [shift={(859.48,492.07)}, rotate = 33.92] [fill={rgb, 255:red, 74; green, 144; blue, 226 }  ,fill opacity=1 ][line width=0.08]  [draw opacity=0] (14.29,-6.86) -- (0,0) -- (14.29,6.86) -- cycle    ;
\draw [shift={(728.46,470.87)}, rotate = 347.54] [fill={rgb, 255:red, 74; green, 144; blue, 226 }  ,fill opacity=1 ][line width=0.08]  [draw opacity=0] (14.29,-6.86) -- (0,0) -- (14.29,6.86) -- cycle    ;
\draw [shift={(645.16,552.8)}, rotate = 292.84] [fill={rgb, 255:red, 74; green, 144; blue, 226 }  ,fill opacity=1 ][line width=0.08]  [draw opacity=0] (14.29,-6.86) -- (0,0) -- (14.29,6.86) -- cycle    ;
\draw  [line width=2.25]  (893.84,520.59) .. controls (893.84,519.88) and (894.42,519.31) .. (895.12,519.31) .. controls (895.83,519.31) and (896.4,519.88) .. (896.4,520.59) .. controls (896.4,521.3) and (895.83,521.87) .. (895.12,521.87) .. controls (894.42,521.87) and (893.84,521.3) .. (893.84,520.59) -- cycle ;
\draw  [line width=2.25]  (859.04,489.19) .. controls (859.04,488.48) and (859.62,487.91) .. (860.32,487.91) .. controls (861.03,487.91) and (861.6,488.48) .. (861.6,489.19) .. controls (861.6,489.9) and (861.03,490.47) .. (860.32,490.47) .. controls (859.62,490.47) and (859.04,489.9) .. (859.04,489.19) -- cycle ;
\draw  [line width=2.25]  (810.04,467.59) .. controls (810.04,466.88) and (810.62,466.31) .. (811.32,466.31) .. controls (812.03,466.31) and (812.6,466.88) .. (812.6,467.59) .. controls (812.6,468.3) and (812.03,468.87) .. (811.32,468.87) .. controls (810.62,468.87) and (810.04,468.3) .. (810.04,467.59) -- cycle ;
\draw  [line width=2.25]  (762.04,463.59) .. controls (762.04,462.88) and (762.62,462.31) .. (763.32,462.31) .. controls (764.03,462.31) and (764.6,462.88) .. (764.6,463.59) .. controls (764.6,464.3) and (764.03,464.87) .. (763.32,464.87) .. controls (762.62,464.87) and (762.04,464.3) .. (762.04,463.59) -- cycle ;
\draw [line width=1.5]    (677.97,664.39) .. controls (697.25,677.77) and (710.71,686.91) .. (738.01,694.08) ;
\draw [shift={(741.5,694.97)}, rotate = 193.84] [fill={rgb, 255:red, 0; green, 0; blue, 0 }  ][line width=0.08]  [draw opacity=0] (11.61,-5.58) -- (0,0) -- (11.61,5.58) -- cycle    ;
\draw   (750.44,710) -- (750.74,709.67) -- (753.8,712.43) -- (756.56,709.36) -- (756.97,709.74) -- (754.22,712.8) -- (757.28,715.56) -- (756.98,715.89) -- (753.92,713.13) -- (751.16,716.19) -- (750.75,715.82) -- (753.5,712.76) -- cycle ;
\draw [line width=1.5]    (741.5,694.97) .. controls (772.95,701.26) and (797.14,695.53) .. (829.9,685.85) ;
\draw [shift={(833.5,684.78)}, rotate = 163.41] [fill={rgb, 255:red, 0; green, 0; blue, 0 }  ][line width=0.08]  [draw opacity=0] (11.61,-5.58) -- (0,0) -- (11.61,5.58) -- cycle    ;
\draw   (822.44,656.8) -- (822.74,656.47) -- (825.8,659.23) -- (828.56,656.16) -- (828.97,656.54) -- (826.22,659.6) -- (829.28,662.36) -- (828.98,662.69) -- (825.92,659.93) -- (823.16,662.99) -- (822.75,662.62) -- (825.5,659.56) -- cycle ;
\draw [line width=1.5]    (833.5,684.78) .. controls (870.2,670.5) and (927.76,614.19) .. (915.58,560.74) ;
\draw [shift={(914.75,557.47)}, rotate = 74.22] [fill={rgb, 255:red, 0; green, 0; blue, 0 }  ][line width=0.08]  [draw opacity=0] (11.61,-5.58) -- (0,0) -- (11.61,5.58) -- cycle    ;
\draw   (937.44,572.8) -- (937.74,572.47) -- (940.8,575.23) -- (943.56,572.16) -- (943.97,572.54) -- (941.22,575.6) -- (944.28,578.36) -- (943.98,578.69) -- (940.92,575.93) -- (938.16,578.99) -- (937.75,578.62) -- (940.5,575.56) -- cycle ;
\draw [line width=1.5]    (677.97,664.39) .. controls (739.23,711.07) and (793,698.25) .. (829.5,686.25) .. controls (866,674.25) and (951.5,600.75) .. (902.5,530.75) .. controls (855.21,463.2) and (783.25,460.7) .. (747.36,464.56) ;
\draw [shift={(743.6,465)}, rotate = 352.71] [fill={rgb, 255:red, 0; green, 0; blue, 0 }  ][line width=0.08]  [draw opacity=0] (11.61,-5.58) -- (0,0) -- (11.61,5.58) -- cycle    ;
\draw (607.13,629.33) node [anchor=north west][inner sep=0.75pt]  [font=\huge]  {$\mathbb{X}_{n}$};
\draw (521.43,466.47) node [anchor=north west][inner sep=0.75pt]  [font=\Huge]  {$\mathcal{X}$};
\draw (639.38,660.08) node [anchor=north west][inner sep=0.75pt]  [font=\huge]  {$x_{s}$};
\draw (615.13,510.33) node [anchor=north west][inner sep=0.75pt]  [font=\huge,color={rgb, 255:red, 0; green, 103; blue, 255 }  ,opacity=1 ]  {$\pi^{\star}$};
\draw (690.88,705.08) node [anchor=north west][inner sep=0.75pt]  [font=\huge]  {$x( k)$};

\end{tikzpicture}

%% file: figure/dyn_receding_optimality.tex
\tikzset{every picture/.style={line width=0.75pt}} 

\begin{tikzpicture}[x=0.75pt,y=0.75pt,yscale=-1,xscale=1]

\draw  [fill={rgb, 255:red, 128; green, 128; blue, 128 }  ,fill opacity=0.3 ] (1021.2,488.94) .. controls (1021.2,454.73) and (1048.93,427) .. (1083.14,427) -- (1439.66,427) .. controls (1473.87,427) and (1501.6,454.73) .. (1501.6,488.94) -- (1501.6,674.76) .. controls (1501.6,708.97) and (1473.87,736.7) .. (1439.66,736.7) -- (1083.14,736.7) .. controls (1048.93,736.7) and (1021.2,708.97) .. (1021.2,674.76) -- cycle ;
\draw  [fill={rgb, 255:red, 108; green, 215; blue, 108 }  ,fill opacity=0.6 ] (1124.8,447) .. controls (1185.5,453) and (1364.3,442.8) .. (1436,455.4) .. controls (1507.7,468) and (1446,607.75) .. (1463,686.25) .. controls (1480,764.75) and (1276.5,706) .. (1193.5,715) .. controls (1110.5,724) and (1081.3,718) .. (1070.5,688) .. controls (1059.7,658) and (1066.5,607) .. (1064.5,554) .. controls (1062.5,501) and (1064.1,441) .. (1124.8,447) -- cycle ;
\draw  [fill={rgb, 255:red, 108; green, 215; blue, 108 }  ,fill opacity=0.6 ] (1140,458.67) .. controls (1200.7,464.67) and (1349.63,457.4) .. (1421.33,470) .. controls (1493.03,482.6) and (1437.67,597.5) .. (1454.67,676) .. controls (1471.67,754.5) and (1285.67,695.67) .. (1202.67,704.67) .. controls (1119.67,713.67) and (1092.8,704) .. (1082,674) .. controls (1071.2,644) and (1082,608.33) .. (1080,555.33) .. controls (1078,502.33) and (1079.3,452.67) .. (1140,458.67) -- cycle ;
\draw  [fill={rgb, 255:red, 80; green, 227; blue, 194 }  ,fill opacity=0.67 ] (1109.44,656.82) .. controls (1096.08,628.34) and (1103.63,592.98) .. (1126.31,577.85) .. controls (1148.99,562.71) and (1178.2,573.53) .. (1191.56,602.01) .. controls (1204.92,630.49) and (1197.37,665.85) .. (1174.69,680.99) .. controls (1152.01,696.12) and (1122.8,685.3) .. (1109.44,656.82) -- cycle ;
\draw [color={rgb, 255:red, 0; green, 103; blue, 255 }  ,draw opacity=1 ][line width=2.25]  [dash pattern={on 6.75pt off 4.5pt}]  (1193.5,641) .. controls (1262.73,665.61) and (1325.7,691.3) .. (1358.25,663.5) .. controls (1390.8,635.7) and (1417.45,610) .. (1414.25,558) .. controls (1411.05,506) and (1347,464.5) .. (1288,463.5) .. controls (1229,462.5) and (1217.6,469.4) .. (1192,487) .. controls (1167.17,504.07) and (1131.42,562.54) .. (1140.98,609.84) ;
\draw [shift={(1142,614.2)}, rotate = 255.07] [fill={rgb, 255:red, 0; green, 103; blue, 255 }  ,fill opacity=1 ][line width=0.08]  [draw opacity=0] (14.29,-6.86) -- (0,0) -- (14.29,6.86) -- cycle    ;
\draw [shift={(1284.46,670.75)}, rotate = 193.81] [fill={rgb, 255:red, 0; green, 103; blue, 255 }  ,fill opacity=1 ][line width=0.08]  [draw opacity=0] (14.29,-6.86) -- (0,0) -- (14.29,6.86) -- cycle    ;
\draw [shift={(1405.04,610.13)}, rotate = 117.51] [fill={rgb, 255:red, 0; green, 103; blue, 255 }  ,fill opacity=1 ][line width=0.08]  [draw opacity=0] (14.29,-6.86) -- (0,0) -- (14.29,6.86) -- cycle    ;
\draw [shift={(1360.9,484.42)}, rotate = 31.49] [fill={rgb, 255:red, 0; green, 103; blue, 255 }  ,fill opacity=1 ][line width=0.08]  [draw opacity=0] (14.29,-6.86) -- (0,0) -- (14.29,6.86) -- cycle    ;
\draw [shift={(1229.11,467.83)}, rotate = 347.03] [fill={rgb, 255:red, 0; green, 103; blue, 255 }  ,fill opacity=1 ][line width=0.08]  [draw opacity=0] (14.29,-6.86) -- (0,0) -- (14.29,6.86) -- cycle    ;
\draw [shift={(1147.93,550.37)}, rotate = 291.67] [fill={rgb, 255:red, 0; green, 103; blue, 255 }  ,fill opacity=1 ][line width=0.08]  [draw opacity=0] (14.29,-6.86) -- (0,0) -- (14.29,6.86) -- cycle    ;

\draw [color={rgb, 255:red, 0; green, 0; blue, 0 }  ,draw opacity=1 ][line width=1.5]    (1191.2,631.8) .. controls (1258.08,662.22) and (1327.48,686.52) .. (1358.68,656.12) .. controls (1389.88,625.72) and (1411.88,611.12) .. (1408.68,559.12) .. controls (1405.48,507.12) and (1342.08,457.32) .. (1293.68,456.12) .. controls (1245.28,454.92) and (1220.28,474.02) .. (1194.68,491.62) .. controls (1169.72,508.78) and (1133.09,562.62) .. (1143.13,609.79) ;
\draw [shift={(1144,613.4)}, rotate = 255.07] [fill={rgb, 255:red, 0; green, 0; blue, 0 }  ,fill opacity=1 ][line width=0.08]  [draw opacity=0] (11.61,-5.58) -- (0,0) -- (11.61,5.58) -- cycle    ;
\draw [shift={(1281.13,665.17)}, rotate = 194.11] [fill={rgb, 255:red, 0; green, 0; blue, 0 }  ,fill opacity=1 ][line width=0.08]  [draw opacity=0] (11.61,-5.58) -- (0,0) -- (11.61,5.58) -- cycle    ;
\draw [shift={(1401.01,607.89)}, rotate = 118.36] [fill={rgb, 255:red, 0; green, 0; blue, 0 }  ,fill opacity=1 ][line width=0.08]  [draw opacity=0] (11.61,-5.58) -- (0,0) -- (11.61,5.58) -- cycle    ;
\draw [shift={(1362.67,484.11)}, rotate = 39.1] [fill={rgb, 255:red, 0; green, 0; blue, 0 }  ,fill opacity=1 ][line width=0.08]  [draw opacity=0] (11.61,-5.58) -- (0,0) -- (11.61,5.58) -- cycle    ;
\draw [shift={(1234.24,467.44)}, rotate = 336.93] [fill={rgb, 255:red, 0; green, 0; blue, 0 }  ,fill opacity=1 ][line width=0.08]  [draw opacity=0] (11.61,-5.58) -- (0,0) -- (11.61,5.58) -- cycle    ;
\draw [shift={(1150.57,550.71)}, rotate = 293.56] [fill={rgb, 255:red, 0; green, 0; blue, 0 }  ,fill opacity=1 ][line width=0.08]  [draw opacity=0] (11.61,-5.58) -- (0,0) -- (11.61,5.58) -- cycle    ;
\draw [line width=1.5]    (1149.83,616.22) .. controls (1152.05,615.31) and (1153.65,615.92) .. (1154.63,618.07) .. controls (1155.66,620.19) and (1157.23,620.72) .. (1159.34,619.67) .. controls (1161.45,618.6) and (1163.02,619.11) .. (1164.05,621.19) .. controls (1165.14,623.3) and (1166.72,623.82) .. (1168.79,622.75) .. controls (1170.99,621.74) and (1172.61,622.31) .. (1173.66,624.46) .. controls (1174.61,626.6) and (1176.13,627.18) .. (1178.21,626.2) -- (1180.55,627.14) -- (1187.83,630.26) ;
\draw [shift={(1191.2,631.8)}, rotate = 204.91] [fill={rgb, 255:red, 0; green, 0; blue, 0 }  ][line width=0.08]  [draw opacity=0] (11.61,-5.58) -- (0,0) -- (11.61,5.58) -- cycle    ;
\draw  [line width=3.75]  (1147.3,616.22) .. controls (1147.3,614.82) and (1148.44,613.69) .. (1149.83,613.69) .. controls (1151.23,613.69) and (1152.37,614.82) .. (1152.37,616.22) .. controls (1152.37,617.62) and (1151.23,618.75) .. (1149.83,618.75) .. controls (1148.44,618.75) and (1147.3,617.62) .. (1147.3,616.22) -- cycle ;
\draw  [line width=3.75]  (1188.67,631.8) .. controls (1188.67,630.4) and (1189.8,629.27) .. (1191.2,629.27) .. controls (1192.6,629.27) and (1193.73,630.4) .. (1193.73,631.8) .. controls (1193.73,633.2) and (1192.6,634.33) .. (1191.2,634.33) .. controls (1189.8,634.33) and (1188.67,633.2) .. (1188.67,631.8) -- cycle ;
\draw [line width=1.5]    (1403,451.25) -- (1403.5,469.25) ;
\draw [shift={(1403.5,469.25)}, rotate = 268.41] [color={rgb, 255:red, 0; green, 0; blue, 0 }  ][line width=1.5]    (0,2.91) -- (0,-2.91)(5.68,-2.55) .. controls (3.61,-1.2) and (1.72,-0.35) .. (0,0) .. controls (1.72,0.35) and (3.61,1.2) .. (5.68,2.55)   ;
\draw [shift={(1403,451.25)}, rotate = 88.41] [color={rgb, 255:red, 0; green, 0; blue, 0 }  ][line width=1.5]    (0,2.91) -- (0,-2.91)(5.68,-1.71) .. controls (3.61,-0.72) and (1.72,-0.15) .. (0,0) .. controls (1.72,0.15) and (3.61,0.72) .. (5.68,1.71)   ;

\draw (1132.43,659.43) node [anchor=north west][inner sep=0.75pt]  [font=\huge]  {$\mathbb{X}_{n}$};
\draw (1027.43,462.47) node [anchor=north west][inner sep=0.75pt]  [font=\Huge]  {$\mathcal{X}$};
\draw (1101.08,617.08) node [anchor=north west][inner sep=0.75pt]  [font=\huge]  {$x( k)$};
\draw (1410,439.63) node [anchor=north west][inner sep=0.75pt]  [font=\huge]  {$\epsilon $};
\draw (1116.88,506.82) node [anchor=north west][inner sep=0.75pt]  [font=\huge,color={rgb, 255:red, 0; green, 103; blue, 255 }  ,opacity=1 ,rotate=-359.36]  {$\pi ^{\star }$};
\draw (1354.77,600.33) node [anchor=north west][inner sep=0.75pt]  [font=\huge]  {$\pi ^{\mathrm{p}}$};
\draw (1157.57,588.83) node [anchor=north west][inner sep=0.75pt]  [font=\huge]  {$\hat{\pi }$};
\draw (1166.28,643.68) node [anchor=north west][inner sep=0.75pt]  [font=\huge]  {$x^{\mathrm{p}}$};

\end{tikzpicture}

%% file: sections/4.1_sample_complexity_lower_bound.tex
\change{
\section{Lower bound on sample complexity for dynamics exploration}
\looseness -1 
This section presents a lower bound on the sample complexity of the dynamics exploration problem, quantifying the minimum number of samples required to guarantee near-optimality.
We establish these results by reducing our problem to known lower bounds on the simple regret in noisy Gaussian process bandit optimization~\citep{scarlett2017lower}. We show that even for a simple special case of state-independent dynamics, it is hard to solve it with better rates.
We focus on two widely used kernels: the squared exponential (SE) and the Mat\'ern kernels. 

\begin{restatable}[Sample complexity lower bound]{theorm}{restatelowerbound} Fix $\epsilon \in (0, \frac{1}{2})$, $\delta \in (0, \frac{1}{8})$ and $\nfin \in \R$. Suppose there exists an algorithm for system~\eqref{eq:system_dyn} where \mbox{$\dyntrue_i \in \RKHS[i]$} with \mbox{$\|\dyntrue_i\|_{\RKHS[i]} \leq B $}, $i \in \Intrange{1}{\statedim}$ and noise $\eta$ being bounded and conditionally $\sigma$-sub-Gaussian, that guarantees with probability at least $1-\delta$ that after at most $\nfin$ samples, \label{thm:lower_bound}
\begin{align}
    \Jobj[]{\state^{\mathrm{p}}, \dyntrueVec }{\pessiPolicy} \geq  \max_{\state^\star \in \safeInit{\nfin}, \pi^\star \in {\truePolicySet[, \epsilon]{c,\nfin}}(\safeInit{\nfin};\Horizon)} \Jobj[]{\state^\star, \dyntrueVec}{\pi^\star}    -  K'\epsguarantee'.  \label{eq:lb_objective}
\end{align}
Then, the following holds with $ \epsilon = \epsilon'/r > 0$ where $r = (1-8\delta)/(1-2\delta)$ and provided that $\epsilon/B$ is sufficiently small.
\begin{itemize}
    \item For the squared exponential kernel, it is necessary that
    \begin{align*}
        \nfin = \Omega\left(\frac{\sigma^2}{\epsilon^2} \left ( \log \frac{B}{\epsilon}\right)^{\statedim/2} \right).
    \end{align*}
    \item For the Mat\'ern kernel, it is necessary that
    \begin{align*}
        \nfin = \Omega\left(\frac{\sigma^2}{\epsilon^2} \left (\frac{B}{\epsilon}\right)^{\statedim/\nu} \right).
    \end{align*}
\end{itemize}
\end{restatable}

\mypar{Proof Sketch} 
We reduce our problem to known lower bounds for simple regret in the noisy Gaussian process bandit optimization setting \citep{scarlett2017lower}. 
The results in~\citet{scarlett2017lower} assume Gaussian noise. 
By applying an algorithm for bounded noise to the GP bandit setting with Gaussian noise and removing the instances where the noise exceeds the truncation bound (failure event) in probability, we extend the lower bound to a truncated Gaussian noise distribution. 
The reduction of the dynamics exploration problem is as follows: consider a special instance with state-independent dynamics, $\state_1 = \dyntrueVec(\coninput_0) + \noise_n$, and reward function $r(\state) = \state$. Similar to the bandit setting, this construction ensures that taking an action yields a noisy dynamics observation as a reward.
Since \dynExplor solves this setup within $\n^\star$ samples, the lower bound for Gaussian process bandit optimization also applies to dynamics exploration setup with probability $1-\delta$. The complete proof is in \cref{apxsex:lowerbound}.


For the squared exponential kernel, \dynExplor guarantees optimality within $\n^\star = \mathcal{O}^\star\left(\frac{1}{\epsilon^2} \left(\log \frac{1}{\epsilon}\right)^{2\statedim} \right)$ samples which matches the lower bound in \cref{thm:lower_bound} up to a factor 2 versus 1/2 in the exponent. In contrast, for the Mat\'ern kernel, the gap is larger; see~\citet{scarlett2017lower} for a discussion. 
These lower-bound results indicate that, for certain classes of kernel functions, the exploration strategy of \dynExplor is close to optimal. However, note that this result characterizes worst-case optimality; tighter upper bounds may still be achievable for specific subclasses of problems.

}

%% file: sections/5_experiments.tex
\section{Experiments}
\label{sec:experiments}

\looseness -1 
We empirically evaluate the performance of \dynExplor in multiple environments. Here we present experiments on two representative tasks: car racing and drone navigation under aerodynamic effects. 
We compare \dynExplor against two baselines: 
i) No learning, which attempts to solve the task using only prior data without gathering new information (measurements),
ii) Two-stage, which first explores the dynamics by gathering online data (similar to \cref{alg:full_domain_exploration_basic}) and then optimizes the objective with the learned dynamics. 
This baseline is akin to \citet{as2024actsafe}, which pessimistically explores and then optimizes, although the considered baseline requires no resets.
Additionally, we compare performance against a clairvoyant agent that has exact knowledge of the dynamics, representing the optimal performance. 

\looseness -1
We implement the optimization problems as described in ~\cref{rem:implementation} with a sampling-based approach. The implementation is based on \texttt{Python} interfaces of \texttt{acados}~\citep{verschueren_acadosmodular_2022} and~\texttt{CasADi}~\citep{andersson_casadi_2019}; for efficient GP sampling, we employ~\texttt{GPyTorch}~\citep{gardner_gpytorch_2018}. For each environment, we report the average with a confidence bound of 1 standard deviation computed on 10 different random seeds. 
The computation of control input is fast (real-time feasible), for instance, on an i7-11800H processor (2.30GHz) with RTX A2000, the computation of a single control sequence took 32.0 $\pm$ 7.7 ms for the car environment. 
While this section highlights the core experiments, additional results and implementation details are provided in~\cref{apxsec:experiments}.
Below we explain our main observations for each environment:

\begin{figure}
    \begin{subfigure}[t]{0.52\columnwidth}
      \centering
  	\includegraphics[scale=0.63]{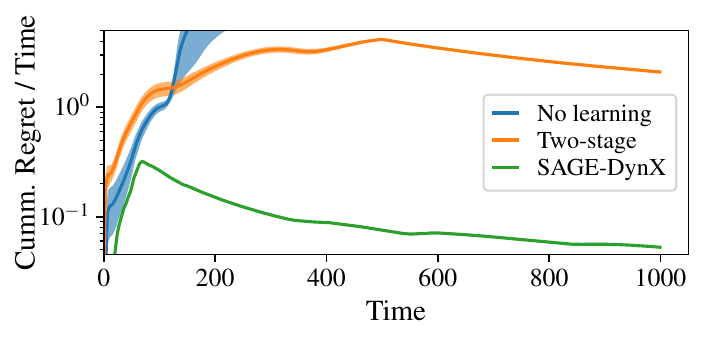}
    \caption{Drone navigation}
    \label{fig:drone}
    \end{subfigure}
~
    \begin{subfigure}[t]{0.52\columnwidth}
  	\centering
  	\includegraphics[scale=0.58]{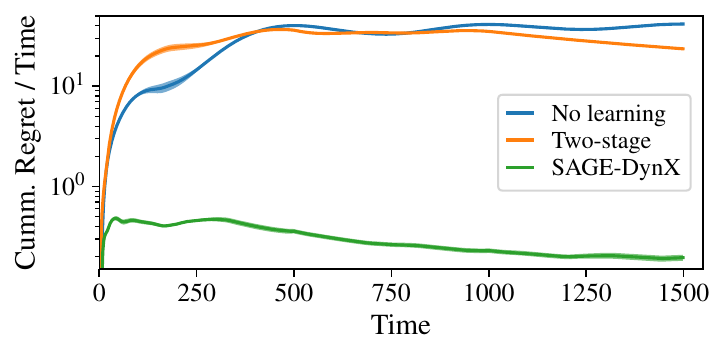}
    \caption{Car Racing}
    \label{fig:car_racing}
    \end{subfigure}
\caption{
\looseness -1
Cumulative regret over time (averaged across runs) in different environments. \dynExplor achieves an order of magnitude lower regret compared to the baselines in both experiments.
} 
\end{figure}

\looseness -1
\mypar{Drone navigation} The dynamics model is a 6-dimensional nonlinear system, representing position, angular orientation, linear velocities, and angular rates from ~\citep{singh2023robust,sasfi2023robust}. 
The control inputs are two-dimensional, corresponding to the thrust force produced by the propellers. The dynamics also include disturbances corresponding to aerodynamic effects.
The task requires the drone to track a heart-shaped trajectory 
while satisfying box constraints in the state space. The rewards correspond to the accuracy of tracking a time-varying heart-shaped reference trajectory.
\cref{fig:drone} shows the performance of \dynExplor compared to the two baselines.
We plot cumulative regret over time, where regret is computed with the position difference from the clairvoyant system at any time step. 
The no-learning baseline, due to high model uncertainty, quickly fails to maintain accurate tracking. 
For the two stage baseline, we explore dynamics for the first 500 time steps and then maximize the rewards. 
We observe a high initial regret due to random exploration of the dynamics, and then it starts to decay after 500 steps. In contrast, \dynExplor initially deviates in the beginning to gather data and then quickly comes closer to the optimal behaviour, as evident by the steady decline in the regret shown in \cref{fig:drone}.

\looseness -1
\textbf{Car racing} is an interesting application in which, while racing, the model can be learned online. 
The car dynamics model is given by a 4D nonlinear system representing position, orientation, and velocity from~\citep{7225830}.  
The control commands are two-dimensional, representing throttle and steering.  
The task requires the car to follow a reference path while satisfying the track constraints. Similarly to the drone task, the car tries to track a reference, and the rewards are assigned as per the accuracy of tracking. 
\cref{fig:car_racing} shows cumulative regret over time. Both the no-learning and two-stage algorithms have high regret due to uncertain dynamics and random exploration, respectively. After 1000 steps, the two-stage algorithm tries to maximize reward, and thus, we see a decrease in the regret. In contrast to baselines, \dynExplor attains an order of magnitude lower regret. Initially, \dynExplor's regret increases attributed to exploring the dynamics, and then comes closer to the optimal behaviour as shown in \cref{fig:car_racing}.

%% file: sections/6_conclusion.tex
\section{Conclusion}
\label{sec:conclusion}
We proposed a general framework for safe, non-episodic online learning of unknown dynamics, introducing the notion of guaranteed exploration in policy space.
To our knowledge, this is the first theoretical analysis that successfully addresses this problem. 
Building on this, we proposed a theoretically grounded and efficient reward-maximization algorithm that achieves near-optimal performance while learning the dynamics only where necessary for optimality.
%
The proposed approach has two key limitations:
a) 
We utilize a robust bound on the noise $\eta(k)\in\mathcal{W}$, which may be conservative.
This can be naturally addressed by replacing the high-probability safety guarantees with (weaker) probabilistic safety guarantees using standard stochastic prediction methods~\citep{mesbah2016stochastic,ao2025stochastic}. 
b) Although we ensure safety for all times $k\in\mathbb{N}$, our performance bound is only valid for rewards over a finite horizon $\Horizon$. 
Providing similar performance bounds for infinite-horizon performance would require more involved tools, such as (stochastic) turnpikes~\citep{schiessl2024relationship}.


%% file: appx/extended_related_works.tex
\section{Extended related works}
\label{apxsec:extended_related_works}

\begin{table}[h]
\centering
\centering
\setlength{\tabcolsep}{3pt} 
\renewcommand{\arraystretch}{1.4} 
\small
\begin{tabular}{lccccc}
\hline \rule{0pt}{4.5ex}
& \makecell{Safety\\(training)}
 & \makecell{Safety\\(final $\pi$)} 
 & Optimality 
& \makecell{Nonepisodic \\ (without resets)}
& \makecell{Unknown \\ dynamics}  \\[0.9em] \hline\hline
\makecell{RL \citep{curi2020efficient}    }                 & \xmark   & \xmark   & \checkmark & \xmark   & \checkmark \rule{0pt}{1.75em}\\[0.45em] \hline 
\makecell{CMDP \citep{efroni2020exploration}}                 & \xmark   & ($\checkmark$) & \checkmark & \xmark   & \checkmark \rule{0pt}{1.75em}\\[0.45em] \hline
\makecell{Safe exploration CMDP \\ \citep{as2024actsafe}} 
& ($\checkmark$) & ($\checkmark$) & \checkmark & \xmark   & \checkmark \rule{0pt}{1.75em}\\[0.45em] \hline
\makecell{Safety Augmentation \\ \citep{wabersich2023data}  } & $\checkmark$   & $\checkmark$   & \xmark   & \checkmark & \checkmark \rule{0pt}{1.75em}\\[0.45em] \hline
\makecell{Safe constraint exploration \\ \citep{prajapat2025safe}} & $\checkmark$   & $\checkmark$   & ($\checkmark$) & \checkmark & \xmark \rule{0pt}{1.75em}\\[0.45em] \hline
\makecell{\textbf{SageDynX (Ours)} }        & $\checkmark$   & $\checkmark$   & \checkmark & \checkmark & \checkmark \rule{0pt}{1.75em}\\[0.45em] \hline
\end{tabular}
\vspace{1em}
\caption{Comparison of different approaches on safety, optimality, non-episodicity, and handling unknown dynamics. In the table above, the bracket symbol ($\checkmark$) denotes that the property holds but in a different form, e.g., safety holds in expectation for CMDPs. As shown in the table above, \dynExplor is the only algorithm that has all the properties, i.e., it guarantees optimality and safety during the entire execution while operating in a non-episodic setting with unknown dynamics.}
\label{tab:sagedynXrelated}
\end{table}

\looseness -1 
\mypar{Model-based RL and constrained MDPs} 
Model-based reinforcement learning (MBRL) approaches first learn a predictive model of the environment's dynamics, which is then leveraged for decision-making (either control, planning, or policy optimization). To effectively learn dynamics, MBRL methods employ various exploration strategies, for example, optimism-based~\citep{kakade2020information, curi2020efficient}, posterior (Thompson) sampling~\citep{osband2013more}, and deep ensemble-based \citep{chua2018deep}, among others.
Once the model is updated, a range of planning and policy optimization strategies can be applied to maximize cumulative rewards,
such as the standard policy gradient~\citep{curi2020efficient,prajapatsubmodular}, as well as sampling-based planners like cross-entropy method \citep{BOTEV201335}, MPPI \citep{kakade2020information}. These approaches can be combined with receding horizon control and re-planning techniques~\citep{lowrey2018plan}, for a general overview, see \citep{moerland2023model}. All these approaches are designed for unconstrained environments and do not account for state or action constraints.

Particularly appealing regret bounds have recently been shown in~\citet{kakade2020information}, which exhibit a polynomial dependence on the planning horizon $H$. In contrast, our analysis yields an exponential dependence on the horizon length when the Lipschitz constant $L > 1$.
This is due to different assumptions: we assume Lipschitz continuity, whereas \citet{kakade2020information} assumes Gaussian noise and a bounded second moment of cost under the worst policy, for an unbounded (unconstrained) domain.
To clarify the differences, consider the linear quadratic regulator (LQR) setting, where (see \cite[Remark 3.6]{kakade2020information}) 
a stabilizing policy is required to satisfy the bounded second moment of cost assumption. This sufficient condition corresponds to the (weighted \citep[Remark 2]{prajapat2025finite}) Lipschitz constant $L<1$, under which our analysis also shows a similar polynomial dependence on the horizon. 

A more closely related line of work involves constrained Markov decision processes (CMDPs)~\citep{altman2021constrained}, which enforce safety through constraints on the expected cumulative cost~\citep{achiam2017constrained,ding2020natural}. 
Due to the similarity in structure between cumulative rewards and costs, these approaches often leverage policy gradient methods, typically solved using Lagrangian multipliers or primal-dual techniques. 
In contrast, our work has state-wise (hard) constraints~\citep{zhao2023state,wachi2020safe}, which are more strict; that is, a policy trained with state-wise constraints is also safe in the CMDP sense; however, the converse is not assured~\citep{zhao2023state}. 
A large body of literature on model-based CMDPs focuses on empirically reducing constraint violations \citep{zanger2021safe,as2022constrained}, while theoretical guarantees are often limited to simplified settings such as finite (tabular) MDPs \citep{efroni2020exploration,muller2024truly}. 
Much of this work primarily addresses safety after learning (see a survey \citep{gu2024review}), with a few exceptions. 
For example, \citet{wachi2020safe} consider state-wise constraints and provides safety guarantees during exploration, but only within the tabular setting. 
In contrast, \citet{as2024actsafe} consider continuous spaces and proposes an algorithm that first uses pessimism for full exploration and then optimism for maximization of rewards in the explored region.
However, all of these methods
rely on \emph{episodic} environment resets in CMDPs, making their applicability to real-world scenarios challenging.




\mypar{Safety augmentation} There are also safety layer techniques that modify a standard (unconstrained) RL algorithm to enforce safety, such as safety filters~\citep{wabersich2023data}, control barrier functions~\citep{ames2016control}, and shielding~\citep{alshiekh2018safe}. 
 Although these approaches can sometimes also ensure safety without resets, they usually cannot maintain the optimality and regret guarantees of the original RL algorithm.

\mypar{Safe exploration for unknown constraints} Another line of work focuses on guaranteed exploration with unknown safety-critical constraints~\citep{safe-bo-sui15, sui2018stagewise,turchetta2016safemdp,prajapat2022near,tokmak2025safe}, primarily in finite (discrete) input spaces, and is typically used for safe Bayesian optimization (BO). 
The methodology has been applied to policy tuning~\citep{berkenkamp2023bayesian}, cost function tuning~\citep{puigjaner2025performance}, and more recently scaled to continuous spaces~\citep{fiedler2024safety} and continuous dynamical systems~\citep{prajapat2025safe}. 
However, a common assumption in these approaches is to exactly choose (and reach) where to collect information.
This is particularly relevant, as uncertain dynamics make it impossible to exactly reach the desired informative states, since we cannot precisely control the unknown system. 
Our approach addresses this problem using optimistic exploration objectives and the analysis in \cref{prop:suffi_info_epscollect}. 

\mypar{Model predictive control}
The problem of ensuring safety when controlling an uncertain dynamical system is commonly studied in the field of optimal control using model predictive control (MPC)~\citep{rawlings2017model}. 
MPC-based methods use online numerical optimization and model-based predictions to ensure safe operation. 
These approaches typically do not rely on resets, but ensure persistent safety with a receding horizon implementation and a terminal set constraint, which is similar to the considered safe set $\mathbb{X}_n$ (\cref{assump:safeSet})~\citep{rawlings2017model}. 
To handle model uncertainty, robust and stochastic MPC techniques predict robust or probabilistic reachable sets
 to ensure safe operation, see~\citep{singh2023robust,mesbah2016stochastic,ao2025stochastic,houska2019robust,koller_learning_based_2018,prajapat2025finite}. 
In addition, to improve performance, online model updates are also incorporated in the framework of learning-based MPC or robust adaptive MPC, which typically uses dynamics model sets similar to $\mathcal{F}_n$ in \cref{lem:beta}, see \citep{hewing2020learning,sasfi2023robust,brunke2021safe,lew2022safe,soloperto2023guaranteed}.
In general, these approaches can ensure safe operation, but typically lack the exploration guarantees needed to achieve~\cref{obj:maximum_exploration} and \cref{obj:reward_maximization} (optimality results). 
The problem of active learning in MPC is largely based on heuristics and lacks strong guarantees of exploration; see the review papers~\citep{mesbah2018stochastic,heirung2018model}.
Closer to our work, \citet{lew2022safe} proposed a sequential exploration-exploitation framework that ensures safety, but the proposed active learning component provides no guarantees on uncertainty reduction, and the analysis simply assumes that the task is eventually feasible. 
In~\citep{soloperto2023guaranteed}, an approach is developed that also provides exploration guarantees. 
However, the theoretical guarantees are only asymptotic and cannot account for probabilistic noise. 
Furthermore, while \cref{prop:suffi_info_epscollect} shows that we can ensure informative measurements based on optimistic exploration objectives, this approach relies on a worst-case exploration objective, i.e., the more uncertain the system is, the less the approach can provide informative data.



%% file: appx/safe_dyn_exploration.tex
\section{Proof for safe dynamics exploration}
\label{apxsec:dyn_exploration_proof}
All theoretical results leverage the uncertainty bound from \cref{lem:beta}, which holds jointly with probability $1-\delta$, and hence all the intermediate claims are only valid with the same probability of $1-\delta$. 
For simplicity of exposition, we will not state the probabilistic nature of the guarantees in the intermediate lemmas and propositions.

Given the dataset $\D$, the posterior mean $\mu_{\n,i}(z)$ and variance $\gppostvar^2_{\n,i}(z)$ of unknown dynamics at test inputs $z$, $z^{\prime}$ for any component $i\in \Intrange{1}{\statedim}$ are given by:
\begin{align}
    \mu_{\n,i}(z) &= \bm{\alpha}^\top_i \bm{k}^i_{D_\n}(z), \label{eq:mean_update}\\
    \kernelfunc^i_{D_\n}(z,z') &= \kernelfunc^i(z,z) - {\bm{\kernelfunc}^i_{D_\n}}^{\!\!\!\!\!\top}(z)   
    (K^i_{D_\n} + \sigma^2 \identity)^{-1}
    \bm{\kernelfunc}^i_{D_\n}(z'), \label{eq:kernel_posterior_update}\\ 
    \gppostvar^2(z) &= \kernelfunc^i_{D_\n}(z,z),
\end{align}
where $\bm{\alpha}_i = [\alpha^i_1, \hdots, \alpha^i_{D_\n}]^\top \coloneqq (K^i_{D_\n} + \sigma^2 \identity)^{-1} Y^i$, $\bm{\kernelfunc}^i_{D_\n}(z) = [k^i(z_1,z), . . . , k^i(z_{\datasetdim}, z)]^{\top}$, where $\kernelfunc^i$ is the kernel (\cref{assump:q_RKHS}), $Y^i\coloneqq e_i^\top Y$ are the measurements of the $i^{th}$ GP component ($e_i\in \R^\statedim$ is a basis vector)
and $K^i_{D_\n}$ is the resulting kernel matrix $[k^i_{D_\n}(z,z^{\prime})]_{z,z^{\prime} \in Z}$. 
For ease of notation, we use the shorthand $w_n(x,u)\coloneqq\max_{i\in\Intrange{1}{n}} w_{n,i}(x,u)$ to denote the maximal uncertainty across all state components.

\mypar{Choice of tolerance $\epsguarantee,\epsterminal,\epscollect$} 
As mentioned in~\cref{sec:full_expl_theory}, the tolerance $\epsguarantee>0$ is a user-chosen constant that can in principle be chosen arbitrarily small, subject to a constraint related to the noise magnitude $\noisebound$. In particular, the constant $\epscollect>0$ determines the threshold for informative measurements that are used to update the GP and can be chosen to an arbitrarily small positive constant. 
The threshold $\epsterminal$ used for planning and termination in \cref{alg:full_domain_exploration_basic} can then be chosen according to~\eqref{eq:JK_formula_epsilon_d}. 
Lastly, the tolerance $\epsguarantee>0$ for the theoretical guarantees in Theorem~\ref{thm:maximum_dynamics_exploration} can be chosen according to~\eqref{eq:JK_espilon_formula} and \eqref{eq:JK_bound_epsilon_2}, where $\epsclose\geq0$ is arbitrarily small. 
In combination, these formulas provide a lower bound on the achievable guarantees $\epsguarantee>0$ that depends linearly on the noise magnitude $\noisebound$.

We will first prove that the optimistic and pessimistic policy sets~\cref{eq:def_pessi_set,eq:def_opti_set} are outer and inner approximations of the true policy set~\eqref{eq:true_policy_set}, respectively.
\begin{proposition} \label{prop:over_under_approx}
    Suppose $\dyntrueVec\in\dynSet[n]$ then $\forall \n \in \Intp, X \subseteq \X$ it holds that $\pessiSet[]{\n}(X;\Horizon)\subseteq \truePolicySet[]{\n}(X;\Horizon)  \subseteq \optiSetwe[]{\n}(X;\Horizon)$.
\end{proposition}
\begin{proof}
    Consider some $\pessiPolicy \in \pessiSet[]{\n}(X;\Horizon)$. Then by definition~\eqref{eq:def_pessi_set}, $\exists \state_0 \in X$: 
    $\state_\h \in \X$,  $\coninput_{\h}\coloneqq \pessiPolicy[\h](\state_\h) \in \inputSpace$, $\state_\Horizon \in \safeInit{\n}$, $\forall \dynVec \in \dynSet[\n],\forall \noise_\h \in \W, \forall \h \in \Intrange{0}{\Horizon-1}$ with $\state_{\h+1} = \dynVec(\state_\h, \coninput_\h) + \noise_\h$. 

    Note the difference in dynamics used in the definition of the true policy set~\eqref{eq:true_policy_set} and the pessimistic policy set~\eqref{eq:def_pessi_set}. Since the above expression holds for all $\dynVec\in\dynSet[\n]$ and thus also for $\dyntrueVec \in \dynSet[\n]$, it follows that $\pessiPolicy \in \truePolicySet[]{\n}(X;\Horizon)$. This proves $\pessiSet[]{\n}(X;\Horizon)\subseteq \truePolicySet[]{\n}(X;\Horizon)$.

    Similarly, suppose $\pi^\star \in \truePolicySet[]{\n}(X;\Horizon)$. 
    Then, 
    $$\exists \state_0 \in X:
    \state_\h \in \X,  \coninput_{\h}\coloneqq \pi^\star_\h(\state_\h) \in \inputSpace, \state_\Horizon \in \safeInit{\n}, \forall \noise_\h, \in \W, \forall \h \in \Intrange{0}{\Horizon-1}$$ with $\state_{\h+1} = \dyntrueVec(\state_\h, \coninput_\h) + \noise_\h.$
    Since $\dyntrueVec\in\dynSet[\n]$,  there exists a function $\dynVec = \dyntrueVec\in\dynSet[\n]$, which satisfies the expression above and thus $\pi^\star \in \optiSetwe[]{\n}(X;\Horizon)$ by definition \eqref{eq:def_opti_set}. 
\end{proof}

The following proposition shows that the proposed sampling strategy always yields at least one informative sample $|\mathcal{D}_c| \geq 1$, validating the intuition from \cref{fig:selecting_sampling_rule}. 
\begin{proposition}
\label{prop:suffi_info_epscollect}
Suppose Assumption~\ref{assump:lipschitz} holds and $\dyntrueVec\in\dynSet[n]$.
Let $\pessiPolicy$ be a feasible solution to the safe exploration Problem~\eqref{eq:sampling_rule_timex} for some $x(k)\in \mathcal{X}$, $n\in\mathbb{N}$ with 
\begin{align}
\label{eq:JK_formula_epsilon_d}    
\epsterminal \coloneqq \epscollect + \LipWidthPi\LipDyn_{w,\totHorizon} (\epscollect + \noisebound),
\end{align}
where $\LipDyn_{w,\totHorizon} = \sum_{k=0}^{\totHorizon-1}( \LipDyn+\LipWidthPi)^k$. Then, $\exists \h \in \Intrange{0}{\totHorizon-1}, i\in \Intrange{1}{\statedim} : \cwidth[\n,i](\state^\star_h,\coninput^\star_\h) \geq \epscollect$ where $\state^\star_{\h+1} = \dyntrueVec(\state^\star_\h, \coninput^\star_\h) + \eta_h$ with $\coninput^\star_\h = \pessiPolicy[\h](\state^\star_\h), \eta_h \in \W$, $\state^\star_0=\state(k)$.
\end{proposition}
\begin{proof} We prove this by contradiction. Suppose $\pessiPolicy \in \pessiSet[]{\n}(\state(k);\totHorizon)$ is a feasible solution of Problem~\eqref{eq:sampling_rule_timex} and under this policy $\cwidth[\n,i](\state^\star_h,\coninput^\star_\h) < \epscollect, \forall \h \in \Intrange{0}{\totHorizon-1}, i \in \Intrange{1}{\statedim}$. 
Let $\state_{\h+1} = \dynVec(\state_\h, \coninput_\h)$ with $\coninput_\h = \pessiPolicy[\h](\state_\h)$ and $\state(\ki) = \state_0 = \state^\star_0$ be the corresponding noise-free trajectory with some $\dynVec\in\dynSet[n]$. 
Lipschitz continuity (\cref{assump:lipschitz}) implies
\begin{align}
\cwidth[\n,i](\state_h,\pessiPolicy(\state_\h))&\leq \cwidth[\n,i](\state^\star_h,\pessiPolicy(\state_\h^\star))+ \LipWidthPi \|\state_h - \state^\star_h\|   \nonumber\\
&< \LipWidthPi \|\state_h - \state^\star_h\|+\epsilon_c,
\label{eq:JK_bound_w_sequential}
\end{align}
\begin{align}
\|\state_{h+1}-\state_{h+1}^\star\|\! &=
\|\dynVec(\state_\h, \pessiPolicy(\state_h))-\dynVec^\star(\state^\star_\h, \pessiPolicy(\state^\star_\h))-\eta_h\|\nonumber\\
&\leq \|\eta_h\|\!+\!\|\dynVec^\star(\state_\h, \pessiPolicy(\state_h))\!-\!\dynVec^\star(\state^\star_\h, \pessiPolicy(\state^\star_\h))\|
\!+\!\|\dynVec(\state_\h, \pessiPolicy(\state_h))\!-\!\dynVec^\star(\state_\h, \pessiPolicy(\state_\h))\|\nonumber\\
&\stackrel{\eqref{eq:dyn_set}}{\leq} \noisebound + \LipDyn\|\state_h- \state_h^\star\|
+\cwidth[\n,i](\state_h,\pessiPolicy(\state_\h)),\nonumber\\
&\stackrel{\eqref{eq:JK_bound_w_sequential}}{<} \noisebound + \LipDyn\|\state_h- \state_h^\star\|
+\LipWidthPi \|\state_h - \state^\star_h\|+\epscollect, \label{eq:JK_bound_x_sequential}
\end{align}
which, on using recursion, implies,
\begin{align}
\label{eq:JK_bound_x_total}
\|\state_h-\state_h^\star\|< &\underbrace{\sum_{k=0}^{h-1}( \LipDyn+\LipWidthPi)^k}_{\eqqcolon\LipDyn_{w,h}} (\epscollect+\noisebound).
\end{align}
Considering the point $h'\in\Intrange{0}{\totHorizon-1}$ with $w_{n,i}(x_{h'},u_{h'})\geq \epsterminal$ from Problem~\eqref{eq:sampling_rule_timex}, we have
\begin{align*}
\cwidth[\n,i](\state^\star_{h'},\coninput^\star_{\h'})&\geq \cwidth[\n,i](\state_{h'},\coninput_{\h'})-\LipWidthPi  \|\state_{h'}-\state^\star_{h'}\| \\
&\geq \epsterminal-\LipWidthPi\LipDyn_{w,h'}(\epscollect+\noisebound)\\
&\geq \epsterminal-\LipWidthPi\LipDyn_{w,\totHorizon}(\epscollect+\noisebound)\stackrel{\eqref{eq:JK_formula_epsilon_d}}{=}\epscollect,
\end{align*}
which yields a contradiction.
\end{proof}
To study sample complexity, we denote by $\numdataPts_{\n}$ the number of new measurements (of all $\statedim$ components) collected to be used for the $\n^{th}$ model update. \change{Let $\numdataPts^\epsilon_\n$ be the set of indices of the informative measurements, i.e., $\{ \ d \ | \ \cwidth[\n-1, i](\stateAct_{\numdataPts, \n})\geq \epscollect\}$}. Note that Proposition~\ref{prop:suffi_info_epscollect} 
ensures $|\numdataPts^\epsilon_{\n}| \geq1, \forall \n \in \Intp$ and hence $\numdataPts_{\n}\geq1$. 
Define $D_\n = \sum_{j=1}^{\n} \numdataPts_j + D_0$ with $D_0$ as the prior data. If there are no prior data, $D_0 = 0$.
Furthermore, we define the posterior kernel matrix $k^i_{D_{n-1}}(\cdot)$ evaluated for the $d_n$ measurements corresponding to the $i$-th GP by $K_{d_{n},d_{n}}^i\in\mathbb{R}^{d_n\times d_n}$ and the eigenvalues of this matrix by $\lambda_{r,n,i}$, $r=1,\dots,d_n$.

The following lemma provides a sample complexity bound, ensuring that Algorithm~\ref{alg:full_domain_exploration_basic} terminates in finite time.
\begin{lemma}[Sample complexity, adapted from~\citet{prajapat2025safe}] \label{lem:sample_complexity} Let \cref{assump:q_RKHS,assump:safeSet,assump:lipschitz,assump:sublinear} hold and $\n^\star$ be the largest integer satisfying
$\frac{\n^\star}{\beta_{\n^\star} \gamma_{\n^\star}} \leq \frac{C_1}{\epscollect^2}$
with $C_1 =  8 \totHorizon/ \log (1 + \totHorizon\noiseconst)$.
Consider Algorithm~\ref{alg:full_domain_exploration_basic}, which collects all measurements along the path returned by Problem \eqref{eq:sampling_rule_timex}.
There exists an iteration $\bar{n}\leq n^\star$, such that at the end of the iteration with $x(k)\in\mathbb{X}_n$ satisfies:
\begin{align}
    \cwidth[\nfin, i](\state_\h, \coninput[\h]) < \epsterminal, \forall i \in \Intrange{1}{\statedim}, \h \in \Intrange{0}{\totHorizon-1}, \dynVec \in \dynSet[\n], \pi \in \pessiSet[]{\nfin}(\state(k), \totHorizon) \label{eqn:lemma_termination_condition}
\end{align}
where $\state_{\h+1} = \dynVec(\state_\h,\coninput_\h), \coninput_\h = \pi_\h(\state_\h)$ and $\state_0=\state(k)$, i.e., \cref{alg:full_domain_exploration_basic} terminates in $\nfin$ iterations.
\end{lemma}

\input{appx/lemma2_Alternative}

The following intermediate lemma uses Lipschitz continuity (\cref{assump:lipschitz}) to bound the effect of different dynamics and policies and derive the tolerance $\epsguarantee$. 
\begin{lemma} \label{lem:contra_eps} 
Let \cref{assump:q_RKHS,assump:safeSet,assump:lipschitz,assump:sublinear} hold. 
Then \cref{alg:full_domain_exploration_basic} terminates at some state $\state(\ki)=\state_s\in\safeInit{\nfin}$ with $\nfin\leq \n^\star$. 
Consider two policies $\pi^b \in \optiSetwe[,\frac{\epsguarantee}{2}]{c,\nfin}(\state_s;\Horizon + \delta \h) \backslash \pessiSet[]{\nfin}(\state_s;\Horizon + \delta \h)$, $\pessiPolicy \in \pessiSet[]{\nfin}(\state_s;\Horizon + \delta \h)$ such that $\|\pi^b(x) - \pessiPolicy(x)\| \leq \epsclose$ $\forall x\in\X$ with arbitrary small $\epsclose>0$ and some $\delta \h \in \Intrange{0}{\Delta \Horizon}$. 
%
Then it holds that $\forall \dynVec_1,\dynVec_2 \in \dynSet[\nfin]$: $\| \state_\h^{s}  - \state_\h^o\| \leq \epsguarantee/2$, $\|\coninput_\h^s-\coninput_\h^o\|\leq \epsguarantee/2, \forall \h \in \Intrange{0}{\Horizon + \delta \h-1}$, with $\state^s_{\h+1} = \dynVec_1(\state^s_{\h}, \coninput_\h^s) + \noise_\h$, $\coninput_\h^s \coloneqq \pi^b(\state^s_{\h})$ and $\state^o_{\h+1} = \dynVec_2(\state^o_{\h}, \coninput_\h^o)) + \noise_\h,\coninput_\h^o\coloneqq\pi^b(\state^o_{\h}), \state^o_0= \state^s_0=\state_s$, $\eta_h\in\W$. 
\end{lemma}
\begin{proof} From the termination criterion, we know that $\cwidth[\nfin, i](\state_\h, \coninput[\h]) < \epsterminal,\forall \h\in \Intrange{0}{\totHorizon-1}$, $\dynVec \in \dynSet[\nfin]$, $\pi \in \pessiSet[]{\nfin}(\safeInit{\nfin};\totHorizon)$ with $\state_{\h+1} = \dynVec(\state_\h,\coninput_\h)$, $\coninput_\h = \pi_\h(\state_\h)$, $\state_0=\state_s\in\safeInit{n}$. 
Using the invariance property of \cref{assump:safeSet}, we know that for all $\delta\h \in \Intrange{0}{\Delta \Horizon}, \pessiPolicy \in \pessiSet[]{\nfin}(\state_s;\Horizon + \delta \h),~\exists \pi_f \in \Pi_{\Delta \Horizon - \delta \h}: [\pessiPolicy,\pi_f] \in \pessiSet[]{\nfin}(\state_s;\totHorizon)$ , i.e., after reaching in the terminal set with $\pessiPolicy$, the appended policy $\pi_f$ ensures that $\forall \dynVec \in \dynSet[\nfin]$ with noise, the propagated state still remains in the safe set. 
Hence, $\cwidth[\n, i](\state_\h, \coninput[\h]) < \epsterminal,\forall \h\in \Intrange{0}{\Horizon + \delta \h-1}$, $\dynVec \in \dynSet[\nfin]$, $\pi \in \pessiSet[]{\nfin}(\state_s;\Horizon + \delta \h)$ with $\state_{\h+1} = \dynVec(\state_\h,\coninput_\h)$, $\coninput_\h = \pi_\h(\state_\h)$, $\state_0=\state_s\in\safeInit{\nfin}$
and thus $\cwidth[n](\statetrue[,p]_{\h})<\epsterminal$ $\forall \h\in \Intrange{0}{\Horizon + \delta \h}$. 

Define the true dynamics subject to the two different policies by
$\statetrue[,b]_{\h+1} \coloneqq \dyntrueVec(\statetrue[,b]_{\h}, \pi^b_{\h}(\statetrue[,b]_{\h})) + \eta_\h$ and $\statetrue[,p]_{\h+1} \coloneqq \dyntrueVec(\statetrue[,p]_{\h}, \pessiPolicy[\h](\statetrue[,p]_{\h}))~\forall \h \in \Intrange{0}{\Horizon + \delta \h-1}$, where only $\statetrue[,b]_{\h}$ is subject to noise $\eta_h$. 

The remainder of this proof repeatedly uses Lipschitz continuity arguments similar to \cref{prop:suffi_info_epscollect}. 
The true dynamics subject to the two policies satisfies
\begin{align}
\|\statetrue[,b]_{\h+1}-\statetrue[,p]_{\h+1}\|& \leq\LipDyn\|\statetrue[,b]_{\h}-\statetrue[,p]_{\h}\|+L_{\mathrm{f}} \epsclose+\noisebound
\leq \dots\leq \sum_{j=0}^{h}\LipDyn^j (L_{\mathrm{f}}\epsclose+\noisebound),\label{eqn:two_policy_recursion}
\end{align}
\begin{align*}
    \cwidth[\nfin](\statetrue[,b]_{\h}, \pi^b_{\h}(\statetrue[,b]_{\h})) - &\cwidth[\nfin](\statetrue[,p]_{\h}, \pessiPolicy[\h](\statetrue[,p]_{\h}))\\ 
   &\leq \LipWidth \|\statetrue[,b]_{\h}-\statetrue[,p]_{\h}\| + \LipWidth \| \pi^b_{\h}(\statetrue[,b]_{\h}) - \pi^b_{\h}(\statetrue[,p]_{\h}) + \pi^b_{\h}(\statetrue[,p]_{\h}) - \pessiPolicy[\h](\statetrue[,p]_{\h}) \|\\
   &\leq \LipWidth  (1+\LipPi) \|\statetrue[,b]_{\h}-\statetrue[,p]_{\h}\| + \LipWidth \epsclose\\
   &\stackrel{\eqref{eqn:two_policy_recursion}}{\leq} \underbrace{\LipWidth\left(1+L_{\mathrm{f}}(1+\LipPi)\sum_{j=0}^{h-1}\LipDyn^j\right)}_{\eqqcolon K_{\epsilon,h}}\epsclose+\underbrace{\LipWidthPi\sum_{j=0}^{h-1}\LipDyn^j}_{=\LipWidthPi \LipDyn_{\h}}\noisebound,
   \end{align*}
which ensures
\begin{align*}
 \cwidth[\nfin](\statetrue[,b]_{\h}, \pi^b_{\h}(\statetrue[,b]_{\h})) \leq \cwidth[\nfin](\statetrue[,p]_{\h}, \pi^p_{\h}(\statetrue[,b]_{\h}))+ K_{\epsilon,h}\epsclose+\LipWidthPi \LipDyn_{\h}\noisebound<\epsterminal + K_{\epsilon,h}\epsclose+\LipWidthPi \LipDyn_{\h}\noisebound.
\end{align*}
Next, we bound the difference between the true dynamics and some dynamics $\dynVec_1\in\dynSet[n]$, both subject to the same policy $\pi^b$ and same noise $\eta_h$:
\begin{align*}
 \cwidth[\nfin](\state^s_\h, \pi^b_{\h}(\state^s_\h))
 & \leq \cwidth[\nfin](\statetrue[,b]_{\h}, \pi^b_{\h}(\statetrue[,b]_{\h}))+\LipWidthPi\|\statetrue[,b]_{\h}-\state^s_\h\|\\
 & < \epsterminal+K_{\epsilon,h}\epsclose+\LipWidthPi \LipDyn_{\h}\noisebound+\LipWidthPi\|\statetrue[,b]_{\h}-\state^s_\h\|, \numberthis \label{eq:diff_f1_with_true_dyn}\\
\|\statetrue[,b]_{\h+1}-\state^s_{\h+1}\| & \leq\LipDyn\|\statetrue[,b]_{\h}-\state^s_{\h}\|+  \cwidth[\nfin](\state^s_\h, \pi^b_{\h}(\state^s_\h))\\
& \stackrel{\eqref{eq:diff_f1_with_true_dyn}}{\leq} (\LipDyn+\LipWidthPi)\|\statetrue[,b]_{\h}-\state^s_{\h}\|+  \epsterminal+K_{\epsilon,h}\epsclose+\LipWidthPi \LipDyn_{\h}\noisebound, 
 \end{align*}
which using recursion until $\statetrue[,b]_{0}=\state^s_{0}$ implies, 
$$\|\statetrue[,b]_{\h}-\state^s_{\h}\| \leq\sum_{j=0}^{ h-1}(\LipDyn+\LipWidthPi)^j (\epsterminal+K_{\epsilon,h-j-1}\epsclose+ \LipWidthPi \LipDyn_{h-j-1}\noisebound).$$

Analogously, we have for $\state^o$ with dynamics $\dynVec_2\in\dynSet[n]$:
\begin{align*}
\|\statetrue[,b]_{\h}-\state^o_{\h}\|& \leq\sum_{j=0}^{ h-1}(\LipDyn+\LipWidthPi)^j (\epsterminal+K_{\epsilon,h-j-1}\epsclose+ \LipWidthPi \LipDyn_{h-j-1}\noisebound),
 \end{align*}
 which yields
\begin{align*}
\|\state^s_{h}-\state^o_{\h}\|\leq 2\sum_{j=0}^{ h-1}(\LipDyn+\LipWidthPi)^j (\epsterminal+K_{\epsilon,h-j-1}\epsclose+ \LipWidthPi \LipDyn_{h-j-1}\noisebound)\leq \epsguarantee/2,\\
\|\coninput^s_{h}-\coninput^o_{\h}\|\leq 2\LipPi\sum_{j=0}^{ h-1}(\LipDyn+\LipWidthPi)^j(\epsterminal+K_{\epsilon,h-j-1}\epsclose+ \LipWidthPi \LipDyn_{h-j-1}\noisebound)\leq \epsguarantee/2\,
 \end{align*} 
with
\begin{align}
\label{eq:JK_espilon_formula}
\epsguarantee > \max\{1,\LipPi\} 4 \sum_{j=0}^{ \totHorizon-1}(\LipDyn+\LipWidthPi)^j (\epsterminal+ \LipWidthPi \LipDyn_{\totHorizon-j-1}\noisebound), 
 \end{align} 
and $\epsclose>0$ sufficiently small. 
\end{proof}

\begin{figure}
\begin{subfigure}[b]{0.48\columnwidth}
    \centering
    \scalebox{0.5}{\input{figure/dyn_opti_in_pessi}}
    \vspace{0.2em}
\caption{\looseness -1 Contradiction: pick $\pi^b : \|\pi^b - \pessiPolicy\| <\epsclose$} 
    \label{fig:dyn_opti_in_pessi} 
\end{subfigure}
~
\begin{subfigure}[b]{0.48\columnwidth}
    \centering
    \scalebox{0.6}{\input{figure/dyn_traj_ball_eps}} 
    \caption{\looseness -1 Closeness of trajectories of $\dynVec^s \in \dynSet[\nfin]$ under $\pi^b$} 
    \label{fig:dyn_traj_ball_eps} 
\end{subfigure}
\caption{ 
\looseness -1 
Illustrations of the ingredients required in the proof of \cref{obj:maximum_exploration}. As shown in \cref{fig:dyn_opti_in_pessi}, we pick an optimistically safe policy $\pi^b$ arbitrarily close to the boundary of the pessimistic policy set in a proof by contradiction. As highlighted in \cref{fig:dyn_traj_ball_eps}, we show that applying the policy $\pi^b$ will keep all the $\dynVec^s \in \dynSet[\nfin]$ in a small ball of radius $\epsguarantee/2$ around the optimistic trajectory.} 
\end{figure}

The following lemma relates the termination criterion in \cref{alg:full_domain_exploration_basic} to the uncertainty in the optimistic constraint set to prove exploration of the optimistic set $\optiSetwe[,\frac{\epsguarantee}{2}]{c,\nfin}(\state_s;\Horizon + \delta \h)$.
\begin{lemma}\label{lem:uniform_exploration}
Let \cref{assump:q_RKHS,assump:safeSet,assump:lipschitz,assump:sublinear} hold. 
Then, \cref{alg:full_domain_exploration_basic} terminates in $\nfin$ iteration and the current state $\state(\ki)=\state_s\in\safeInit{\nfin}$ satisfies 
$\optiSetwe[,\frac{\epsguarantee}{2}]{c,\nfin}(\state_s;\Horizon+\delta \h) \subseteq \pessiSet[]{\nfin}(\state_s;\Horizon + \delta \h)$. 
\end{lemma}
\begin{proof} 
%
For contradiction, assume $\optiSetwe[,\frac{\epsguarantee}{2}]{c,\nfin}(\state_s;\Horizon + \delta \h) \backslash \pessiSet[]{\nfin}(\state_s;\Horizon + \delta \h)\neq \emptyset$.
Due to path connectedness of $\optiSetwe[,\frac{\epsguarantee}{2}]{c,\nfin}(\state_s;\Horizon + \delta \h)$ (cf. \cref{fig:dyn_exploration_convergence}), $\exists \pi^b \in \optiSetwe[,\frac{\epsguarantee}{2}]{c,\nfin}(\state_s;\Horizon + \delta \h) \backslash \pessiSet[]{\nfin}(\state_s;\Horizon + \delta \h)$ arbitrary close to the boundary of $\pessiSet[]{\nfin}(\state_s;\Horizon + \delta \h)$ (cf. \cref{fig:dyn_opti_in_pessi}). 
Hence,we consider a policy $\pi^b$ such that $\|\pi^b(x) - \pessiPolicy(x)\| \leq \epsclose$ with $\epsclose>0$ arbitrarily small. 

Now, since $\pi^b \in \optiSetwe[,\frac{\epsguarantee}{2}]{c,\nfin}(\state_s;\Horizon + \delta \h)$, this implies  
$\exists \dynVec^o \in \dynSet[\nfin]: \state^o_\h \in \X \ominus \ball{\frac{\epsguarantee}{2}}, \coninput_\h  \in \inputSpace \ominus \ball{\frac{\epsguarantee}{2}}, \forall \h \in \Intrange{0}{\Horizon + \delta \h-1}, \eta_\h \in \W$, where $\state^o_{\h+1} = \dynVec^o(\state^o_\h,\coninput_\h) + \eta_\h, \coninput_\h \coloneqq \pi^{b}_\h(\state^o_\h)$ and  $\state^o_{\Horizon + \delta \h} \in \safeInit{\nfin} \ominus \ball{\frac{\epsguarantee}{2}}.$ 

From~\cref{lem:contra_eps}, 
we know that $\forall \dynVec^s \in \dynSet[\nfin], \eta_\h \in \W, \h \in \Intrange{0}{\Horizon + \delta \h-1}, \state_\h^{s} \in \state_\h^{o} \oplus \ball{\frac{\epsguarantee}{2}}$,  $ \coninput_\h^{s} \in \coninput_\h^{o} \oplus \ball{\frac{\epsguarantee}{2}}$
where $\state^s_{\h+1} = \dynVec^s(\state^s_{\h},\coninput_h^s) + \eta_\h$, $\coninput_h^s=\pi^b(\state^s_{\h})$, $\state^o_{\h+1} = \dynVec^o(\state^o_{\h},\coninput_h^o ) + \eta_\h$, $\coninput_h^o=\pi^b(\state^o_{\h})$ $\state_s = \state^o_0= \state^s_0$;  
(as shown in ~\cref{fig:dyn_traj_ball_eps}) and both dynamics are driven by the same noise sequence $\noise_\h, \h\in\Intrange{0}{\Horizon + \delta \h-1}$.


This implies that $\forall \dynVec^s \in \dynSet[\nfin]: \state^s_\h \in \X \ominus \ball{\frac{\epsguarantee}{2}} \oplus \ball{\frac{\epsguarantee}{2}} \subseteq \X, \forall \h \in \Intrange{0}{\Horizon + \delta \h}, \state^o_{\Horizon + \delta \h}\in\safeInit{\nfin}\ominus \ball{\frac{\epsguarantee}{2}} \oplus \ball{\frac{\epsguarantee}{2}} \subseteq \safeInit{\nfin}$  and $\pi^b_\h(\state^s_{\h}) \in \inputSpace \ominus \ball{\frac{\epsguarantee}{2}} \oplus \ball{\frac{\epsguarantee}{2}} \subseteq \inputSpace$. 

Hence $\pi^b \in \pessiSet[]{\nfin}(\safeInit{\nfin};\Horizon + \delta \h)$. However, this is a contradiction. This implies~$\optiSetwe[,\frac{\epsguarantee}{2}]{c,\nfin}(\state_s;\Horizon + \delta \h) \subseteq \pessiSet[]{\nfin}(\state_s;\Horizon + \delta \h)$.
\end{proof}

The following lemma utilizes controllability to relate the set of safe policies with a free initial condition in the safe set $\safeInit{n}$ to those starting at some fixed initial condition but with a larger horizon $\Horizon + \delta \h$.
\begin{lemma} \label{lem:horizon_diff}
   Let \cref{assump:q_RKHS,assump:safeSet,assump:lipschitz} hold. $\forall  \state' \in  \safeInit{\n}$, $\pi \in \truePolicySet[,\epsguarantee]{c,\n}(\state'; \Horizon), \exists \hat{\pi} \in\Pi_{\delta \h}: 
[\hat{\pi}, \pi ] \in \truePolicySet[,\frac{\epsguarantee}{2}]{c,\n}(\state_s;\Horizon + \delta \h)$ $\forall \n \in \Intp$.
\end{lemma}
\begin{proof}
Pick $\pi\in\truePolicySet[,\epsguarantee]{c,\n}(\state';\Horizon)$ with initial state $\state'=\state^\star_0 \in \safeInit{\n}$. 
Using the controllability property in \cref{assump:safeSet}, $\exists \hat{\pi} \in\Pi_{\delta \h}$ such that $\dyntrueVec$ can be controlled from $\state'$ to $\state_0^\star \in \safeInit{\n}$ in steps $\delta \h$, without accounting for noise $\eta$. 
We will prove that $[\hat{\pi}, \pi ] \in \truePolicySet[,\frac{\epsguarantee}{2}]{c,\n}(\state_s;\Horizon + \delta \h)$, i.e., the trajectory resulting from applying the appended policy $[\hat{\pi}, \pi ]$ satisfies the constraints. 
\cref{fig:dyn_init_set_appending} illustrates the two involved trajectories. 
%
\begin{figure}
    \centering
    \scalebox{0.7}{\input{figure/dyn_init_set_appending}}
    \caption{
 Illustration in state space of the trajectories used to relate the policy sets under two different horizons in \cref{lem:horizon_diff}. 
    The policy $\pi \in \truePolicySet[,\epsguarantee]{c,\n}(\safeInit{\n}; \Horizon)$ drives the system from a state $\state'$ to $\state^\star_{\Horizon}\in \safeInit{\n}$, while ensuring that all intermediate states satisfy $\state^\star_\h \in \X \ominus \ball{\epsguarantee}$,  $\h \in \Intrange{0}{\Horizon}$ (shown with the dashed line). 
    We show that the concatenated policy $[\hat{\pi}, \pi ] \in \truePolicySet[,\frac{\epsguarantee}{2}]{c,\n}(\state_s;\Horizon + \delta \h)$ control the true system $\dyntrueVec$ starting at $\state_s$, lands in a ball around $\state'$ (marked with circle) due to noise and then follows $\state^\star_\h, \h \in \Intrange{0}{\Horizon}$ closely. The resulting trajectory $\state^{\star,l}_\h$ satisfies the constraints $\state^{\star,l}_\h \in \X \ominus \ball{\frac{\epsguarantee}{2}}$, $\h \in \Intrange{0}{\Horizon}$. 
    A similar argument applies to the input constraints, though not shown in the state space figure above.}
    \label{fig:dyn_init_set_appending}
\end{figure}
Given \cref{assump:lipschitz}, the closed-loop dynamics of $\dyntrueVec$ with policy $\pi$ is $\LipDyn$-Lipschitz continuous.
Let $\state^{\star,l}_h$ and $\state^\star_h$ denote the state sequences when applying the policy $\pi$ and the same noise sequence $\eta_h$ to dynamics $\dyntrueVec$ with initial conditions $x_0^{\star,l}$ and $x_0^\star$ , respectively.
Due to noise, 
the deviation satisfies $\|\state^\star_0 - \state^{\star,l}_0\| \leq 
\sum_{h=0}^{ \delta h - 1}\LipDyn^h\noisebound$ where $\state^{\star,l}_0$ is the state $\dyntrueVec$ ends after controlling with $\hat{\pi}$. 
Analogous to~\cref{prop:over_under_approx}, it holds that
\begin{align*}
     \|x_{h+1}^\star -x_{h+1}^{\star,l}\| &=    \| \dyntrueVec(\state^\star_{h}, \pi_\h(\state^{\star}_{h})) - \dyntrueVec(\state^{\star,l}_{\h}, \pi_\h(\state^{\star,l}_{\h}))\|
     \leq \LipDyn\|\state^{\star}_{\h}-\state^{\star,l}_{\h}\|
     \leq \LipDyn^h\|\state^\star_0-\state^{\star,l}_0\|\\
     & \leq \max_{h\in\Intrange{0}{H}} \{\LipDyn^h\} \sum_{i=0}^{ \delta \h - 1}\LipDyn^i \noisebound \leq \epsguarantee/2,
\end{align*}
and 
\begin{align*}
   \|\pi_h(\state^\star_\h) - \pi_h(\state^{\star,l}_\h) \|\leq \LipPi \| \state^\star_\h - \state^{\star,l}_\h \| \leq \epsguarantee/2 \quad  \forall \h \in \Intrange{0}{\Horizon-1}.
\end{align*}
with 
\begin{align}
     \label{eq:JK_bound_epsilon_2}
     \epsilon \geq \max\{1, L_\pi\} 2 \max\{\LipDyn^\Horizon,1\} \LipDyn_{\delta 
     \h}\noisebound.
\end{align}
Thus, $[\hat{\pi}, \pi] \in \truePolicySet[,\frac{\epsguarantee}{2}]{c,\n}(\state_s;\Horizon + \delta \h)$. This completes the proof.\end{proof}
Given, these lemmas, we are ready to prove our main result.
\restatedynexploration*
\begin{proof}
Safety is ensured since \cref{alg:full_domain_exploration_basic} applies the pessimistic policy $\pessiPolicy \in \pessiSet[]{\n}(\state(\ki);\totHorizon), \forall \n \geq 0$ which, by definition, ensures constraint satisfaction $\forall \dynVec \in \dynSet[\n]$. By \cref{assump:q_RKHS} and \cref{lem:beta}, the unknown system satisfies $\dyntrueVec \in \dynSet[\n]$, thereby guaranteeing constraint satisfaction for the unknown system~\eqref{eq:system_dyn}, with at least probability $1-\delta$.


Initially at $\n=0$, \cref{assump:safeSet} (control invariance) ensures that $\pessiSet[]{0}(\state_s;\totHorizon)$ is non-empty for any $x_s\in \safeInit{0}$. 
\cref{prop:suffi_info_epscollect} shows that the sampling rule~\eqref{eq:sampling_rule_timex} ensures collected measurements has at least one that satisfy $w_{n,i}(z(k))\geq \epsilon_c$. 
Now by \cref{lem:sample_complexity}, we know that $\exists \nfin \leq \n^\star, \state(\ki) \in \safeInit{\nfin} : \forall i \in \Intrange{1}{\statedim}, \h \in \Intrange{0}{\totHorizon-1}, \dynVec \in \dynSet[\n], \pi \in \pessiSet[]{\nfin}(\state_s, \totHorizon)$, $\cwidth[\nfin, i](\state_\h, \coninput[\h]) < \epsterminal$ where $\state_{\h+1} = \dynVec(\state_\h,\coninput_\h), \coninput_\h = \pi_\h(\state_\h)$. This implies that sampling rule~\eqref{eq:sampling_rule_timex} will become infeasible in $\nfin\leq \n^\star$ model updates, which implies that \cref{alg:full_domain_exploration_basic} will terminate in $\nfin$ iterations.

Finally, \cref{lem:uniform_exploration} shows that for the final state $\state(\ki)= \state_s \in \safeInit{\n} $
we have $\optiSetwe[,\frac{\epsguarantee}{2}]{c,\nfin}(\state_s;\Horizon + \delta \h) \subseteq \pessiSet[]{\nfin}(\state_s;\Horizon + \delta \h)$ which using  \cref{prop:over_under_approx} implies $\truePolicySet[,\frac{\epsguarantee}{2}]{c,\nfin}(\state_s;\Horizon + \delta \h) \subseteq \optiSetwe[,\frac{\epsguarantee}{2}]{c,\nfin}(\state_s;\Horizon + \delta \h) \subseteq \pessiSet[]{\nfin}(\state_s;\Horizon + \delta \h)$.


Thus, \cref{lem:horizon_diff} ensures that~\cref{obj:maximum_exploration} holds:\\ 
$\pi^\star\in\truePolicySet[,\epsguarantee]{c,\nfin}(\safeInit{\nfin},\Horizon)$ implies $\exists \hat{\pi} \in\Pi_{\delta \h}: [\hat{\pi},\pi^\star]\in \truePolicySet[,\frac{\epsguarantee}{2}]{c,\nfin}(\state_s,\Horizon + \delta \h) \subseteq \pessiSet[]{\nfin}(\state_s,\Horizon + \delta\h)$. 
\end{proof}

The following lemma characterizes the mutual information, which is utilized for the sample complexity bound in Lemma~\ref{lem:sample_complexity}.
\begin{lemma}[{\citep[Lemma 5]{prajapat2022near}}]\label{lem:mutual_information} 
Let $\entropy(Y_{\D_\n})$ be the Shannon entropy for noisy samples $Y_{\D_\n}$ collected at the set $\D_{\n}$ \citep{beta_srinivas}.
 Then, the mutual information $I(Y_{\D_\nfin};\dyntrueVec_{\D_\nfin}) \coloneqq  \entropy(Y_{\D_\nfin}) - \entropy(Y_{\D_\nfin}| \dyntrueVec_{\D_\nfin})$ is given by
\begin{align*}
    I(Y_{\D_\nfin};\dyntrueVec_{\D_\nfin})= \sum_{\n=0}^{\nfin} \sum_{i=1}^{\statedim} 
    \sum_{\numdataPts=1}^{\numdataPts_{\n}}
    \frac{1}{2} \log(1+\noiseconst \lambda_{\numdataPts, \n, i})) 
\end{align*}
\end{lemma}
\begin{proof}
The data set $\D_\n = \{d_{0}, \hdots, \numdataPts_{\n}\}$
\begin{align} 
    \entropy(Y_{\D_\n}) &= \entropy(Y_{\numdataPts_{\n}} | Y_{\D_{\n-1}}) + \entropy(Y_{\D_{\n-1}}) \\
    &= \sum_{i=1}^{\statedim} \left( \frac{1}{2} \log (\det (2 \pi e (\sigma^2 I  + K^i_{\numdataPts_{\n},\numdataPts_{\n}}) ) ) \right) + \entropy(Y_{d_{\n-1}}|Y_{\D_{\n-1}}) + \hdots \label{eqn:MV-Gaussian}\\
    &= \sum_{i=1}^{\statedim} \left( \frac{\numdataPts_{\n}}{2} \log (2 \pi e \sigma^2) + \frac{1}{2} \log (\det (I  + \sigma^{-2} K^i_{\numdataPts_{\n},\numdataPts_{\n}}) ) ) \right) + \entropy(Y_{d_{\n-1}}|Y_{\D_{\n-1}}) + \hdots \label{eqn:refactoring-det}\\
    &= \sum_{\n=0}^{\nfin} \sum_{i=1}^{\statedim} \frac{\numdataPts_{\n}}{2} \log (2 \pi e \sigma^2) + \frac{1}{2} \log (\det (I  + \sigma^{-2} K^i_{\numdataPts_{\n},\numdataPts_{\n}}) ) ) \label{eqn:recursion}
\end{align}
For~\eqref{eqn:MV-Gaussian}, we used that each component $i$ is independently sampled as $Y_{\numdataPts_{\n}} \sim \mathcal{N}(Z_{\numdataPts_{\n}}, \sigma^2 I  + K^i_{\numdataPts_{\n},\numdataPts_{\n}})$ is jointly a multivariate Gaussian. 
\cref{eqn:refactoring-det} follows by 
\begin{align*}
\log \left(\det \left(2 \pi e \left(\sigma^2 I  + K^i_{\numdataPts_{\n},\numdataPts_{\n}}\right)\right)\right) &=  \log\left({\left(2\pi e\sigma^2\right)}^{\numdataPts_{\n}}\det \left(\sigma^2 I  + K^i_{\numdataPts_{\n},\numdataPts_{\n}}\right)\right) \\
&= {\numdataPts_{\n}}\log{\left(2\pi e\sigma^2\right)}+\log\left(\det \left(\sigma^2 I  + K^i_{\numdataPts_{\n},\numdataPts_{\n}}\right)\right).
\end{align*}
Finally, \cref{eqn:recursion} follows by repeating the above 2 steps recursively until $\n=0$.
$\entropy(Y_{\D_\n}| \dyntrueVec_{\D_\n}) = \sum_{\n=0}^{\nfin} \sum_{i=1}^{\statedim} \frac{\numdataPts_{\n}}{2} \log (2 \pi e \sigma^2)$ is the entropy due to noise. On substituting this together with \cref{eqn:recursion} in the mutual information definition, we obtain
\begin{align*}
I(Y_{\D_\nfin};\dyntrueVec_{\D_\nfin}) &= \frac{1}{2} \sum_{\n=0}^{\nfin} \sum_{i=1}^{\statedim}  \log (\det (I  + \sigma^{-2} K^i_{\numdataPts_{\n},\numdataPts_{\n}}) )    \\
&\labelrel={step:Log_mat_inequality}  \frac{1}{2} \sum_{\n=0}^{\nfin} \sum_{i=1}^{\statedim} \log(\prod_{d=1}^{\numdataPts_{\n}}( 1 + \sigma^{-2} \lambda_{\numdataPts, \n, i})) \\
&= \frac{1}{2} \sum_{\n=0}^{\nfin} \sum_{i=1}^{\statedim} \sum_{d=1}^{\numdataPts_{\n}} \log( 1 + \sigma^{-2} \lambda_{\numdataPts, \n, i}). \numberthis \label{eqn: macopt-information-gain}
\end{align*}
Step~\eqref{step:Log_mat_inequality} follows by orthonormal eigen decomposition of the kernel matrix, where $\lambda_{d,n,i}$ are the eigenvalues of $K^i_{d_n,d_n}$.
\end{proof}

%% file: appx/lemma2_Alternative.tex
\begin{proof} 
\looseness -1 
This analysis follows from \citet{beta_srinivas}, and we show that only a limited number of measurements can be gathered until \eqref{eq:sampling_rule_timex} becomes infeasible. As defined earlier, confidence width $\cwidth[\n, i](\stateAct)= 2 \sqrt{\betaconst[\n]} \sigconst[\n, i](\stateAct)$. 
First, note that a finite $n^\star$ exists due to the assumed sublinear growth of $\beta_n\gamma_n$, cf. Assumption~\ref{assump:sublinear}.

A key difference from \citet{beta_srinivas,prajapat2025safe} is that we consider $\statedim$ vector-valued GPs and allow the model to be updated with up to $\totHorizon$ measurements simultaneously. 
Following analysis from \citet{prajapat2022near}, we can bound the uncertainty at location $z_{\numdataPts,\n, i}$ after $\n-1$ model updates as,
\begin{align*}
        \sum_{i=1}^{\statedim}
        \sum_{\numdataPts \in \numdataPts^\epsilon_{\n}}  \cwidth[\n-1, i]^2(\stateAct_{\numdataPts, \n})  \! 
        &= \sum_{i=1}^{\statedim}
        \sum_{\numdataPts \in \numdataPts^\epsilon_{\n}} 4 \betaconst[\n] \sigconst[\n-1,i]^2(\stateAct_{\numdataPts, \n}) \labelrel\leq{step:all_meas_bound} \sum_{i=1}^{\statedim}
        \sum_{\numdataPts=1}^{\numdataPts_{\n}} 4 \betaconst[\n] \sigconst[\n-1,i]^2(\stateAct_{\numdataPts, \n}) \\
        &= 4 \betaconst[\n]  
        \sum_{i=1}^{\statedim}  \mathrm{trace}(K^{i}_{d_n,d_n}) = 4 \betaconst[\n]  
        \sum_{i=1}^{\statedim}  \sum_{\numdataPts=1}^{\numdataPts_{\n}} 
        \lambda_{\numdataPts, \n, i}\\
        &\labelrel\leq{step:eig_bound}  4 \betaconst[\n]  \sum_{i=1}^{\statedim}  \sum_{\numdataPts=1}^{\numdataPts_{\n}} \sigma^2 C_2 \log(1+\noiseconst \lambda_{\numdataPts, \n, i}) \\ 
        &\leq  C_1\betaconst[\n] \sum_{i=1}^{\statedim}  \sum_{\numdataPts=1}^{\numdataPts_{\n}}\frac{1}{2} \log(1+\noiseconst \lambda_{\numdataPts, \n, i}) \numberthis \label{eq:cw_bound}.
\end{align*}
\noindent Step~\eqref{step:all_meas_bound} follows since $\sigconst[\n-1,i]^2(\stateAct) \geq 0, \forall z$ and $|\numdataPts^\epsilon_{\n}| \leq \numdataPts_{\n}$. 
Step~\eqref{step:eig_bound} uses the fact that $s \leq C_2 \log(1+s)$ where $C_2 = \totHorizon \noiseconst/\log(1+\totHorizon\noiseconst) \geq 1$ $ \forall s \in [0, \noiseconst\totHorizon]$ with $s=\sigma^{-2}\lambda_{d,n,i}$. 
Note that $\lambda_{d,n,i}\leq \mathrm{trace}(K^{i}_{d_n,d_n})\leq d_n\leq \totHorizon$, using the fact that  $k_{D_n}(x,x)\leq k(x,x)\leq 1$, see Assumption~\ref{assump:q_RKHS}.
Finally, substituting $C_1 = 8\totHorizon/\log(1+ \totHorizon\noiseconst)$ results in~\cref{eq:cw_bound}.

The sampling rule~\eqref{eq:sampling_rule_timex}, i.e.,  collecting measurements 
until any one of the
components has uncertainty above $\epscollect$, implies
\begin{align}
\sum_{\numdataPts=1}^{|\numdataPts^\epsilon_{\n}|} \epscollect^2 \leq  \sum_{i=1}^{\statedim}
        \sum_{\numdataPts \in \numdataPts^\epsilon_{\n}} \! \cwidth[\n-1, i]^2(\stateAct_{\numdataPts, \n}) &\leq 
        C_1 \betaconst[\n] \sum_{i=1}^{\statedim}  \sum_{\numdataPts=1}^{\numdataPts_{\n}}\frac{1}{2} \log(1+\noiseconst \lambda_{\numdataPts, \n, i})) 
\end{align}
Analogous to \citet{prajapat2025safe}, suppose we sample for $\nfin$ iterations before terminating, which implies
\begin{align}
    \sum_{n=1}^{\nfin} |d^\epsilon_n| \epscollect^2 \leq  C_1 \sum_{\n=1}^{\nfin} \sum_{i=1}^{\statedim} \betaconst[\n,i] 
    &\sum_{\numdataPts=1}^{\numdataPts_{\n}}
    \frac{1}{2} \log(1+\noiseconst \lambda_{\numdataPts, \n, i})) \leq C_1 \betaconst[\nfin] I(Y_{\D_\nfin};\dyntrueVec_{\D_\nfin})) \leq C_1 \betaconst[\nfin] \gamma_{\nfin},
    \label{eqn:mono-gamma}
\end{align}
where we used monotonicity of $\beta_{n,i}, \forall i \in \Intrange{1}{\statedim}$, denote  $\beta_{n} = \max_{i} \beta_{n,i}$, the definition of the maximal information capacity $\gamma_n$ and the mutual information according to \cref{lem:mutual_information}:
\begin{align*}
    I(Y_{\D_\nfin};\dyntrueVec_{\D_\nfin})= \sum_{\n=0}^{\nfin} \sum_{i=1}^{\statedim} 
    \sum_{\numdataPts=1}^{\numdataPts_{\n}}
    \frac{1}{2} \log(1+\noiseconst \lambda_{\numdataPts, \n, i})).
\end{align*}
This implies \vspace{-2em}
\begin{align}
  &\frac{\nfin}{\betaconst[\nfin] \gamma_{\nfin}}\leq \frac{1}{\betaconst[\nfin] \gamma_{\nfin}} \sum_{n=1}^{\nfin} |\numdataPts^\epsilon_n| \leq \frac{C_1 }{\epscollect^2} \label{eqn:sample_complexity}
\end{align}
which ensures $\nfin\leq n^\star$ by definition of $n^\star$. 


Note that the condition~\eqref{eqn:lemma_termination_condition} in the lemma corresponds to the termination criterion in \cref{alg:full_domain_exploration_basic} and hence the algorithm terminates after at most $\nfin$ iterations.
\end{proof}

%% file: figure/dyn_opti_in_pessi.tex
\tikzset{every picture/.style={line width=0.75pt}} 

\begin{tikzpicture}[x=0.75pt,y=0.75pt,yscale=-1,xscale=1]

\draw  [fill={rgb, 255:red, 128; green, 128; blue, 128 }  ,fill opacity=0.3 ] (90,118.2) .. controls (90,91.86) and (111.36,70.5) .. (137.7,70.5) -- (433.8,70.5) .. controls (460.14,70.5) and (481.5,91.86) .. (481.5,118.2) -- (481.5,261.3) .. controls (481.5,287.64) and (460.14,309) .. (433.8,309) -- (137.7,309) .. controls (111.36,309) and (90,287.64) .. (90,261.3) -- cycle ;
\draw  [fill={rgb, 255:red, 245; green, 166; blue, 35 }  ,fill opacity=0.72 ] (126,122) .. controls (138,111.5) and (205,105.6) .. (227,113.2) .. controls (249,120.8) and (357,70.4) .. (391.4,76) .. controls (425.8,81.6) and (470.74,270.31) .. (453.4,282.8) .. controls (436.06,295.29) and (369.66,303.65) .. (339.8,304.8) .. controls (309.94,305.95) and (221,278.8) .. (210.2,276.8) .. controls (199.4,274.8) and (132.1,255.6) .. (125,230) .. controls (117.9,204.4) and (114,132.5) .. (126,122) -- cycle ;
\draw  [fill={rgb, 255:red, 108; green, 215; blue, 108 }  ,fill opacity=0.6 ] (126,122) .. controls (138,111.5) and (194.2,85.4) .. (241,94) .. controls (287.8,102.6) and (391,193.67) .. (386,213.67) .. controls (381,233.67) and (226.6,272) .. (210.2,276.8) .. controls (193.8,281.6) and (132.1,255.6) .. (125,230) .. controls (117.9,204.4) and (114,132.5) .. (126,122) -- cycle ;
\draw  [line width=3.75]  (401.24,175.05) .. controls (401.24,174.34) and (401.81,173.77) .. (402.52,173.77) .. controls (403.22,173.77) and (403.8,174.34) .. (403.8,175.05) .. controls (403.8,175.75) and (403.22,176.32) .. (402.52,176.32) .. controls (401.81,176.32) and (401.24,175.75) .. (401.24,175.05) -- cycle ;
\draw  [line width=3.75]  (362.64,179.25) .. controls (362.64,178.54) and (363.21,177.97) .. (363.92,177.97) .. controls (364.62,177.97) and (365.2,178.54) .. (365.2,179.25) .. controls (365.2,179.95) and (364.62,180.52) .. (363.92,180.52) .. controls (363.21,180.52) and (362.64,179.95) .. (362.64,179.25) -- cycle ;
\draw [line width=1.5]  [dash pattern={on 5.63pt off 4.5pt}]  (367.89,178.81) -- (398.54,175.48) ;
\draw [shift={(402.52,175.05)}, rotate = 173.79] [fill={rgb, 255:red, 0; green, 0; blue, 0 }  ][line width=0.08]  [draw opacity=0] (8.13,-3.9) -- (0,0) -- (8.13,3.9) -- cycle    ;
\draw [shift={(363.92,179.25)}, rotate = 353.79] [fill={rgb, 255:red, 0; green, 0; blue, 0 }  ][line width=0.08]  [draw opacity=0] (8.13,-3.9) -- (0,0) -- (8.13,3.9) -- cycle    ;

\draw (229,173.13) node [anchor=north west][inner sep=0.75pt]  [font=\LARGE]  {$\upPi _{n}^{p}$};
\draw (372.07,252.87) node [anchor=north west][inner sep=0.75pt]  [font=\LARGE]  {$\upPi _{n}^{o,\frac{\epsilon}{2}}$};
\draw (416.33,171.93) node [anchor=north west][inner sep=0.75pt]  [font=\LARGE]  {$\pi ^{b}$};
\draw (330.53,171.53) node [anchor=north west][inner sep=0.75pt]  [font=\LARGE]  {$\pi ^{p}$};
\draw (370.62,150.07) node [anchor=north west][inner sep=0.75pt]  [font=\LARGE]  {$\epsclose$};
\draw (100.43,88) node [anchor=north west][inner sep=0.75pt]  [font=\Huge]  {$\upPi $};

\end{tikzpicture}

%% file: figure/dyn_traj_ball_eps.tex
\tikzset{every picture/.style={line width=0.75pt}} 

\begin{tikzpicture}[x=0.75pt,y=0.75pt,yscale=-1,xscale=1]

\draw  [fill={rgb, 255:red, 80; green, 227; blue, 194 }  ,fill opacity=1 ] (1179.75,147) .. controls (1179.75,128.06) and (1195.11,112.7) .. (1214.05,112.7) .. controls (1233,112.7) and (1248.35,128.06) .. (1248.35,147) .. controls (1248.35,165.94) and (1233,181.3) .. (1214.05,181.3) .. controls (1195.11,181.3) and (1179.75,165.94) .. (1179.75,147) -- cycle ;
\draw [color={rgb, 255:red, 245; green, 166; blue, 35 }  ,draw opacity=1 ][line width=1.5]  [dash pattern={on 5.63pt off 4.5pt}]  (1046.7,251.85) .. controls (1087.39,257.67) and (1090.05,244.33) .. (1098.05,222.33) .. controls (1106.05,200.33) and (1146.05,185.67) .. (1164.05,196.33) .. controls (1182.05,207) and (1201.39,182.33) .. (1214.05,147) ;
\draw [shift={(1086,246.35)}, rotate = 143.42] [fill={rgb, 255:red, 245; green, 166; blue, 35 }  ,fill opacity=1 ][line width=0.08]  [draw opacity=0] (13.4,-6.43) -- (0,0) -- (13.4,6.44) -- (8.9,0) -- cycle    ;
\draw [shift={(1132.92,195.19)}, rotate = 161.94] [fill={rgb, 255:red, 245; green, 166; blue, 35 }  ,fill opacity=1 ][line width=0.08]  [draw opacity=0] (13.4,-6.43) -- (0,0) -- (13.4,6.44) -- (8.9,0) -- cycle    ;
\draw [shift={(1200.35,176.44)}, rotate = 122.96] [fill={rgb, 255:red, 245; green, 166; blue, 35 }  ,fill opacity=1 ][line width=0.08]  [draw opacity=0] (13.4,-6.43) -- (0,0) -- (13.4,6.44) -- (8.9,0) -- cycle    ;
\draw  [fill={rgb, 255:red, 0; green, 0; blue, 0 }  ,fill opacity=1 ][line width=3.75]  (1046.7,251.85) .. controls (1046.7,251.14) and (1047.28,250.57) .. (1047.98,250.57) .. controls (1048.69,250.57) and (1049.26,251.14) .. (1049.26,251.85) .. controls (1049.26,252.55) and (1048.69,253.12) .. (1047.98,253.12) .. controls (1047.28,253.12) and (1046.7,252.55) .. (1046.7,251.85) -- cycle ;
\draw [color={rgb, 255:red, 0; green, 0; blue, 0 }  ,draw opacity=1 ][line width=1.5]  [dash pattern={on 5.63pt off 4.5pt}]  (1047.98,250.57) .. controls (1091.2,237.26) and (1091.78,229.65) .. (1088.79,210.72) .. controls (1085.8,191.79) and (1130.01,171.21) .. (1155.6,172.6) .. controls (1181.19,173.99) and (1184.46,156.83) .. (1190.18,128.65) ;
\draw [shift={(1083.99,234.51)}, rotate = 142.74] [fill={rgb, 255:red, 0; green, 0; blue, 0 }  ,fill opacity=1 ][line width=0.08]  [draw opacity=0] (13.4,-6.43) -- (0,0) -- (13.4,6.44) -- (8.9,0) -- cycle    ;
\draw [shift={(1121.97,179.24)}, rotate = 157.25] [fill={rgb, 255:red, 0; green, 0; blue, 0 }  ,fill opacity=1 ][line width=0.08]  [draw opacity=0] (13.4,-6.43) -- (0,0) -- (13.4,6.44) -- (8.9,0) -- cycle    ;
\draw [shift={(1184.43,152.68)}, rotate = 111.06] [fill={rgb, 255:red, 0; green, 0; blue, 0 }  ,fill opacity=1 ][line width=0.08]  [draw opacity=0] (13.4,-6.43) -- (0,0) -- (13.4,6.44) -- (8.9,0) -- cycle    ;
\draw  [fill={rgb, 255:red, 0; green, 0; blue, 0 }  ,fill opacity=1 ][line width=3.75]  (1188.9,128.65) .. controls (1188.9,127.94) and (1189.48,127.37) .. (1190.18,127.37) .. controls (1190.89,127.37) and (1191.46,127.94) .. (1191.46,128.65) .. controls (1191.46,129.35) and (1190.89,129.92) .. (1190.18,129.92) .. controls (1189.48,129.92) and (1188.9,129.35) .. (1188.9,128.65) -- cycle ;
\draw  [color={rgb, 255:red, 245; green, 166; blue, 35 }  ,draw opacity=1 ][fill={rgb, 255:red, 245; green, 166; blue, 35 }  ,fill opacity=1 ][line width=3.75]  (1212.78,147) .. controls (1212.78,146.29) and (1213.35,145.72) .. (1214.05,145.72) .. controls (1214.76,145.72) and (1215.33,146.29) .. (1215.33,147) .. controls (1215.33,147.71) and (1214.76,148.28) .. (1214.05,148.28) .. controls (1213.35,148.28) and (1212.78,147.71) .. (1212.78,147) -- cycle ;
\draw [line width=1.5]  [dash pattern={on 5.63pt off 4.5pt}]  (1214.05,116.7) -- (1214.05,143) ;
\draw [shift={(1214.05,147)}, rotate = 270] [fill={rgb, 255:red, 0; green, 0; blue, 0 }  ][line width=0.08]  [draw opacity=0] (8.13,-3.9) -- (0,0) -- (8.13,3.9) -- cycle    ;
\draw [shift={(1214.05,112.7)}, rotate = 90] [fill={rgb, 255:red, 0; green, 0; blue, 0 }  ][line width=0.08]  [draw opacity=0] (8.13,-3.9) -- (0,0) -- (8.13,3.9) -- cycle    ;

\draw (1110.06,229.7) node [anchor=north west][inner sep=0.75pt]  [font=\LARGE,color={rgb, 255:red, 245; green, 166; blue, 35 }  ,opacity=1 ]  {$x_{h+1}^{o} =\dynVec^{o}\left( x_{h}^{o} ,\pi _{h}^{b}\left( x_{h}^{o}\right)\right)$};
\draw (964.06,101.1) node [anchor=north west][inner sep=0.75pt]  [font=\LARGE,color={rgb, 255:red, 0; green, 0; blue, 0 }  ,opacity=1 ]  {$x_{h+1}^{s} =\dynVec^{s}\left( x_{h}^{s} ,\pi _{h}^{b}\left( x_{h}^{s}\right)\right)$};
\draw (1226.42,85.0) node [anchor=north west][inner sep=0.75pt]  [font=\LARGE]  {$\frac{\epsilon}{2}$};

\end{tikzpicture}

%% file: figure/dyn_init_set_appending.tex
\tikzset{every picture/.style={line width=0.75pt}} 

\begin{tikzpicture}[x=0.75pt,y=0.75pt,yscale=-1,xscale=1]

\draw  [fill={rgb, 255:red, 108; green, 215; blue, 108 }  ,fill opacity=0.6 ] (149.18,101.26) .. controls (216.18,113.16) and (281.18,97.85) .. (352.18,106.35) .. controls (423.18,114.85) and (418.18,202.85) .. (419.18,275.35) .. controls (420.18,347.85) and (236.58,307.55) .. (220.18,312.35) .. controls (203.78,317.15) and (112.78,292.95) .. (105.68,267.35) .. controls (98.58,241.75) and (82.18,89.35) .. (149.18,101.26) -- cycle ;
\draw  [fill={rgb, 255:red, 128; green, 128; blue, 128 }  ,fill opacity=0.3 ] (72,129.72) .. controls (72,103.28) and (93.43,81.85) .. (119.87,81.85) -- (387.31,81.85) .. controls (413.75,81.85) and (435.18,103.28) .. (435.18,129.72) -- (435.18,273.33) .. controls (435.18,299.77) and (413.75,321.2) .. (387.31,321.2) -- (119.87,321.2) .. controls (93.43,321.2) and (72,299.77) .. (72,273.33) -- cycle ;
\draw  [fill={rgb, 255:red, 108; green, 215; blue, 108 }  ,fill opacity=0.6 ] (152,111) .. controls (212.5,126) and (271.68,104.35) .. (342.68,112.85) .. controls (413.68,121.35) and (409.68,195.85) .. (410.68,268.35) .. controls (411.68,340.85) and (230.4,298.2) .. (214,303) .. controls (197.6,307.8) and (124.6,286.6) .. (117.5,261) .. controls (110.4,235.4) and (91.5,96) .. (152,111) -- cycle ;
\draw  [fill={rgb, 255:red, 108; green, 215; blue, 108 }  ,fill opacity=0.6 ] (164.68,118.26) .. controls (225.18,133.26) and (264.01,117.6) .. (333.68,122.85) .. controls (403.35,128.09) and (393.68,193.35) .. (397.18,260.35) .. controls (400.68,327.35) and (246.68,289.26) .. (220.68,294.26) .. controls (194.68,299.26) and (130.18,278.76) .. (121.68,237.26) .. controls (113.18,195.76) and (104.18,103.26) .. (164.68,118.26) -- cycle ;
\draw  [fill={rgb, 255:red, 80; green, 227; blue, 194 }  ,fill opacity=0.67 ] (152.33,255.57) .. controls (138.06,225.15) and (144.88,188.23) .. (167.56,173.1) .. controls (190.24,157.96) and (220.18,170.35) .. (234.45,200.76) .. controls (248.72,231.17) and (241.9,268.09) .. (219.22,283.23) .. controls (196.54,298.36) and (166.6,285.98) .. (152.33,255.57) -- cycle ;
\draw  [line width=3.75]  (211.18,261.66) .. controls (211.18,260.45) and (212.16,259.48) .. (213.36,259.48) .. controls (214.57,259.48) and (215.54,260.45) .. (215.54,261.66) .. controls (215.54,262.86) and (214.57,263.84) .. (213.36,263.84) .. controls (212.16,263.84) and (211.18,262.86) .. (211.18,261.66) -- cycle ;
\draw [line width=1.5]    (203.18,270.76) .. controls (253.18,303.26) and (301.68,312.76) .. (335.68,298.76) .. controls (369.68,284.76) and (380.46,202.27) .. (348.18,172.76) .. controls (316.38,143.68) and (260.39,127.74) .. (206.63,192.73) ;
\draw [shift={(204.18,195.76)}, rotate = 308.51] [fill={rgb, 255:red, 0; green, 0; blue, 0 }  ][line width=0.08]  [draw opacity=0] (11.61,-5.58) -- (0,0) -- (11.61,5.58) -- cycle    ;
\draw    (171.78,238.5) .. controls (173.77,237.23) and (175.41,237.57) .. (176.71,239.51) .. controls (178.24,241.39) and (179.89,241.53) .. (181.66,239.93) .. controls (183.51,238.3) and (185.16,238.48) .. (186.62,240.48) .. controls (187.71,242.59) and (189.28,243.13) .. (191.34,242.12) .. controls (193.6,241.29) and (195.09,242.05) .. (195.82,244.42) .. controls (195.84,246.63) and (196.95,247.86) .. (199.15,248.1) .. controls (201.36,249.09) and (201.85,250.68) .. (200.64,252.85) .. controls (199.17,254.61) and (199.33,256.25) .. (201.11,257.78) -- (201.16,258.72) -- (201.62,266.61) ;
\draw [shift={(201.9,269.48)}, rotate = 263.55] [fill={rgb, 255:red, 0; green, 0; blue, 0 }  ][line width=0.08]  [draw opacity=0] (8.93,-4.29) -- (0,0) -- (8.93,4.29) -- cycle    ;
\draw  [line width=3.75]  (201.75,195.76) .. controls (201.75,194.41) and (202.84,193.33) .. (204.18,193.33) .. controls (205.52,193.33) and (206.61,194.41) .. (206.61,195.76) .. controls (206.61,197.1) and (205.52,198.18) .. (204.18,198.18) .. controls (202.84,198.18) and (201.75,197.1) .. (201.75,195.76) -- cycle ;
\draw [line width=1.5]  [dash pattern={on 5.63pt off 4.5pt}]  (212.18,260.76) .. controls (255.43,300.76) and (351.68,330.76) .. (367.18,230.26) .. controls (382.37,131.77) and (240.36,118.77) .. (207.38,174.76) ;
\draw [shift={(205.5,178.28)}, rotate = 295.8] [fill={rgb, 255:red, 0; green, 0; blue, 0 }  ][line width=0.08]  [draw opacity=0] (11.61,-5.58) -- (0,0) -- (11.61,5.58) -- cycle    ;
\draw   (194.67,260.76) .. controls (194.67,250.38) and (203.08,241.96) .. (213.46,241.96) .. controls (223.84,241.96) and (232.25,250.38) .. (232.25,260.76) .. controls (232.25,271.13) and (223.84,279.55) .. (213.46,279.55) .. controls (203.08,279.55) and (194.67,271.13) .. (194.67,260.76) -- cycle ;
\draw  [line width=3.75]  (199.47,269.48) .. controls (199.47,268.14) and (200.56,267.05) .. (201.9,267.05) .. controls (203.24,267.05) and (204.33,268.14) .. (204.33,269.48) .. controls (204.33,270.82) and (203.24,271.9) .. (201.9,271.9) .. controls (200.56,271.9) and (199.47,270.82) .. (199.47,269.48) -- cycle ;
\draw  [line width=3.75]  (203.75,177.76) .. controls (203.75,176.41) and (204.84,175.33) .. (206.18,175.33) .. controls (207.52,175.33) and (208.61,176.41) .. (208.61,177.76) .. controls (208.61,179.1) and (207.52,180.18) .. (206.18,180.18) .. controls (204.84,180.18) and (203.75,179.1) .. (203.75,177.76) -- cycle ;
\draw  [line width=3.75]  (167.75,237.76) .. controls (167.75,236.41) and (168.84,235.33) .. (170.18,235.33) .. controls (171.52,235.33) and (172.61,236.41) .. (172.61,237.76) .. controls (172.61,239.1) and (171.52,240.18) .. (170.18,240.18) .. controls (168.84,240.18) and (167.75,239.1) .. (167.75,237.76) -- cycle ;
\draw [line width=1.5]    (389.58,288.25) -- (403.63,299.25) ;
\draw [shift={(406.78,301.72)}, rotate = 218.07] [fill={rgb, 255:red, 0; green, 0; blue, 0 }  ][line width=0.08]  [draw opacity=0] (6.97,-3.35) -- (0,0) -- (6.97,3.35) -- cycle    ;
\draw [shift={(386.43,285.78)}, rotate = 38.07] [fill={rgb, 255:red, 0; green, 0; blue, 0 }  ][line width=0.08]  [draw opacity=0] (6.97,-3.35) -- (0,0) -- (6.97,3.35) -- cycle    ;
\draw [line width=1.5]    (386.68,305.45) -- (390.26,308.5) ;
\draw [shift={(393.31,311.09)}, rotate = 220.41] [fill={rgb, 255:red, 0; green, 0; blue, 0 }  ][line width=0.08]  [draw opacity=0] (4.64,-2.23) -- (0,0) -- (4.64,2.23) -- cycle    ;
\draw [shift={(383.63,302.86)}, rotate = 40.41] [fill={rgb, 255:red, 0; green, 0; blue, 0 }  ][line width=0.08]  [draw opacity=0] (4.64,-2.23) -- (0,0) -- (4.64,2.23) -- cycle    ;

\draw (154.43,195.53) node [anchor=north west][inner sep=0.75pt]  [font=\large]  {$\mathbb{X}_{n}$};
\draw (79.23,96.87) node [anchor=north west][inner sep=0.75pt]  [font=\Huge]  {$\mathcal{X}$};
\draw (152.63,236.83) node [anchor=north west][inner sep=0.75pt]  [font=\large]  {$x_{s}$};
\draw (213.63,246.83) node [anchor=north west][inner sep=0.75pt]  [font=\large]  {$x'$};
\draw (169.83,259.23) node [anchor=north west][inner sep=0.75pt]  [font=\large]  {$x_{0}^{\star ,l}$};
\draw (200.33,200.33) node [anchor=north west][inner sep=0.75pt]  [font=\large]  {$x_{H}^{\star ,l}$};
\draw (176.13,170.33) node [anchor=north west][inner sep=0.75pt]  [font=\large]  {$x_{H}^{\star }$};
\draw (257.73,271.63) node [anchor=north west][inner sep=0.75pt]  [font=\large]  {$\pi $};
\draw (180.93,223.33) node [anchor=north west][inner sep=0.75pt]  [font=\large]  {$\hat{\pi }$};
\draw (396.02,273.23) node [anchor=north west][inner sep=0.75pt]  [font=\LARGE]  {$\epsilon $};
\draw (357.25,305.62) node [anchor=north west][inner sep=0.75pt]  [font=\LARGE,rotate=-359.37]  {$\epsilon /2$};

\end{tikzpicture}

%% file: appx/goal_directed_exploration.tex
\section{Reward maximization with intrinsic exploration}
\label{apxsec:rew_maxim}

This section provides a more detailed explanation of the proposed reward maximization algorithm \dynExplor in~\cref{apxsec:rew_maxim_1}, and then we provide the proof of the theoretical guarantees in~\cref{apxsec:proof_of_rew_maxim}.

\mypar{Choice of tolerance $\epsguarantee,\epspessi,\epsterminal,\epscollect$} 
As in the maximum dynamics exploration setting (\cref{sec:full_expl_theory}), in \cref{thm:rew_maxim}, tolerance $\epsguarantee>0$ can be chosen arbitrarily small, constrained only by the noise magnitude $\noisebound$.
The constants $\epscollect$ and $\epsterminal$ serve analogous roles — $\epscollect>0$ sets the threshold for informative measurement, while $\epsterminal$, used in \cref{alg:rew_maxim} for planning is selected as in~\eqref{eq:JK_formula_epsilon_d}.
The tolerance $\epspessi$ that tightens the constraints for the pessimistic policy set can be chosen satisfying \eqref{eq:epsilon_lower_bound_convergence}, \eqref{eq:eps_dash_formula}, and \eqref{eq:pessi_eps_dash_shrink}.
Lastly, the tolerance $\epsguarantee>0$ for certifying~\cref{obj:reward_maximization} can be chosen according to~\eqref{eq:esp_diff_formula}, and the \cref{obj:reward_maximization} holds with constant $K$ defined as per~\eqref{eq:constantK}. 
Together, these formulas provide a lower bound on the achievable performance with $\epsguarantee>0$ that scales linearly with the noise magnitude $\noisebound$. In the noise-free case, $\noisebound=0$, the tolerance $\epsguarantee>0$ can be chosen arbitrarily small.

\subsection{Detailed explanation on \dynExplor algorithm}
\label{apxsec:rew_maxim_1}
\dynExplor revolves mainly around~\cref{lin:solve_relax_prob}, where at each time step $\ki$, the agent solves Problem~\eqref{eq:slack_there_exists_goal} using it's current state $\state(\ki)$ and a task-driven objective $J^{\mathrm{any}}$. 
The agent executes the resulting policy $\pessiPolicy[e]$, collects new measurements until having sufficient information, and then updates the model. It instantaneously replans (resolves Problem~\eqref{eq:slack_there_exists_goal}) with the updated model at the current state $\state(\ki)$.  
As shown in~\cref{fig:dyn_receding_main}, a key component is that the agent always ensures the existence of a safe return path without needing to execute it explicitly.
Over time, with data, the agent learns and gets closer to the optimal behavior with iterations.
The agent continues to explore until the dynamics is sufficiently well learned to achieve the objective of the task. 
Finally, when the agent is guaranteed to achieve close-to-optimal behavior, it returns and executes the safe close-to-optimal policy. 

The approach does an intrinsic exploration while maximizing a task-driven objective $J^{\mathrm{any}}$. This objective can be heuristically designed, as discussed in ~\cref{rem:Jany}. It can naturally blend exploration and task-oriented behavior, for example, minimizing distance to the optimistic trajectory. 

To determine when the dynamics are sufficiently learned (guaranteed to exhibit close to optimal behavior), we need a termination criterion. There are particularly two challenges for termination criteria: \\
i) should \emph{not} require a uniform reduction in uncertainty (less than $\epsterminal$) throughout\\
ii) should guarantee bounded regret (with returned policy) despite not uniformly knowing the dynamics along the executed trajectory.

To overcome this, we define a termination criterion that involves solving the pessimistic problem~\eqref{eq:pessi_obj} with the following objective,
\begin{align} \label{eq:pessi_objective}
    \Jobj[\mathrm{p}]{\state, \bm{\mu}_{\n}}{\pi} \coloneqq \Jobj[]{\state, \bm{\mu}_\n}{\pi} - \LiprewPi  \sum_{\h=0}^{\Horizon-1}  \sum_{i=0}^{\h} L^{i} \cwidth[\n](\state_{i}, \pi_{i}(\state_{i})),
\end{align}
representing a pessimistic estimate (lower bound) of rewards (c.f. \cref{lem:optimality_with_return_policy}), ensuring that on termination, even for the worst realization (in terms of objective) of the dynamics will achieve close to $\Jobj[\mathrm{p}]{\state, \bm{\mu}_{\n}}{\pi}$ of the true trajectory. With this, we define the termination criteria as follows,
\begin{align}\label{eq:termination_criteria}
\Jobj[\mathrm{p}]{\statePessi, \bm{\mu}_\n}{\pessiPolicy} \geq  \Jobj[]{\state^{\mathrm{o}}, \dynVec^o }{\optiPolicy} - 3 \LiprewPi \sum_{\h=0}^{\Horizon-1} \LipDyn_{\h} \epsterminal,
\end{align}
where $\state^\mathrm{p}$,$\pessiPolicy$ are the optimal solution of the pessimistic problem~\eqref{eq:pessi_obj} and $\state^{\mathrm{o}}, \dynVec^o,\optiPolicy$ are the optimal solution of the optimistic problem~\eqref{eq:opti_obj}.
Notably, while the small uncertainty (smaller than tolerance $\epsterminal$) in the true trajectory after executing the returned policy is a sufficient stopping condition, 
the termination criterion~\eqref{eq:termination_criteria} does not necessarily require this.


One may also employ alternative termination criteria, which may improve computational efficiency by involving fewer optimization problems to be solved. For example, if the model uncertainty along the optimistic trajectory is uniformly below $\epsterminal$, then executing the corresponding optimistic policy is sufficient to ensure that the true system performs close to the optimistic estimate—and therefore close to the optimal behaviour. This termination criterion avoids solving the pessimistic problem ~\eqref{eq:pessi_obj}, but may require more iterations to satisfy the uncertainty condition compared to the termination criterion in \eqref{eq:termination_criteria}. 

While the termination criterion in \eqref{eq:termination_criteria} requires solving both the optimistic and pessimistic problems, these problems do not directly influence the exploration process and are only used for checking termination.
As a result, it need not be evaluated at every iteration 
$\n$; instead, it can be checked periodically or heuristically, depending on the task and available computational budget.

\subsection{Proof of \texorpdfstring{\cref{thm:rew_maxim}}{Lg}}
\label{apxsec:proof_of_rew_maxim}
The following lemma bounds the deviation between the optimistic and the true dynamics.
\begin{lemma}\label{lem:deviation_with_uncertainity}
Let \cref{assump:lipschitz} hold and suppose that $\dyntrueVec\in\dynSet[_n]$. 
Consider any policy $\pi \in \Pi_{\Horizon}$, dynamics $\dynVec\in\dynSet[n]$, initial condition $x_0^\star=x_0^o\in\mathcal{X}$, and the two trajectories $\state^\star_{\h+1} = \dynVec^\star(\state^\star_\h, \pi_\h(\state^\star_\h))$, $\stateOpti_{\h+1} = \dynVec(\stateOpti_\h, \pi_\h(\stateOpti_\h))$. 
Then, for any $\h \in \Intp$ it holds that,
 \begin{align}
     \|\state^{\star}_{\h+1} - \state^{\mathrm{o}}_{\h+1} \| \leq \sum_{i=0}^{\h}\LipDyn^{i} \cwidth[\n-1](\state_{i}, \pi_{i}(\state_{i})) 
 \end{align}
\end{lemma}
\begin{proof} Consider,
    \begin{align*}
    \|\state^{\star}_{\h+1} - \state_{\h+1} \| 
    &\leq \| \dyntrueVec(\state^{\star}_{\h}, \pi_{\h}(\state^{\star}_{\h})) - \dynVec(\state_{\h}, \pi_{\h}(\state_{\h})) \| \\
    &\leq  \| \dyntrueVec(\state^{\star}_{\h}, \pi_{\h}(\state^{\star}_{\h})) - \dyntrueVec(\state_{\h}, \pi_{\h}(\state_{\h})) \| + \| \dyntrueVec(\state_{\h}, \pi_{\h}(\state_{\h}))  -\dynVec(\state_{\h}, \pi_{\h}(\state_{\h})) \| \\
    &\labelrel\leq{step:epstermbounded} \LipDyn \|\state^{\star}_{\h} - \state_{\h} \| + \cwidth[\n-1](\state_{\h}, \pi_{\h}(\state_{\h}))  \\
    &\vdots \\
    &\leq \sum_{i=0}^{\h} \LipDyn^{i} \cwidth[\n-1](\state_{i}, \pi_{i}(\state_{i})) \label{eq:convergence_x_bound} \numberthis
\end{align*}
Step \eqref{step:epstermbounded} follows since $\| \dyntrueVec(\state_{\h}, \pi_{\h}(\state_{\h})) - \dynVec(\state_{\h}, \pi_{\h}(\state_{\h})) \| \leq \cwidth[\n](\state_{\h}, \pi_{\h}(\state_{\h}))$ and Lipschtiz continuity of the closed-loop system $\dyntrueVec$. The final inequality follows using recursion.
\end{proof}



The following lemma demonstrates that the policy $[\hat{\pi},\pessiPolicy]$ returned by \cref{alg:rew_maxim} ensures the safety of the unknown system.

\begin{lemma}[Safety with returned policy] \label{lem:safety_with_return_policy}
Let~\cref{assump:q_RKHS,assump:safeSet,assump:lipschitz,assump:sublinear,assump:dynamics_tol_terminal} hold. 
Let $\hat{\pi} \in \Pi_{\delta \h}$ control the mean dynamics from $\state(\ki)$ to $\statePessi$ and $\pessiPolicy\in \pessiSet[,\epssafeset]{\n}(\statePessi; \Horizon)$ obtained by solving \eqref{eq:opti_obj}. Then the returned policy $\retPolicy \coloneqq [\hat{\pi}, \pessiPolicy] \in \truePolicySet[]{\n}(\state(\ki); \Horizon + \delta \h)$, that is, the resulting closed-loop system from the returned policy
satisfies constraints and ends in the safe set.
\end{lemma}
\begin{proof}
Given \cref{assump:lipschitz}, the closed-loop dynamics of $\dyntrueVec$ with policy $\pi$ is $\LipDyn$-Lipschitz continuous.
Similarly to \cref{lem:horizon_diff},
let $\state^{\star,l}_h$ and $\state^\star_h$ denote the state sequences when applying the policy $\pessiPolicy$, and the same noise sequence $\eta_h$ to the dynamics $\dyntrueVec$ with initial conditions $x_0^{\star,l}$ and $x_0^\star$ , respectively.
Due to controllability within the safe set (\cref{assump:safeSet}), unknown dynamics up to $\epsdiff$ tolerance in the safe set (\cref{assump:dynamics_tol_terminal}) and noise, 
the deviation satisfies $\|\state^\star_0 - \state^{\star,l}_0\| \leq 
\sum_{h=0}^{ \delta h - 1}\LipDyn^h (\epsdiff + \noisebound)$ where $\state^{\star,l}_0$ is the state at which $\dyntrueVec$ ends after controlling with $\hat{\pi}$ (policy $\hat{\pi}$ is optimized with mean dynamics $\bm{\mu}_\n$). 

Analogous to~\cref{prop:over_under_approx}, it holds that
\begin{align*}
     \|x_{h+1}^\star - x_{h+1}^{\star,l}\| &=    \| \dyntrueVec(\state^\star_{h}, \pi_\h(\state^{\star}_{h})) - \dyntrueVec(\state^{\star,l}_{\h}, \pi_\h(\state^{\star,l}_{\h}))\|
     \leq \LipDyn\|\state^{\star}_{\h}-\state^{\star,l}_{\h}\|
     \leq \LipDyn^h\|\state^\star_0-\state^{\star,l}_0\|\\
     & \leq \max_{h\in\Intrange{0}{H}} \{\LipDyn^h\} \sum_{i=0}^{ \delta \h - 1}\LipDyn^i (\epsdiff + \noisebound) \leq \epspessi
\end{align*}
and 
\begin{align*}
     \|\pi_h(\state^\star_\h) - \pi_h(\state^{\star,l}_\h) \|\leq \LipPi \| \state^\star_\h - \state^{\star,l}_\h \| \leq \epspessi,\quad \forall \h \in \Intrange{0}{\Horizon-1}.
\end{align*}
with 
\begin{align}
     \label{eq:epsilon_lower_bound_convergence}
     \epspessi \geq \max\{1, L_\pi\}  \max\{\LipDyn^{\Horizon},1\} \LipDyn_{\Delta \Horizon} (\epscollect + \noisebound).
\end{align}
Since tightening $\epspessi$ is used in \eqref{eq:opti_obj} while optimizing for $\pessiPolicy$ implies $\pessiPolicy \in \pessiSet[, \epspessi]{\n}(\state(\ki); \Horizon) \subseteq \truePolicySet[, \epspessi]{\n}(\state(\ki); \Horizon)$,  the appended policy satisfies $[\hat{\pi}, \pessiPolicy] \in \truePolicySet[]{\n}(\state(\ki); \Horizon + \delta \h)$. 

Furthermore, since $[\hat{\pi}, \pessiPolicy] \in \truePolicySet[]{\n}(\state(\ki); \Horizon + \delta \h)$, this implies that the resulting closed-loop system $\dyntrueVec \in \dynSet[\n]$ with $[\hat{\pi}, \pessiPolicy]$ satisfies constraints~\eqref{eq:constraint_set}.
\end{proof}

\begin{lemma}[Optimality with return policy] \label{lem:optimality_with_return_policy} 
Let~\cref{assump:q_RKHS,assump:safeSet,assump:lipschitz,assump:sublinear,assump:dynamics_tol_terminal} hold. 
Suppose the termination criterion \eqref{eq:termination_criteria} is satisfied.
Then the resulting closed-loop system from 
the return policy $[\hat{\pi}, \pessiPolicy]$ 
guarantees~\cref{obj:reward_maximization} with constant 
\begin{align} \label{eq:constantK}
    K \epsguarantee =3\LiprewPi \sum_{\h=0}^{\Horizon-1} \LipDyn_{\h}\epsterminal + \LiprewPi \LipDyn_{\Horizon } \LipDyn_{\delta \h} (\epscollect+\noisebound).
\end{align}
\end{lemma}
\begin{proof} 
We will derive a lower bound for $\Jobj[]{\state(\ki), \dyntrueVec}{\retPolicy}$ which is given by,
\begin{align} \label{eq:lb_rew_ret_policy}
    \Jobj[]{\state(\ki), \dyntrueVec}{[\hat{\pi},\pessiPolicy]} = \Jobj[]{\state(\ki), \dyntrueVec}{\hat{\pi}} + \E_{\state^{\star}_0 \sim \hat{\pi}} \Jobj[]{\state^{\star}_0, \dyntrueVec}{\pessiPolicy} \labelrel{\geq}{step:rewgeq0} \E_{\state^{\star}_0 \sim \hat{\pi}} \Jobj[]{\state^{\star}_0, \dyntrueVec}{\pessiPolicy},
\end{align}
where $\state^{\star}_{0}$ is the state (random variable) of the true dynamics controlling from $\state(\ki)$ with the policy $\hat{\pi}$ under the noise distribution which is denoted by $\E_{\state^{\star}_0 \sim \hat{\pi}}$. Step~\eqref{step:rewgeq0} follows from $r(\cdot,\cdot)\geq0$. 

Next, we will bound the difference in cumulative rewards between mean dynamics $\bm{\mu}_\n$ and $\dyntrueVec$, using the policy $\pessiPolicy$, when starting, respectively, from $\statePessi$ and $\state^\star_0$ and realizing the same noise sequence, resulting in trajectories $\state_\h$ and $\state^\star_\h$ respectively.
\begin{align*}
\!\!\!\! \Jobj[]{\statePessi, \bm{\mu}_\n}{\pessiPolicy} &- \E_{\state^{\star}_0 \sim \hat{\pi}}\Jobj[]{\state^{\star}_0, \dyntrueVec}{\pessiPolicy} \\
&= \E_{\state^{\star}_0 \sim \hat{\pi}} \E_{\eta} \left(\sum_{\h=0}^{\Horizon-1} r(x_\h, \pessiPolicy[\h](x_\h)) - r(x^\star_\h, \pessiPolicy[\h](x^\star_\h)) \right)\\
&
\leq \E_{\state^{\star}_0 \sim \hat{\pi}} \E_{\eta} \LiprewPi \sum_{\h=0}^{\Horizon-1} \|\state_\h - \state^\star_\h \| \\
&\leq \LiprewPi \sum_{\h=0}^{\Horizon-1} \left( \LipDyn^{\h} \LipDyn_{\delta \h} (\epscollect + \noisebound) + \E_{\noise} \sum_{i=0}^{\h-1} \LipDyn^{i}  \cwidth[\n-1](\state_{i}, \pessiPolicy[i](\state_{i}))  \right)\\
&= \LiprewPi \LipDyn_{\Horizon} \LipDyn_{\delta \h} (\epscollect + \noisebound) + \E_{\noise} \LiprewPi \sum_{\h=0}^{\Horizon-1} \sum_{i=0}^{\h-1} \LipDyn^{i}  \cwidth[\n-1](\state_{i}, \pessiPolicy[i](\state_{i}))  \numberthis \label{eq:mean_true_bound}
\end{align*}
In the above, the first inequality follows from the Lipschitz continuity of rewards and the policies and for the second inequality, analogous to \cref{lem:deviation_with_uncertainity}, we used that 
\begin{align*}
    \|x_\h - x^\star_\h \| &\leq \LipDyn \|\state^\star_{\h-1} - \state_{\h-1}\| + \cwidth[\n-1](\state_{\h-1}, \pessiPolicy(\state_{\h-1}))\\
    &\leq \LipDyn \|\state^\star_{\h-2} - \state_{\h-2}\| + \LipDyn\cwidth[\n-1](\state_{\h-2}, \pessiPolicy(\state_{\h-2})) +  \cwidth[\n-1](\state_{\h-1}, \pessiPolicy(\state_{\h-1})) \\
    &\leq L^{\h} L_{\delta \h} (\epscollect + \noisebound) + \sum_{i=0}^{\h-1} \LipDyn^{i}  \cwidth[\n-1](\state_{i}, \pessiPolicy(\state_{i}))
\end{align*}
where the last inequality uses that $\|x_0 - x^\star_0\|\leq \LipDyn_{\delta \h} (\epscollect + \noisebound)$ by \cref{assump:dynamics_tol_terminal} of knowing dynamics upto $\epsdiff$ tolerance in the safe set. Since the upper bound is constant (does not depend on $\state^{\star}_0$), $\E_{\state^{\star}_0 \sim \hat{\pi}}$ disappears.



Finally, on substituting \cref{eq:mean_true_bound} in \cref{eq:lb_rew_ret_policy}, we get, 
\begin{align*}
    \Jobj[]{\state(\ki), \dyntrueVec}{[\hat{\pi},\pessiPolicy]}
    &\geq \Jobj[]{\statePessi, \bm{\mu}_\n}{\pessiPolicy} - \E_{\noise} \LiprewPi \sum_{\h=0}^{\Horizon-1} \sum_{i=0}^{\h-1} \LipDyn^{i}  \cwidth[\n-1](\state_{i}, \pessiPolicy(\state_{i})) - \LiprewPi \LipDyn_{\Horizon } \LipDyn_{\delta \h} (\epscollect + \noisebound) \\
    &\labelrel\geq{step:substermination} \Jobj[]{\state^{\mathrm{o}}, \dynVec^o }{\optiPolicy} - 3 \LiprewPi \sum_{\h=0}^{\Horizon-1} \LipDyn_{\h} \epsterminal - \LiprewPi \LipDyn_{\Horizon } \LipDyn_{\delta \h} (\epscollect + \noisebound) \\
    &\labelrel\geq{step:truesys} \Jobj[]{\state^{\star}, \dyntrueVec }{\pi^\star} - 3 \LiprewPi \sum_{\h=0}^{\Horizon-1} \LipDyn_{\h} \epsterminal - \LiprewPi \LipDyn_{\Horizon } \LipDyn_{\delta \h} (\epscollect + \noisebound)
\end{align*}

Step~\eqref{step:substermination} follows from the termination criterion~\eqref{eq:termination_criteria}.    Step~\eqref{step:truesys} follows since the optimistic policy set is always a subset of the true policy set, see~\cref{prop:over_under_approx}. The last inequality implies satisfying~\cref{obj:reward_maximization} with K as per \eqref{eq:constantK}.
\end{proof}

\begin{lemma}\label{lem:true_less_uncertainity_pessi_policy}
Let~\cref{assump:q_RKHS,assump:lipschitz,assump:safeSet,assump:sublinear} hold. 
Let $\nfin\leq\n^\star$ be as defined in \cref{thm:maximum_dynamics_exploration}.
Consider any $\state^\star_0 \in \safeInit{\nfin}$, $\pessiPolicy \in \pessiSet[,\epspessi]{\nfin}(\state^\star_0, \Horizon)$ and the resulting trajectory $\state^\star_{h+1} = \dyntrueVec(\state^\star, \pessiPolicy(\state^\star))$, $\forall \h \in \Intrange{0}{\Horizon-1}$. 
Then the following holds:
\begin{align} \label{eq:short_pessi_uncertainity_bound}
    \cwidth[\nfin](\state, \pessiPolicy[\h](\state)) \leq \LipWidthPi\|\state - \state^\star_{\h}\| + \epsterminal, \h \in \Intrange{0}{\Horizon}, \state \in \X.  
\end{align}
\end{lemma}
\begin{proof}
Given the sampling criteria in \cref{alg:rew_maxim} and using \cref{lem:sample_complexity}, we obtain $\exists \nfin \leq \n, \state_s \in \safeInit{\nfin}:$ 
 \begin{align} \label{eq:pessi_uncertainity_threshold}
    \cwidth[\nfin](\state_\h, \pessiPolicy(\state_\h)) < \epsterminal, \forall \h \in \Intrange{0}{\totHorizon-1}, \dynVec \in \dynSet[\n], \pessiPolicy \in \pessiSet[]{\nfin}(\state_s, \totHorizon),
\end{align}
where $\state_{\h+1} = \dynVec(\state_{\h}, \pessiPolicy[\h](\state_{\h}))$ with $\state_0 = \state_s$. Suppose (cf.~\cref{eq:JK_espilon_formula})
\begin{align}
    \epspessi > \max \{ 1, \LipPi \} 4 \sum_{j=0}^{ \totHorizon-1 }(\LipWidthPi+\LipDyn)^j (\epsterminal+ \LipWidthPi \LipDyn_{\Horizon-j-1}\noisebound) \label{eq:eps_dash_formula}
\end{align}
which ensures  $\truePolicySet[,\frac{\epspessi}{2}]{c,\nfin}(\state_s;\totHorizon) \subseteq \pessiSet[]{\nfin}(\state_s;\totHorizon)$ 
by using \cref{lem:uniform_exploration}. Along with \eqref{eq:pessi_uncertainity_threshold}, this implies
\begin{align}\label{eq:long_true_traj}
    \cwidth[\nfin](\state^{\star'}_\h, \pi^\star_\h(\state^{\star'}_\h)) < \epsterminal, \forall \h \in \Intrange{0}{\totHorizon-1}, \pi^\star \in \truePolicySet[,\frac{\epspessi}{2}]{c,\nfin}(\state_s, \totHorizon),
\end{align}
with $\state^{\star'}_{\h+1} = \dynVec(\state^{\star'}_{\h}, \pi^\star_\h(\state^{\star'}_{\h}))$ with $\state^{\star'}_0 = \state_s$.
Using the invariance property of \cref{assump:safeSet}, for all $\delta\h \in \Intrange{0}{\Delta \Horizon}, \pi^\star \in \truePolicySet[,\frac{\epspessi}{2}]{c,\nfin}(\state_s, \Horizon + \delta \h),~\exists \pi_f \in \Pi_{\Delta \Horizon - \delta \h}: [\pi^\star,\pi_f] \in \truePolicySet[,\frac{\epspessi}{2}]{c,\nfin}(\state_s, \totHorizon)$ which from~\eqref{eq:long_true_traj} implies,
\begin{align}\label{eq:long_true_traj_dh}
    \cwidth[\nfin](\state^{\star'}_\h, \pi^\star_\h(\state^{\star'}_\h)) < \epsterminal, \forall \h \in \Intrange{0}{\totHorizon-1}, \pi^\star \in \truePolicySet[,\frac{\epspessi}{2}]{c,\nfin}(\state_s, \Horizon + \delta \h),
\end{align}
Since $\epspessi$ satisfies \eqref{eq:epsilon_lower_bound_convergence}, it also holds that $\epspessi\geq \max\{1,\LipPi\} 2 \max \{\LipDyn^\Horizon, 1\} \LipDyn_{\Delta \Horizon}\noisebound$. Then by using \cref{lem:horizon_diff} we get,
\begin{align} \label{eq:true_traj_bound}
    \cwidth[\nfin](\state^\star_\h, \pi^\star_\h(\state^\star_\h)) < \epsterminal, \forall \h \in \Intrange{0}{\Horizon-1}, 
     \state^\star_0 \in \safeInit{\nfin}, 
    \pi^\star \in \truePolicySet[,\epspessi]{c,\nfin}(\state^\star_0, \Horizon),
\end{align}
with $\state^{\star}_{\h+1} = \dynVec(\state^{\star}_{\h}, \pi^\star_\h(\state^{\star}_{\h}))$. In the above, optimistic trajectory $\state^\star_\h$ satisfies $\cwidth[\nfin](\state^\star_\h, \pi^\star_\h(\state^\star_\h)) < \epsterminal$,
since $\state^{\star'}_\h$ in \eqref{eq:long_true_traj} is a noiseless trajectory and thus with $\hat{\pi}$ we can control it from $\state_s$ to any location $\state^\star_0 \in \safeInit{\n}$ exactly.

Using \cref{prop:over_under_approx}, we know that $\pessiSet[,\epspessi]{\nfin}(\safeInit{\nfin}, \Horizon) \subseteq \truePolicySet[,\epspessi]{c,\nfin}(\safeInit{\nfin}, \Horizon)$ and thus the noise-free trajectory from true dynamics under pessimistic policy $\pessiPolicy$ satisfies,
\begin{align}
    \cwidth[\nfin](\state^\star_\h, \pessiPolicy[\h](\state^\star_\h)) < \epsterminal, \forall \h \in \Intrange{0}{\Horizon-1}, 
    \state^\star_0 \in \safeInit{\nfin}, 
    \pessiPolicy \in \pessiSet[,\epspessi]{\nfin}(\state^\star_0, \Horizon). 
\end{align}
Finally, using the Lipschitz continuity of $\cwidth$ and $\pi$ with any point $\state \in \X$ yields \eqref{eq:short_pessi_uncertainity_bound}, which concludes the proof.
\end{proof}

In the following, we show that the sampling rule~\eqref{eq:sampling_rule_timex} ensures exploration of the pessimistic policy set $\pessiSet[,\epspessi]{\n}(\safeInit{\n}; \Horizon)$.

\begin{figure}
    \centering
    \scalebox{0.8}{\input{figure/dyn_init_set_appending_pessi}}
    \caption{
 Illustration in state space of the trajectories used to relate the pessimistic policy sets under two different horizons in \cref{lem:horizon_diff}. 
    The policy $\pessiPolicy \in \pessiSet[,\epspessi]{\nfin}(\safeInit{\n}; \Horizon)$ drives all dynamics $\dynVec$  from a state $\state'$ to $\statePessi_{\Horizon}\in \safeInit{\n}$ (marked with shaded lines), while ensuring that all intermediate states $\statePessi_\h \in \X \ominus \ball{\epsguarantee'}$, for all $\h \in \Intrange{0}{\Horizon}$. 
    We show that the concatenated policy 
    $[\hat{\pi}, \pessiPolicy ] \in \pessiSet[]{\nfin}(\state_s;\Horizon + \delta \h)$ 
    control any dynamics $\dynVec \in \dynSet[\nfin]$ starting at $\state_s$ to $\state'$, where all the dynamics land in a ball around $\state'$ (marked with a circle) due to unknown dynamics and noise, and then follows $\pessiPolicy$. The resulting trajectories $\statePessi[,l]_\h$ for all dynamics (region indicated by semi-transparent blue color) satisfy the constraints $\statePessi[,l]_\h \in \X$ for all horizons. 
    A similar argument applies to the input constraints, though not shown in the state space figure above.}
    \label{fig:dyn_init_set_appending_pessi}
    \vspace{-1em}
\end{figure}

\begin{lemma} \label{lem:horizon_diff_pessi}
Let \cref{assump:q_RKHS,assump:safeSet,assump:lipschitz,assump:sublinear} hold. $\exists \nfin \leq \n^\star: \forall  \state' \in  \safeInit{\n},~ \pessiPolicy \in \pessiSet[,\epspessi]{\nfin}(\state'; \Horizon), \exists \hat{\pi} \in\Pi_{\delta \h}: 
[\hat{\pi}, \pessiPolicy ] \in \pessiSet[]{\nfin}(\state_s;\Horizon + \delta \h)$.
\end{lemma}
\begin{proof}
Given $\state_s$ as an initial state, any target state  $\state'\in\safeInit{\n}$ and any dynamics $\dynVec\in\dynSet[\n]$ (optimistic), consider
applying policy $\hat{\pi}$ (\cref{assump:safeSet}) that control this dynamics $\dynVec$ from $\state_s \to \state' \eqqcolon \statePessi_0$. 
However, due to noise and unknown dynamics in the safe set, 
the worst deviation satisfies $\|\statePessi_0 - \statePessi[,l]_0\| \leq 
\sum_{h=0}^{ \delta h - 1}\LipDyn^h(\epsdiff + \noisebound)$ where $\statePessi[,l]_0$ results from any (worst-case) dynamics $\dynVec\in\dynSet[\n]$ 
under the policy $\hat{\pi}$ (cf.~\cref{lem:safety_with_return_policy}). 
After that, let $\statePessi[,l]_h$ and $\statePessi_h$ denote the state sequences generated by applying the policy $\pessiPolicy$ and some noise sequences $\eta_h$ to the dynamics $\dynVec$ with initial conditions $\statePessi[,l]_0$ and $\statePessi_0$, respectively. See \cref{fig:dyn_init_set_appending_pessi} for an illustration.

We will bound $\|\statePessi[,l]_{h+1}- \statePessi_{h+1}\|, \forall \h \in \Intrange{0}{\Horizon-1}$ to determine $\epspessi$ such that the policy $[\hat{\pi},\pessiPolicy]$ keeps all dynamics safe, i.e. $[\hat{\pi},\pessiPolicy]\in\pessiSet[]{\nfin}(\state_s;\Horizon + \delta \h)$.  
For this, define noiseless true trajectory sequence $\state^{\star}_{h+1} = \dyntrueVec(\state^\star_\h, \pessiPolicy[\h](\state^\star_\h))$ 
with initial conditions $\state_0^{\star} = \statePessi_0$. Using triangle inequality, we get,
\begin{align}\label{eq:split_pessi_traj}
    \|\statePessi[,l]_{h+1} - \statePessi_{h+1}\| \leq \|\statePessi[,l]_{h+1} -x_{h+1}^{\star}\| + \|x_{h+1}^\star - \statePessi_{h+1}\|.
\end{align}
Analogous to~\cref{prop:suffi_info_epscollect}, it holds that
\begin{align*}
     \|x_{h+1}^{\star} - \statePessi_{h+1}\| &= \| \dyntrueVec(\state^{\star}_{h}, \pessiPolicy[\h](\state^{\star}_{h})) - \dyntrueVec(\statePessi_{\h}, \pessiPolicy[\h](\statePessi_{\h})) \| + \| \dyntrueVec(\statePessi_{\h}, \pessiPolicy[\h](\statePessi_{\h}))- \dynVec(\statePessi_{\h}, \pessiPolicy[\h](\statePessi_{\h}))- \noisebound\| \\
     &\stackrel{\eqref{eq:short_pessi_uncertainity_bound}}{\leq} \LipDyn\|\state^{\star}_{\h}-\statePessi_{\h}\| + \LipWidthPi\|\statePessi_{\h} - \state^\star_{\h}\| + \epsterminal + \noisebound  \\
     &\leq (\LipDyn + \LipWidthPi)^h\|\state^{\star}_0-\statePessi_0\| + \LipDyn_{w,\h} (\epsterminal + \noisebound), \numberthis \label{eq:diff_true_pessi}
\end{align*}
where the first inequality follows using \cref{lem:true_less_uncertainity_pessi_policy}. Analogously, we have
\begin{align}
    \|x_{h+1}^{\star}- \statePessi[,l]_{h+1}\| & \leq (\LipDyn + \LipWidthPi)^h\|\state^{\star}_0-\statePessi[,l]_0 \| + \LipDyn_{w,\h} (\epsterminal + \noisebound)  \nonumber \\
    &\leq  (\LipDyn + \LipWidthPi)^h \sum_{i=0}^{ \delta \h - 1}\LipDyn^i (\epsdiff + \noisebound)  + \LipDyn_{w,\h} (\epsterminal + \noisebound) . \label{eq:diff_true_pessi_l}
\end{align}
Together, \cref{eq:diff_true_pessi,eq:diff_true_pessi_l} on substituting in \cref{eq:split_pessi_traj} imply
\begin{align*}
 \forall \h \in \Intrange{0}{\Horizon-1},~\|\statePessi_{h+1} -\statePessi[,l]_{h+1}\|\leq (\LipDyn + \LipWidthPi)^\Horizon \sum_{i=0}^{ \delta \h - 1}\LipDyn^i (\epsdiff + \noisebound) + 2 \LipDyn_{w,\Horizon} (\epsterminal + \noisebound)  \leq \epspessi.
\end{align*}
and for the input constraints, we have
\begin{align*}
    \forall \h \in \Intrange{0}{\Horizon-1}, \|\pi_h(\statePessi_\h) - \pi_h(\statePessi[,l]_\h) \|\leq \LipPi \| \statePessi_\h - \statePessi[,l]_\h \| \leq \epspessi.
\end{align*}
with 
\begin{align}
     \label{eq:pessi_eps_dash_shrink}
     \epspessi \geq \max\{1, L_\pi\} \left(\max\{ (\LipDyn + \LipWidthPi)^\Horizon,1 \} \LipDyn_{\Delta \Horizon} (\epsdiff + \noisebound) + 2 \LipDyn_{w,\Horizon} (\epsterminal + \noisebound) \right).
\end{align}
\end{proof}

\vspace{-2em}
\begin{lemma} \label{lem:opti_in_pessi}
Let \cref{assump:q_RKHS,assump:safeSet,assump:lipschitz,assump:sublinear} hold. Suppose \cref{alg:rew_maxim} terminates in $\nfin \leq \n^\star$ iterations with $\n^\star$ as defined in \cref{thm:maximum_dynamics_exploration}. Then it holds that $\optiSet[]{\nfin}(\safeInit{\nfin}; \Horizon) \subseteq \pessiSet[,\epspessi]{\nfin}(\safeInit{\nfin}; \Horizon)$.
\end{lemma}
\begin{proof}
Given the sampling rule in \cref{alg:rew_maxim} and using \cref{lem:sample_complexity}, we know $\exists \nfin \leq \n:$ 
$\cwidth[\nfin](\state_\h, \pessiPolicy(\state_\h)) < \epsterminal, \forall \h \in \Intrange{0}{\totHorizon-1}, \dynVec \in \dynSet[\nfin], \pessiPolicy \in \pessiSet[]{\nfin}(\state_s, \totHorizon)$. 
Note that for all $\delta\h \in \Intrange{0}{\Delta \Horizon}, \pi \in \pessiSet[]{\nfin}(\state_s;\Horizon + \delta \h),~\exists \pi_f \in \Pi_{\Delta \Horizon - \delta \h}: [\pi,\pi_f] \in \pessiSet[]{\nfin}(\state_s;\totHorizon)$ using the invariance property of \cref{assump:safeSet}. Hence, 
$\cwidth[\nfin](\state_\h, \pessiPolicy(\state_\h)) < \epsterminal, \forall \h \in \Intrange{0}{\Horizon + \delta \h-1}, \dynVec \in \dynSet[\nfin], \pessiPolicy \in \pessiSet[]{\nfin}(\state_s, \Horizon + \delta \h)$.
Using \cref{lem:horizon_diff_pessi}, this implies
\begin{align}
    \cwidth[\nfin](\state_\h, \pi(\state_\h)) < \epsterminal, \forall \h \in \Intrange{0}{\Horizon-1}, \dynVec \in \dynSet[\n], \pi \in \pessiSet[,\epspessi]{\nfin}(\safeInit{\n}, \Horizon). 
\end{align}
Using the above statement, the proof follows similar to \cref{lem:uniform_exploration}, where we use
\begin{align}
\label{eq:esp_diff_formula}
\epsguarantee -\epspessi > \max\{1,\LipPi\} 2 \sum_{j=0}^{ \Horizon-1}(\LipWidthPi+\LipDyn)^j (\epsterminal+ \LipWidthPi \LipDyn_{\Horizon-j-1}\noisebound), 
 \end{align}    
see also \cref{eq:JK_espilon_formula} in \cref{lem:contra_eps}.
\end{proof}
\vspace{-1em}


\begin{lemma}[Guaranteed termination]\label{lem:guaranteed_termination}
~\cref{assump:q_RKHS,assump:safeSet,assump:lipschitz,assump:sublinear} hold. Consider $\n^\star$ as defined in \cref{thm:maximum_dynamics_exploration}. \cref{alg:rew_maxim} terminates in $\nfin \leq \n^\star$ iterations.
\end{lemma}
\noindent \textbf{Proof} Analogous to~\cref{prop:suffi_info_epscollect}, \cref{alg:rew_maxim} ensures that with every iteration the collected data has at least one informative measurement, i.e., $\cwidth[\n-1,i](\state_\h, \coninput[\h]) \geq \epscollect$. 
Following \cref{lem:sample_complexity}, this implies that there exists $\nfin \leq \n^\star$ such that $\cwidth[\nfin, i](\state_\h, \coninput[\h]) < \epsterminal, \forall i \in \Intrange{1}{\statedim}, \h \in \Intrange{0}{\totHorizon-1}, \dynVec \in \dynSet[\n], \pi \in \pessiSet[]{\nfin}(\state(k), \totHorizon)$ where $\state_{\h+1} = \dynVec(\state_\h,\coninput_\h), \coninput_\h = \pi_\h(\state_\h)$ and $\state_0=\state(k)$.

In the following, we show that this ensures satisfaction of the termination criteria~\eqref{eq:termination_criteria}. 
From~\cref{lem:opti_in_pessi}, we know that $\optiPolicy \in \optiSet[]{\nfin}(\safeInit{\n};\Horizon)  \subseteq \pessiSet[,\epspessi]{\nfin}(\safeInit{\n};\Horizon)$. 
Since $\optiPolicy$ and $\stateOpti$ are feasible candidate for pessimistic problem~\eqref{eq:pessi_obj} we get,
\begin{align}
\Jobj[\mathrm{p}]{\statePessi, \bm{\mu}_{\n} }{\pessiPolicy} &\geq \Jobj[\mathrm{p}]{\state^{\mathrm{o}}, \bm{\mu}_{\n} }{\optiPolicy} \\
&\labelrel{=}{step:pessi_obj_def}\Jobj[]{\state^{\mathrm{o}}, \bm{\mu}_{\n} }{\optiPolicy} - \E_{\noise} \left[\LiprewPi \sum_{\h=0}^{\Horizon-1} \sum_{i=0}^{\h-1} \LipDyn^{i}  \cwidth[\n-1](\state_{i}, \optiPolicy[i](\state_{i}))\right] \\
& \labelrel{\geq}{step:noise_to_no_noise} \Jobj[]{\state^{\mathrm{o}}, \bm{\mu}_{\n} }{\optiPolicy} - \E_{\noise} \left[\LiprewPi \sum_{\h=0}^{\Horizon-1} \sum_{i=0}^{\h-1} \LipDyn^{i} ( \LipWidthPi \| \state^\star_{i} - \state_{i} \| +\epsterminal )\right]\\ 
&\labelrel{\geq}{step:noise_diff_bound} \Jobj[]{\state^{\mathrm{o}}, \bm{\mu}_{\n} }{\optiPolicy} - \E_{\noise} \left[\LiprewPi \sum_{\h=0}^{\Horizon-1} \sum_{i=0}^{\h-1} \LipDyn^{i} ( \LipWidthPi \LipDyn_{w,i} (\epsterminal + \noisebound) + \epsterminal )\right]\\ 
&\geq \Jobj[]{\state^{\mathrm{o}}, \dynVec^\mathrm{o}}{\optiPolicy} - \LiprewPi \sum_{\h=0}^{\Horizon-1} \sum_{i=0}^{\h-1} \LipDyn^{i} ( \LipWidthPi \LipDyn_{w,i} (\epsterminal + \noisebound) + \epsterminal ) - 2 \LiprewPi \sum_{\h=0}^{\Horizon-1} \LipDyn_{w,\h} (\epsterminal + \noisebound) \label{eq:termination_proof_last_Step}
\end{align}
Here \eqref{step:pessi_obj_def} follows from the definition of pessimistic objective \eqref{eq:pessi_obj}. $x_i$, and $x^\star_i$ denote the trajectory under with mean and true dynamics while $x_i$ is perturbed noise whereas $x^\star_h$ is noise-free and $\cwidth[\nfin](\state^\star_i, \optiPolicy[i](\state^\star_i)) < \epsterminal$. Step \eqref{step:noise_diff_bound} follows from the following. Consider,
\begin{align}
    \|x^\star_{\h+1} - \state_{\h+1}\| &\leq \| \dyntrueVec(\state^{\star}_{\h}, \optiPolicy[\h](\state^{\star}_{\h})) - \bm{\mu}_{\nfin}(\state_{\h}, \optiPolicy[\h](\state_{\h})) - \eta \| \\
    &\leq \noisebound + \| \dyntrueVec(\state^{\star}_{\h}, \optiPolicy[\h](\state^{\star}_{\h})) - \dyntrueVec(\state_{\h}, \optiPolicy[\h](\state_{\h}))+ \dyntrueVec(\state_{\h}, \optiPolicy[\h](\state_{\h})) - \bm{\mu}_{\nfin}(\state_{\h}, \optiPolicy[\h](\state_{\h})) \| \nonumber \\
    &\leq (\LipDyn + \LipWidthPi)\|x^\star_{\h} - \state_{\h}\| + \epsterminal + \noisebound \label{eq:process_traj_diff}\\ 
    &\leq \LipDyn_{w,\h+1} (\epsterminal + \noisebound) \label{eq:traj_diff_noise_lip_width}
\end{align}
The step in \eqref{eq:process_traj_diff} follows using $\cwidth[\nfin](\stateOpti_\h, \pi^o_\h(\stateOpti_\h)) < \LipWidthPi\|x^\star_{\h} - \stateOpti_{\h}\| + \cwidth[\nfin](x^\star_{\h}, \pi^o_\h(x^\star_{\h}))$ and then $\cwidth[\nfin](x^\star_{\h}, \pi^o_\h(x^\star_{\h})) < \epsterminal$.

Finally, the last inequality in \eqref{eq:termination_proof_last_Step} follows from the following,
\begin{align}
|\Jobj[]{\state^{\mathrm{o}}, \bm{\mu}_{\n} }{\optiPolicy} - \Jobj[]{\state^{\mathrm{o}}, \dynVec^o }{\optiPolicy} | &\leq \sum_{\h=0}^{\Horizon-1} |r(x_\h, \optiPolicy(x_\h)) - r(\stateOpti_\h, \optiPolicy(\stateOpti_\h))| \\
&\leq \LiprewPi \sum_{\h=0}^{\Horizon-1} \|x_\h - \stateOpti_\h\| \\
&\leq \LiprewPi \sum_{\h=0}^{\Horizon-1} \|x_\h - x^\star_\h\| + \|x^\star_\h - \stateOpti_\h\| \label{eq:rew_diff_conv}
\end{align}
where $\stateOpti_\h$ denotes the trajectory under optimistic dynamics,  perturbed by the same noise sequence as $\state_\h$. 
The trajectory difference bound from \eqref{eq:traj_diff_noise_lip_width} also holds for the optimistic dynamics, since $\dynVec^{\mathrm{o}} \in \dynSet[\nfin]$. 
Hence, on substituting \eqref{eq:traj_diff_noise_lip_width} in \eqref{eq:rew_diff_conv}, we get,
\begin{align*}
    |\Jobj[]{\state^{\mathrm{o}}, \bm{\mu}_{\n} }{\optiPolicy} - \Jobj[]{\state^{\mathrm{o}}, \dynVec^o }{\optiPolicy} | \leq 2 \LiprewPi \sum_{\h=0}^{\Horizon-1} \LipDyn_{w,\h} (\epsterminal + \noisebound). \qedhere
\end{align*}

\restaterewmaximization*
\begin{proof}  
\emph{Feasibility:}
First, we prove the feasibility of~\cref{alg:rew_maxim}. The control invariance property of \cref{assump:safeSet} ensures that $[\pi_\mathrm{f}, \hdots, \pi_\mathrm{f}] \in \pessiSet[,\epspessi]{0}(\safeInit{0}, \Horizon)$. Since the $\pessiSet[,\epspessi]{0}(\safeInit{0}, \Horizon)$ is not empty, this guarantees feasibility of Problem~\eqref{eq:pessi_obj} in ~\cref{lin:solve_pessi_prob}. 
The feasibility of the optimistic Problem~\eqref{eq:opti_obj} is similarly ensured by the control invariance property with tightening of the safe set with $\epsilon$.  
%

Feasibility of Problem~\eqref{eq:slack_there_exists_goal} in \cref{lin:solve_relax_prob} at any $\n\geq1$
is ensured using standard MPC arguments for recursive feasibility~\citep{rawlings2017model}. Without loss of generality, let $\pi^{\mathrm{old}}$ be a feasible policy at any $\n$ and $\h' \in \Intrange{0}{\totHorizon-1}$ be the step where we obtain an informative location ($|\mathcal{D}_c|\geq1$) under $\pi^{\mathrm{old}}$. We build a feasible candidate sequence at iteration $\n+1$ with $\pi \in \Pi_{\totHorizon}$ and $\state_\h \in \X$ be state under any $\dynVec \in \dynSet[\n+1]$, 
by shifting the previous feasible solution $\pi_{\h} = \pi^{\mathrm{old}}_{\h' + \h}, 
\forall \h \in 
\Intrange{0}{\totHorizon - \h' - 1}$
and appending $\pi_{\h} = \pi_f, \forall \h \in \Intrange{\totHorizon - \h'}{\totHorizon -1}$,  from \cref{assump:safeSet}. 
Feasibility follows from $\state_\h \in \X , \forall \h \in \Intrange{0}{\totHorizon-\h'-1}$ where $\state_\h$ is propagated under any $\dynVec \in \dynSet[\n+1] \subseteq \dynSet[\n]$ (nestedness \cref{eq:dyn_set})
with the previous feasible policy; and
$\state_\h \in \safeInit{\n} \subseteq \safeInit{\n+1} \subseteq \X, \forall \h \in \Intrange{\totHorizon-\h'}{\totHorizon-1}$ respectively due
to “invariance” and “monotonicity” of the safe set $\safeInit{\n}$ (\cref{assump:safeSet}). 
Moreover, note that the
constraint $\cwidth[\n](\state_\h,\coninput_\h) \geq \epsterminal - \nu$ in Problem~\eqref{eq:slack_there_exists_goal} is always feasible by
choosing $\nu$ sufficiently large. 

At $\ki=0$, $\pessiSet[]{0}(\state(\ki);\totHorizon)$ is non-empty due to the control invariance property and $\state(0) \in \safeInit{0}$. From then on, for any further $\n\geq1$, the solution of $\n-1$ is still feasible since $\safeInit{\n} \subseteq \safeInit{\n+1}$ by \cref{assump:safeSet}, $\dynSet[\n] \supseteq \dynSet[\n+1]$ by construction~\eqref{eq:dyn_set} and the control invariance property ensures that the system can stay in the safe set.
In \cref{lin:return}, Problem~\eqref{eq:slack_there_exists_goal} is feasible with $\nu=0$. Otherwise it would hold that $\cwidth[\nfin](\state_\h,\coninput_\h) < \epsterminal,~~\forall \h \in \Intrange{0}{\totHorizon-1}, \dynVec \in \dynSet[\nfin], \pessiPolicy \in \pessiSet[]{\totHorizon}(\state_s)$, with $\state_{\h+1} = \dynVec(\state_\h,\coninput_\h), \coninput_\h = \pi_\h(\state_\h)$, which implies $ \optiSet[]{\n}(\safeInit{\n}; \Horizon) \subseteq \pessiSet[,\epspessi]{\n}(\safeInit{\n}; \Horizon)$ by \cref{lem:horizon_diff_pessi,lem:opti_in_pessi}. This would already ensure that the termination condition of \cref{alg:rew_maxim} is satisfied and thus the algorithm will reach~\cref{lin:return} only if it is feasible. 
Hence, all optimization problems are feasible in~\cref{alg:rew_maxim} for all $\n\geq 0$.

\emph{Safety:}
Analogous to \cref{thm:maximum_dynamics_exploration}, safety is ensured since \cref{alg:rew_maxim} applies the policy $\pessiPolicy \in \pessiSet[]{\n}(\state(\ki);\totHorizon), \forall \n \geq 0$ which, by definition, ensures constraint satisfaction $\forall \dynVec \in \dynSet[\n]$. By \cref{assump:q_RKHS} and \cref{lem:beta}, since the unknown system $\dyntrueVec \in \dynSet[\n]$, thereby guaranteeing constraint satisfaction for the unknown system~\eqref{eq:system_dyn} as well.

\emph{Finite termination:}
\cref{lem:guaranteed_termination} ensures that the defined termination criterion is satisfied in $\nfin\leq \n^\star$ iterations.

\emph{Close to optimal performance:} 
First, we show that the agent satisfies constraints after following the returned policy $[\hat{\pi}, \pessiPolicy]$ in \cref{lem:safety_with_return_policy}. Then 
\cref{lem:optimality_with_return_policy} ensures that once the termination criterion is satisfied  (early termination not necessarily uniformly reducing uncertainty), then the returned policy satisfies~\cref {obj:reward_maximization}. 
\end{proof}



\vspace{-2em}
\section{Lower bound on sample complexity}
\label{apxsex:lowerbound}

In the following lemma, we adapt the kernel-based lower bounds on sample complexity for the Gaussian process bandit optimization setting derived in \citet{scarlett2017lower} to have bounded noise.
\begin{lemma} \label{coro:scarlett_lb_prob} Fix $\epsilon \in (0, \frac{1}{2})$, $\delta \in (0, \frac{1}{8})$, $B>0$ and $\nfin \in \R$. 
Consider the  Gaussian process bandit optimization setting for any \mbox{$f \in \RKHS$} with \mbox{$\|f\|_{\RKHS} \leq B $} and the measurement noise to be a Gaussian distribution, i.e. $\noise \sim \N(0, \sigma^2)$. Then with probability at least $1-\delta$ the noise is bounded, i.e.,  
\begin{align}
\probability{|\noise_\n| \leq \noisebound \coloneqq \sigma \left( 2\log\left(\frac{ \nfin}{\delta}\right) \right)^{1/2}, \forall \n \leq \nfin. } \geq 1-\delta \label{eq:bounded_noise_proablity}   
\end{align}
Given that  \eqref{eq:bounded_noise_proablity} holds, suppose that there exists an algorithm that with probability $1-\delta$ achieves a simple regret $f(x^\star)-f(x_\nfin)\leq \epsilon'$ after rounds $\nfin$. Let $ \epsilon = \epsilon'/r > 0$ where $r = (1-8\delta)/(1-2\delta)$ and provided that $\epsilon/B$ is sufficiently small, the following holds:
\begin{itemize}
    \item For the squared exponential kernel, it is necessary that
    \begin{align}
        \nfin =  \Omega\left(\frac{\sigma^2}{\epsilon^2} \left ( \log \frac{B}{\epsilon}\right)^{d/2} \right). \label{eq:sqexp_lb}
    \end{align}
    \item For the Mat\'ern kernel, it is necessary that
    \begin{align}
        \nfin = \Omega\left(\frac{\sigma^2}{\epsilon^2} \left (\frac{B}{\epsilon}\right)^{d/\nu} \right). \label{eq:matern_lb}
    \end{align}
\end{itemize} 
\end{lemma}
\begin{proof} 
The overall idea is as follows: Consider an algorithm designed to operate under truncated Gaussian noise and apply it to the Gaussian process bandit optimization setting with Gaussian noise.  Any instance where the noise exceeds the truncation bound is treated as a failure event, and we remove them in probability.  

For this, we first show that Gaussian noise $\noisebound_n$ is bounded by $\noisebound$ with probability at least $1-\delta$. 
Note that for a normally distributed random variable 1-D $X \sim \N(\mu, \sigma^2)$, we have \citep{beta_srinivas},
\begin{align}
    \probability{|X-\mu|>\noisebound} <  e^{-\frac{\noisebound^2}{2\sigma^2}}
\end{align}
for any $\noisebound>0$. Now, consider $\nfin$ i.i.d. random variables $\noise_1, \noise_2, \hdots \noise_{\nfin} \sim \N(0, \sigma^2)$, then it holds that
\begin{align}
    \probability{\forall n \leq \nfin, |\noise_n|<\noisebound} > (1 -  e^{-\frac{\noisebound^2}{2\sigma^2}})^\nfin 
    &\labelrel\geq{step:binomial} 1 -  \nfin e^{-\frac{\noisebound^2}{2\sigma^2}} \\
 \implies   \probability{\exists n\leq \nfin : |\noisebound_n|>\noisebound} &< \nfin e^{-\frac{\noisebound^2}{2\sigma^2}} \label{eq:prob_lb}
\end{align}
Step~\eqref{step:binomial} follows from the binomial inequality $(1-x)^{n}\geq 1-nx$ for $n\geq 1$ and $x\leq 1$. Let
\begin{align*}
   \delta &=   \nfin e^{-\frac{\noisebound^2}{2\sigma^2}} 
   \implies \log(\frac{\nfin}{\delta}) = \frac{\noisebound^2}{2\sigma^2} 
   \implies \noisebound = \left( 2\sigma^2 \log\left(\frac{\nfin}{\delta}\right) \right)^{1/2} \numberthis
\end{align*}
which together with \eqref{eq:prob_lb} implies that 
\begin{align}
        \probability{|\noise_n| \leq \noisebound, \forall \n \leq \nfin} \geq 1-\delta. \label{eq:noise_small_event}
\end{align}

Using \citet[Theorem 1] {scarlett2017lower} with extension to high probability simple regret explained in \citet[Section D]{scarlett2017lower} using $\nfin$ samples it holds that,
\begin{align}
    \probability{f(x^\star)-f(x_\nfin)\leq \epsilon'} \leq 1 - \frac{1-r}{4-r} = 1-2\delta,\label{eq:regret_small_event}
\end{align}
where $r \in (0,1)$ and 
thus $\delta \in (0,1/8)$.

\noindent Given the bounded noise realizations, we want to bound the probability of the regret being small. Using $A$ as the event described in \eqref{eq:noise_small_event} and $B$ as in  \eqref{eq:regret_small_event}, in particular, we want to bound $\probability{B | A}$. Note that $\probability{B|A} = \probability{A \cap B} / \probability{A} \leq  \probability{B} / \probability{A}$ and therefore
\begin{align*}
   \probability{B | A} &\leq \frac{1-2\delta}{1-\delta} \labelrel{\leq}{step:up_bound} 1-\delta,
\end{align*}
where Step \eqref{step:up_bound} follows by concavity of $\frac{1-2\delta}{1-\delta}$ and being tight at $\delta=0$.
\end{proof}

This lemma shows that given the noise is bounded, achieving simple regret less than $\epsilon'$ requires at least $\nfin$ samples given in \cref{eq:sqexp_lb,eq:matern_lb}. 
The following theorem leverages it to present the sample complexity lower bound for the dynamics exploration problem.


\restatelowerbound*
\begin{proof} 
We reduce our problem to the GP bandit optimization setting described in \cref{coro:scarlett_lb_prob}.

Given a GP bandit optimization problem with any unknown \mbox{$f \in \RKHS$} with \mbox{$\|f\|_{\RKHS} \leq B $}. 
At each round $\n$, we observe $y = f(u(\n)) + \noise_\n$ where noise follows a truncated Gaussian distribution, i.e., $\noise_\n \sim \N(0, \sigma^2)$ 
and conditioned on the event $|\noise_\n| \leq \noisebound$ as per \cref{coro:scarlett_lb_prob}. This is also a conditionally $\sigma$-sub-Gaussian distribution.
The goal is to achieve a simple regret $f(u^\star)-f(u(\nfin))\leq \epsilon$ after rounds $\nfin$ with probability $1-\delta$. This can be converted to a dynamics exploration setup as:

\mypar{An instance} Consider the horizon $\Horizon = 2$, $r(\state) = \state$, the dynamics $\state(\n+1) = f(\coninput(\n)) +  \noise_n$, where the noise is truncated Gaussian distribution 
$|\noise_\n| \leq \noisebound$ 
and $\safeInit{\n} = \X$, where $\X$ is a compact and chosen large enough set such that the state constraints are not active.
For contradiction, suppose that the algorithm guarantees \eqref{eq:lb_objective} in $\nfin < \Omega\left(\frac{\sigma^2}{\epsilon^2} \left ( \log \frac{B}{\epsilon}\right)^{d/2} \right)$ samples. This implies,
\begin{align*}
      \Jobj[]{\state^{\mathrm{p}}, f }{\pessiPolicy} &=  \E \left[ \sum\nolimits_{\h = 0}^{\Horizon-1} r(\state_\h, \coninput_\h) | \state_0 = \state^{\mathrm{p}}, \pessiPolicy \right] \\
     &= \E \left[ r(\state_0, \coninput_0) + r(\state_1, \coninput_1) | \state_0 = \state^{\mathrm{p}}, \pessiPolicy \right] \\
     &= \state^{\mathrm{p}} +  f(\coninput({\nfin})) 
\end{align*}
where $u(\nfin) = \pessiPolicy(\state^{\mathrm{p}})$. Now using \eqref{eq:lb_objective},
\begin{align}
    \state^{\mathrm{p}} +  f(\coninput({\nfin})) &\geq  \E \left[ \max_{x \in \X} x + f(\coninput_{\star}) + \eta \right] - \epsilon' \\
    &\geq \state^{\mathrm{p}} + f(\coninput_{\star})  -\epsilon' \label{eq:regret_reduction}
\end{align}
which implies $f(\coninput^{\star}) - f(\coninput({\nfin})) \leq \epsilon'$ with probability $1-\delta$. 
Therefore, this algorithm would achieve a simple regret $\leq \epsilon'$ for the GP bandit optimization setting in fewer than $\nfin$ samples, which contradicts \cref{coro:scarlett_lb_prob}. 
Hence, with the same probability $1-\delta$, the corresponding lower bound holds for the dynamics exploration problem.
\end{proof}

%% file: figure/dyn_init_set_appending_pessi.tex
 
\tikzset{
pattern size/.store in=\mcSize, 
pattern size = 5pt,
pattern thickness/.store in=\mcThickness, 
pattern thickness = 0.3pt,
pattern radius/.store in=\mcRadius, 
pattern radius = 1pt}
\makeatletter
\pgfutil@ifundefined{pgf@pattern@name@_egw86ddy2}{
\pgfdeclarepatternformonly[\mcThickness,\mcSize]{_egw86ddy2}
{\pgfqpoint{-\mcThickness}{-\mcThickness}}
{\pgfpoint{\mcSize}{\mcSize}}
{\pgfpoint{\mcSize}{\mcSize}}
{
\pgfsetcolor{\tikz@pattern@color}
\pgfsetlinewidth{\mcThickness}
\pgfpathmoveto{\pgfpointorigin}
\pgfpathlineto{\pgfpoint{0}{\mcSize}}
\pgfusepath{stroke}
}}
\makeatother
\tikzset{every picture/.style={line width=0.75pt}} 

\begin{tikzpicture}[x=0.75pt,y=0.75pt,yscale=-1,xscale=1]

\draw  [fill={rgb, 255:red, 128; green, 128; blue, 128 }  ,fill opacity=0.3 ] (74,510.22) .. controls (74,483.78) and (95.43,462.35) .. (121.87,462.35) -- (389.31,462.35) .. controls (415.75,462.35) and (437.18,483.78) .. (437.18,510.22) -- (437.18,653.83) .. controls (437.18,680.27) and (415.75,701.7) .. (389.31,701.7) -- (121.87,701.7) .. controls (95.43,701.7) and (74,680.27) .. (74,653.83) -- cycle ;
\draw  [fill={rgb, 255:red, 108; green, 215; blue, 108 }  ,fill opacity=0.6 ] (154,492) .. controls (214.5,507) and (273.68,485.35) .. (344.68,493.85) .. controls (415.68,502.35) and (411.68,576.85) .. (412.68,649.35) .. controls (413.68,721.85) and (230.18,687.16) .. (208.18,694.66) .. controls (186.18,702.16) and (125.28,669.26) .. (118.18,643.66) .. controls (111.08,618.06) and (93.5,477) .. (154,492) -- cycle ;
\draw  [fill={rgb, 255:red, 108; green, 215; blue, 108 }  ,fill opacity=0.6 ] (166.68,499.26) .. controls (227.18,514.26) and (266.01,498.6) .. (335.68,503.85) .. controls (405.35,509.09) and (395.68,574.35) .. (399.18,641.35) .. controls (402.68,708.35) and (235.68,672.66) .. (209.68,677.66) .. controls (183.68,682.66) and (132.18,659.76) .. (123.68,618.26) .. controls (115.18,576.76) and (106.18,484.26) .. (166.68,499.26) -- cycle ;
\draw  [fill={rgb, 255:red, 80; green, 227; blue, 194 }  ,fill opacity=0.67 ] (154.33,636.57) .. controls (140.06,606.15) and (146.88,569.23) .. (169.56,554.1) .. controls (192.24,538.96) and (222.18,551.35) .. (236.45,581.76) .. controls (250.72,612.17) and (243.9,649.09) .. (221.22,664.23) .. controls (198.54,679.36) and (168.6,666.98) .. (154.33,636.57) -- cycle ;
\draw  [line width=3.75]  (213.18,642.66) .. controls (213.18,641.45) and (214.16,640.48) .. (215.36,640.48) .. controls (216.57,640.48) and (217.54,641.45) .. (217.54,642.66) .. controls (217.54,643.86) and (216.57,644.84) .. (215.36,644.84) .. controls (214.16,644.84) and (213.18,643.86) .. (213.18,642.66) -- cycle ;
\draw    (173.78,619.5) .. controls (175.8,618.26) and (177.43,618.64) .. (178.68,620.63) .. controls (179.89,622.64) and (181.49,623.09) .. (183.48,621.98) .. controls (185.97,621.74) and (187.27,622.8) .. (187.36,625.17) .. controls (187.53,627.51) and (188.84,628.58) .. (191.28,628.38) .. controls (193.53,627.93) and (194.89,628.84) .. (195.38,631.09) .. controls (196.13,633.37) and (197.63,634.12) .. (199.89,633.34) .. controls (202.09,632.5) and (203.59,633.19) .. (204.39,635.4) -- (206.5,636.47) -- (213.31,640.98) ;
\draw [shift={(215.36,642.66)}, rotate = 220.5] [fill={rgb, 255:red, 0; green, 0; blue, 0 }  ][line width=0.08]  [draw opacity=0] (8.93,-4.29) -- (0,0) -- (8.93,4.29) -- cycle    ;
\draw [line width=1.5]    (214.18,641.76) .. controls (257.43,681.76) and (353.68,711.76) .. (369.18,611.26) .. controls (384.37,512.77) and (242.36,499.77) .. (209.38,555.76) ;
\draw [shift={(207.5,559.28)}, rotate = 295.8] [fill={rgb, 255:red, 0; green, 0; blue, 0 }  ][line width=0.08]  [draw opacity=0] (11.61,-5.58) -- (0,0) -- (11.61,5.58) -- cycle    ;
\draw   (200.22,642.66) .. controls (200.22,634.3) and (207,627.52) .. (215.36,627.52) .. controls (223.72,627.52) and (230.5,634.3) .. (230.5,642.66) .. controls (230.5,651.02) and (223.72,657.8) .. (215.36,657.8) .. controls (207,657.8) and (200.22,651.02) .. (200.22,642.66) -- cycle ;
\draw  [line width=3.75]  (205.75,558.76) .. controls (205.75,557.41) and (206.84,556.33) .. (208.18,556.33) .. controls (209.52,556.33) and (210.61,557.41) .. (210.61,558.76) .. controls (210.61,560.1) and (209.52,561.18) .. (208.18,561.18) .. controls (206.84,561.18) and (205.75,560.1) .. (205.75,558.76) -- cycle ;
\draw  [line width=3.75]  (169.75,618.76) .. controls (169.75,617.41) and (170.84,616.33) .. (172.18,616.33) .. controls (173.52,616.33) and (174.61,617.41) .. (174.61,618.76) .. controls (174.61,620.1) and (173.52,621.18) .. (172.18,621.18) .. controls (170.84,621.18) and (169.75,620.1) .. (169.75,618.76) -- cycle ;
\draw [line width=1.5]    (197.68,692.16) -- (197.68,682.66) ;
\draw [shift={(197.68,678.66)}, rotate = 90] [fill={rgb, 255:red, 0; green, 0; blue, 0 }  ][line width=0.08]  [draw opacity=0] (6.97,-3.35) -- (0,0) -- (6.97,3.35) -- cycle    ;
\draw [shift={(197.68,696.16)}, rotate = 270] [fill={rgb, 255:red, 0; green, 0; blue, 0 }  ][line width=0.08]  [draw opacity=0] (6.97,-3.35) -- (0,0) -- (6.97,3.35) -- cycle    ;
\draw  [color={rgb, 255:red, 0; green, 0; blue, 0 }  ,draw opacity=1 ][pattern=_egw86ddy2,pattern size=2.7pt,pattern thickness=0.75pt,pattern radius=0pt, pattern color={rgb, 255:red, 0; green, 0; blue, 0}] (311.68,686.16) .. controls (271.68,692.66) and (218.68,657.16) .. (215.36,642.66) .. controls (212.04,628.15) and (304.18,697.16) .. (337.68,654.16) .. controls (371.18,611.16) and (350.68,566.66) .. (335.18,554.66) .. controls (319.68,542.66) and (288.18,530.66) .. (254.18,542.66) .. controls (220.18,554.66) and (228.18,559.16) .. (218.68,568.66) .. controls (209.18,578.16) and (203.68,557.16) .. (194.68,556.16) .. controls (185.68,555.16) and (202.18,537.66) .. (224.18,523.16) .. controls (246.18,508.66) and (284.18,508.16) .. (320.68,516.16) .. controls (357.18,524.16) and (392.18,571.66) .. (380.18,614.16) .. controls (368.18,656.66) and (351.68,679.66) .. (311.68,686.16) -- cycle ;
\draw  [color={rgb, 255:red, 0; green, 0; blue, 0 }  ,draw opacity=1 ][fill={rgb, 255:red, 74; green, 144; blue, 226 }  ,fill opacity=0.3 ] (312.68,691.66) .. controls (252.18,702.16) and (178.18,635.66) .. (172.18,618.76) .. controls (166.18,601.85) and (307.18,668.16) .. (324.68,638.66) .. controls (342.18,609.16) and (342.68,585.16) .. (326.18,569.16) .. controls (309.68,553.16) and (280.18,549.66) .. (257.68,557.16) .. controls (235.18,564.66) and (234.68,595.66) .. (216.68,586.16) .. controls (198.68,576.66) and (181.18,566.66) .. (172.18,565.66) .. controls (163.18,564.66) and (192.18,524.66) .. (214.18,510.16) .. controls (236.18,495.66) and (285.68,491.16) .. (322.18,499.16) .. controls (358.68,507.16) and (401.18,549.16) .. (401.18,594.16) .. controls (401.18,639.16) and (373.18,681.16) .. (312.68,691.66) -- cycle ;

\draw (156.43,576.53) node [anchor=north west][inner sep=0.75pt]  [font=\large]  {$\mathbb{X}_{n}$};
\draw (81.23,477.87) node [anchor=north west][inner sep=0.75pt]  [font=\Huge]  {$\mathcal{X}$};
\draw (154.63,617.83) node [anchor=north west][inner sep=0.75pt]  [font=\large]  {$x_{s}$};
\draw (215.63,627.83) node [anchor=north west][inner sep=0.75pt]  [font=\large]  {$x'$};
\draw (171.83,640.23) node [anchor=north west][inner sep=0.75pt]  [font=\large]  {$x_{0}^{p,l}$};
\draw (202.33,581.33) node [anchor=north west][inner sep=0.75pt]  [font=\large]  {$x_{H}^{p,l}$};
\draw (178.13,551.33) node [anchor=north west][inner sep=0.75pt]  [font=\large]  {$x_{H}^{p}$};
\draw (270.23,639.63) node [anchor=north west][inner sep=0.75pt]  [font=\large]  {$\pi ^{p}$};
\draw (182.93,604.33) node [anchor=north west][inner sep=0.75pt]  [font=\large]  {$\hat{\pi }$};
\draw (200.52,663.23) node [anchor=north west][inner sep=0.75pt]  [font=\LARGE]  {$\epsilon '$};

\end{tikzpicture}

%% file: appx/experiments.tex
\section{Experiment}
\label{apxsec:experiments}
\vspace{-0.3em}
\looseness -1 
This section presents our experimental setup. We first discuss implementation details, and later explain each environment model, which includes an inverted pendulum, drone navigation, and car racing.


\looseness -1 
\mypar{Implementation details} Across different environments, we employ Gaussian processes or a distribution over a finite number of basis functions as probabilistic models (i.e., Gaussian processes with degenerate kernels) to capture the unknown system dynamics.
To obtain coverage over the dynamics set, we sample multiple dynamics functions from the GP posterior similar to \citet{prajapat2024towards}, as discussed in \cref{rem:implementation}. 
For solving the resulting non-linear finite-horizon optimization problem with sampled dynamics, we utilize the Sequential Quadratic Programming (SQP) algorithm proposed in \citet{prajapat2024towards}.
This optimization procedure yields a unique control input sequence that guarantees safety across all sampled dynamics models. 
According to~\citet{prajapat2025finite}, a finite number of GP samples and an appropriate tolerance yield a reliable over-approximation of the reachable set induced by the dynamics set $\dynSet[\n]$ and thus ensure constraint satisfaction. This requires constraint tightening based on the number of samples. In experiments, we ignore the tightening and directly calibrate the number of samples such that it provides a good approximation of the reachable set and ensures constraint satisfaction. Moreover, we directly optimize the open-loop action sequence instead of affine policies (\cref{rem:implementation}). In the  following, we provide details specific to each environment used in our experiments:

\begin{figure}
    \begin{subfigure}[t]{0.4\columnwidth}
      \centering
  	\scalebox{0.7}{\input{figure/pendulum_setup}} 
    \caption{Inverted pendulum setup}
    \label{fig:pendulum_fig}
    \end{subfigure}
~
    \begin{subfigure}[t]{0.56\columnwidth}
  	\centering
  	\includegraphics[scale=0.6]{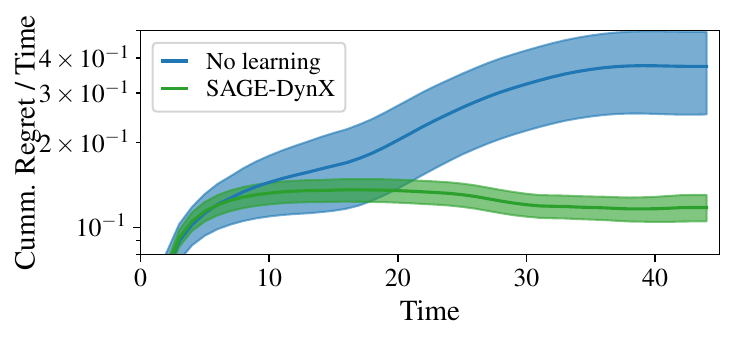}
    \caption{Regret comparison in pendulum}
    \label{fig:pendulum}
    \end{subfigure}
\caption{
\looseness -1
Illustration of the pendulum example. \cref{fig:pendulum_fig}: 
The pendulum begins at the bottom position and aims to reach closer to the green star, while being constrained by the physical red wall and a limit on the angular velocity.
\cref{fig:pendulum} shows the cumulative regret over time (averaged across runs). \dynExplor has significantly lower regret as compared to the no-learning baseline, implying that it more closely follows the optimal trajectory of the clairvoyant agent.
}
\vspace{-1.0em}
\end{figure}

\mypar{Inverted Pendulum} The pendulum with unknown dynamics, where the ground-truth evolution (available only to a clairvoyant agent) is described by the non-linear model:
\begin{align*}
    \begin{bmatrix}
        \theta(\ki+1) \\\\
        \omega(\ki+1)
    \end{bmatrix} =     \begin{bmatrix}
        \theta(\ki) + \omega(\ki) \Delta  + \eta_1(\ki)\\\\
        \omega(\ki) - g_a \sin(\theta(\ki)) \Delta / l + \alpha(\ki) \Delta + \eta_2(\ki)
    \end{bmatrix}. 
\end{align*}
\looseness -1 The state vector is $\state = [\theta, \omega]^\top$, 
where $\theta[\si{rad}]$ is the angular position, $\omega[\si{rad/s}]$ is the angular velocity, and the control input is the angular acceleration $\alpha[\si{rad/s^2}]$.
The dynamics are stochastic, corrupted with uniform noise, $[\eta_1, \eta_2]^\top$ bounded by $\noisebound = 10^{-3}$.
The pendulum has a length of $l = 1\si{m}$, the system is discretized with $\Delta=0.015 \si{s}$, and the acceleration due to gravity is $g_a=9.81 \si{m/s^2}$. 
The GP is trained using
$|\bar{\mathcal{D}}_0| = 27$ prior data points 
on an equally-spaced
$3 \times 3 \times 3$
mesh grid 
in the constraint set 
\mbox{$\mathcal{X}\times\mathcal{U} = \{ (\theta, \omega, \alpha) \in [-2.14,1.14] \times [-2.5,2.5] \times[-8,8] \}$}.
We use a squared exponential kernel with its hyperparameters optimized using the maximum likelihood estimate. 
The task is to control the pendulum from $x = (0,0)$ to a desired
state of $x_{\mathrm{des}}$ marked by a green start in ~\cref{fig:pendulum_fig}. 
The reward function is given by
\mbox{$r(\theta_\h, \alpha_\h) = 50 (x_\h - x_{\mathrm{des}})^2 + 0.1 \alpha_\h^2$}, 
$\forall \h \in \Intrange{0}{\totHorizon-1}$,
where $\totHorizon = 31$. 

To ensure a good coverage of the dynamics set $\dynSet$, we used 50 dynamics functions sampled from the GP. 
We construct a safe set (\cref{assump:safeSet}) around the origin by computing a common Lyapunov function using Jacobians (around the origin) of the dynamics sampled from the set $\dynSet[\n]$, as done in~\citet{prajapat2025finite}.
We recursively solve a pessimistic problem, which finds a unique common control sequence that keeps all the dynamics safe and returns them to the safe set. In line with \dynExplor, instead of executing the entire plan (return back), we replan after collecting every measurement, which are collected after every five time steps.

\looseness -1
\cref{fig:pendulum} compares the performance \dynExplor with a no-learning algorithm, which does not actively move to informative states and does not incorporate online measurements to update the model. 
We plot cumulative regret over time, where regret is computed with the position difference from the clairvoyant agent at any time step. 
The no-learning baseline, due to the high model uncertainty, avoids using high angular velocity (to be safe) and deviates significantly from the optimal behaviour. In contrast, \dynExplor gathers informative data and then progressively aligns more closely to the optimal behaviour.


\begin{figure}
      \centering
  	\includegraphics[scale=0.85]{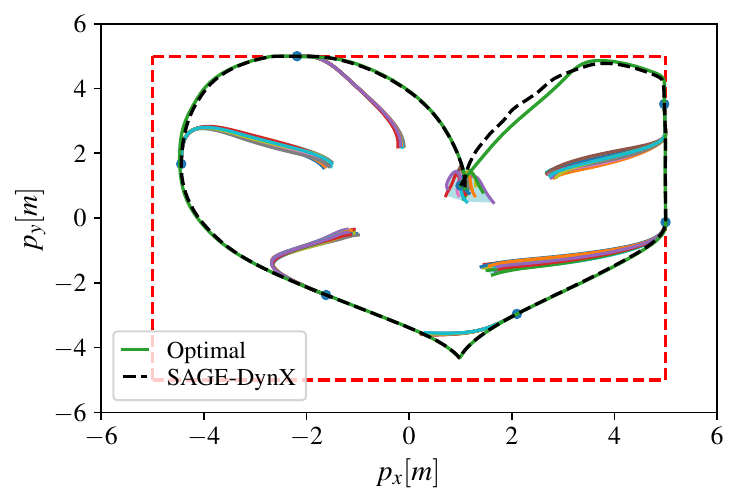}

\caption{
\looseness -1
Demonstration of the \dynExplor algorithm in the drone navigation environment. 
The drone begins at the initial position (1,1) to follow a heart-shaped reference trajectory while satisfying state constraints indicated by dashed red lines. The green line represents the optimal trajectory computed by a clairvoyant agent (known system dynamics).  
The dashed black line shows the agent's trajectory using \dynExplor. 
The thin multicolored lines illustrate predicted trajectories using different sampled dynamics starting from the drone's position at different time steps, highlighted by cyan dots. \dynExplor initially deviates from the optimal behavior, gathers informative data, and quickly converges towards close to optimal performance.
}
\label{fig:drone_setup}
\vspace{-1.0em}
\end{figure}

\mypar{Drone Navigation} We consider a nonlinear drone system described by the following continuous-time dynamics~\citep{singh2023robust,sasfi2023robust}:
\begin{align*}
\begin{bmatrix}
\dot{p}_x \\\\
\dot{p}_y \\\\
\dot{\phi} \\\\
\dot{v}_x \\\\
\dot{v}_y \\\\
\ddot{\phi}
\end{bmatrix}
=
\begin{bmatrix}
v_x \cos(\phi) - v_y \sin(\phi) \\\\
v_x \sin(\phi) + v_y \cos(\phi) \\\\
\dot{\phi} \\\\
v_y \dot{\phi} - g \sin(\phi) \\\\
- v_x \dot{\phi} - g \cos(\phi) \\\\
0
\end{bmatrix}
+
\begin{bmatrix}
0 & 0 \\\\
0 & 0 \\\\
0 & 0 \\\\
0 & 0 \\\\
\frac{1}{m} & \frac{1}{m} \\\\
\frac{l}{J} & -\frac{l}{J} 
\end{bmatrix}
u
+
\begin{bmatrix}
0 \\\\
0 \\\\
0 \\\\
\cos(\phi) \\\\
-\sin(\phi) \\\\
0
\end{bmatrix}
d
\end{align*}
In this model, $p_x$ and $p_y$ denote the drone's horizontal and vertical positions, and $v_x$, $v_y$ represent the corresponding velocities expressed in the body-fixed frame. The variables $\phi$ and $\dot{\phi}$ refer to the drone's orientation angle and its angular velocity. The control vector $\mathbf{u} = [u_1, u_2]^\top$ consists of the thrust forces generated by the two propellers. The constant $d=0.1$ accounts for external wind disturbances. Remaining constants $g= 9.81~[\si{m/s^2}]$, $l= 0.25~[\si{m}]$, $J = 0.00383$ and  $m= 1 ~[\si{kg}]$ denote gravity, the distance from each propeller to the vehicle's center, moment and mass of the drone, respectively. We discretize the continuous-time drone dynamics described above using the Euler discretization method with the time step of $\Delta = 0.1~[\si{s}]$.

We model the dynamics using a finite set of basis functions, incorporating a total of 21 features with unknown parameters. To initialize the model, we use 2 prior data points along each state and input dimension as a mesh grid over the constraint set given by 
$\mathcal{X}\times\mathcal{U}=\{ (x,u) \in [-5,5]\times[-5,5]\times[-\pi,\pi]\times[-2,2]\times[-2,2]\times[-1,1]\times[-1,5]\times[-1,5] \}$.
The drone must learn the dynamics online while staying within these constraints. 
As shown in \cref{fig:drone_setup}, the drone starts at $(p_x,p_y)=(1,1)$
and the task requires tracking of a heart-shaped reference, which is incorporated using a time-varying reward function. The exact equation for generating a heart shape can be found in the submitted code. 
The entire reference is divided into 500 discrete steps, and the agent optimizes over rewards with a receding horizon (rewards based on a moving reference) of length $\totHorizon=31$ that advances one step per action. 
The pessimistic planning uses $15$ samples from the dynamic functions based on the updated model set $\dynSet[n]$. 
The safe set $\mathbb{X}_n$ is an ellipsoid centered around any state with velocity zero, ensuring that the pessimistic plan ends with low velocity.  
We do not observe any constraint violations in the experiment. 
We replan after each measurement, which is collected at every alternate time step. 

\mypar{Car Racing} We model the nonlinear car dynamics using a kinematic bicycle model~\citep{7225830} as follows:
\begin{align*}
    \begin{bmatrix}
        \xpos(\ki+1) \rule{0pt}{1.0em} \\\\ \ypos(\ki+1) \\\\ \theta(\ki+1) \\\\ v(\ki+1) \rule{0pt}{1.0em} 
    \end{bmatrix}  &= \begin{bmatrix}
        \xpos(\ki) + v(\ki) \cos(\theta(\ki) + \zeta_k) \Delta \\\\
        \ypos(\ki) + v(\ki) \sin(\theta(\ki) + \zeta_k) \Delta \\\\
        \theta(\ki) + v(\ki) \sin(\zeta_k) l_r^{-1} \Delta  \\\\
        v(\ki) + a(\ki) \Delta - c v^2(\ki)\Delta
    \end{bmatrix}
\end{align*}
where \mbox{$\zeta_k = \tan^{-1}\left(\frac{l_r}{l_f + l_r} \tan(\delta(k))\right)$} denotes the slip angle~$[\si{rad}]$. The state vector is defined as $\state = [\xpos, \ypos, \theta, v]^\top$, where $[\xpos, \ypos]^\top$ specifies the vehicle's position in a 2D Cartesian frame~$[\si{m}]$, $\theta~[\si{rad}]$ represents the vehicle's heading angle, and $v~[\si{m/s}]$ is the forward velocity. The control vector is given by \mbox{$\coninput = [\delta, a]^\top$}, where $\delta~[\si{rad}]$ is the steering input and $a~[\si{m/s^2}]$ is the applied longitudinal acceleration.
The term $c v^2(\ki)$ denotes the drag force with constant $c = 0.4167$ and the
system is discretized with $\Delta = 0.06~[\si{s}]$. 
The parameters \mbox{$l_f = 1.105~[\si{m}]$} and \mbox{$l_r = 1.738~[\si{m}]$} represent the distances from the vehicle’s center of mass to the front and rear axles, respectively.

We model vehicle dynamics using a finite set of basis functions, incorporating a total of nine features. To initialize the model, we use 64 prior data points, generated by sampling two values along each state and input dimension. These data points are concentrated near the initial state to provide a reliable prior.
The car must learn the dynamics online while staying within track constraints defined by two elliptical boundaries centered at $(x_e,y_e) = (0,0)$. The outer ellipse is given by $({\xpos} - x_e)^2 + 9({\ypos} - y_e)^2 \leq 1600$ and the inner ellipse by $({\xpos} - x_e)^2 + 45({\ypos} - y_e)^2 \geq 400$, forming a corridor the vehicle must remain within as shown in \cref{fig:car_setup}. 
The reward function is time-varying, with the track divided into 500 discrete steps.
The agent optimizes over rewards with a receding horizon of length $\totHorizon=31$ that advances one step per action (receding). The reference follows an elliptical path defined by $({\xpos} - x_e)^2 + 30({\xpos} - y_e)^2 = 900$, which approximately represents the center line.
In addition to the track boundaries, the car is subjected to state constraints and input constraints
$\mathcal{X}\times\mathcal{U}=\{ (x,u) \in [-45,45] \times [-15,15]\times[-40.14,40.14]\times[-15,20]\times[-0.6,0.6]\times[-2,20] \}$.
The pessimistic planning uses $15$ samples from the dynamic functions based on the updated model set $\dynSet[n]$. 
We do not observe any constraint violations in the experiment. 
We replan after each measurement, which is collected at each time step. 

\begin{figure}
      \centering
  	\includegraphics[scale=0.8]{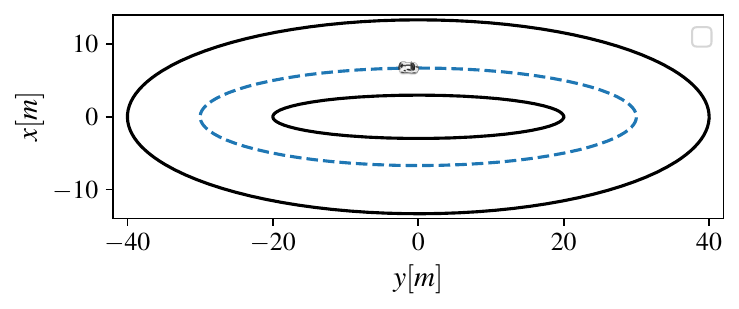}

\caption{
\looseness -1
Illustration of the car track depicted by the black line. The car is required to follow the reference trajectory (blue dashed line) while respecting state and input constraints, as well as staying within the track boundaries.
}
\label{fig:car_setup}
\end{figure}

%% file: figure/pendulum_setup.tex
 
\tikzset{
pattern size/.store in=\mcSize, 
pattern size = 5pt,
pattern thickness/.store in=\mcThickness, 
pattern thickness = 0.3pt,
pattern radius/.store in=\mcRadius, 
pattern radius = 1pt}
\makeatletter
\pgfutil@ifundefined{pgf@pattern@name@_9tzw9tlx2 lines}{
\pgfdeclarepatternformonly[\mcThickness,\mcSize]{_9tzw9tlx2}
{\pgfqpoint{0pt}{0pt}}
{\pgfpoint{\mcSize+\mcThickness}{\mcSize+\mcThickness}}
{\pgfpoint{\mcSize}{\mcSize}}
{\pgfsetcolor{\tikz@pattern@color}
\pgfsetlinewidth{\mcThickness}
\pgfpathmoveto{\pgfpointorigin}
\pgfpathlineto{\pgfpoint{\mcSize}{0}}
\pgfusepath{stroke}}}
\makeatother
\tikzset{every picture/.style={line width=0.75pt}} 

\begin{tikzpicture}[x=0.75pt,y=0.75pt,yscale=-1,xscale=1]

\draw    (407.67,106.33) -- (469.33,106.33) ;
\draw    (438.5,106.33) -- (496,181.67) ;
\draw  [fill={rgb, 255:red, 0; green, 0; blue, 0 }  ,fill opacity=1 ] (476,181.67) .. controls (476,170.62) and (484.95,161.67) .. (496,161.67) .. controls (507.05,161.67) and (516,170.62) .. (516,181.67) .. controls (516,192.71) and (507.05,201.67) .. (496,201.67) .. controls (484.95,201.67) and (476,192.71) .. (476,181.67) -- cycle ;
\draw    (438.83,106.17) -- (438.67,162.33) ;
\draw    (439.33,141.33) .. controls (449.06,146.69) and (450.74,142.08) .. (456.4,135.47) ;
\draw [shift={(458.33,133.33)}, rotate = 133.73] [fill={rgb, 255:red, 0; green, 0; blue, 0 }  ][line width=0.08]  [draw opacity=0] (10.72,-5.15) -- (0,0) -- (10.72,5.15) -- (7.12,0) -- cycle    ;
\draw    (525.33,160.33) -- (525.65,204.33) ;
\draw [shift={(525.67,206.33)}, rotate = 269.58] [color={rgb, 255:red, 0; green, 0; blue, 0 }  ][line width=0.75]    (10.93,-3.29) .. controls (6.95,-1.4) and (3.31,-0.3) .. (0,0) .. controls (3.31,0.3) and (6.95,1.4) .. (10.93,3.29)   ;
\draw    (505.74,119.14) -- (571.38,149.16) ;
\draw  [color={rgb, 255:red, 219; green, 21; blue, 21 }  ,draw opacity=1 ][pattern=_9tzw9tlx2,pattern size=3pt,pattern thickness=0.75pt,pattern radius=0pt, pattern color={rgb, 255:red, 0; green, 0; blue, 0}] (509,112.01) -- (574.64,142.03) -- (571.38,149.16) -- (505.74,119.14) -- cycle ;
\draw  [fill={rgb, 255:red, 126; green, 211; blue, 33 }  ,fill opacity=1 ] (543.33,48.33) -- (545.34,56.42) -- (554.75,55.59) -- (546.57,59.75) -- (550.39,67.33) -- (543.33,61.81) -- (536.28,67.33) -- (540.09,59.75) -- (531.92,55.59) -- (541.33,56.42) -- cycle ;

\draw (440.33,149.33) node [anchor=north west][inner sep=0.75pt]    {$\theta $};
\draw (529.67,195) node [anchor=north west][inner sep=0.75pt]    {$g$};
\draw (478,118.67) node [anchor=north west][inner sep=0.75pt]    {$l$};
\draw (458.67,183.33) node [anchor=north west][inner sep=0.75pt]    {$m$};

\end{tikzpicture}